\documentclass[twocolumn,10pt]{article}
\usepackage[switch]{lineno}
\usepackage{titletoc}

\usepackage[top=1in, bottom=1in, left=.65in, right=.65in]{geometry}
\usepackage[font={scriptsize},labelfont={bf}]{caption}
\usepackage{amsmath}
\usepackage{amsfonts}
\usepackage[utf8]{inputenc}
\usepackage{setspace}
\usepackage{etoolbox}
\usepackage{xcolor}
\usepackage{graphicx}
\usepackage{cprotect}
\usepackage{xspace}
\usepackage{hyperref}
\usepackage[affil-it]{authblk} 
\usepackage[symbol]{footmisc}
\usepackage{bm}

\usepackage{amsthm}
\newtheorem{theorem}{Theorem}[section]

\newtheorem{proposition}[theorem]{Proposition}
\newtheorem{definition}[theorem]{Definition}

\graphicspath{ {./figures/} }
\usepackage{palatino} 
\patchcmd{\maketitle}{\@fnsymbol}{\@alph}{}{}  

\usepackage[numbers,sort&compress]{natbib}

\titlecontents{lsection}[1.8em]
  {\addvspace{3pt}\bfseries}{\contentslabel{2.6em}}{}
  {\titlerule*[0.6pc]{.}\contentspage}
\titlecontents{lsubsection}[4.4em]
  {}{\contentslabel{3.2em}}{}
  {\titlerule*[0.6pc]{.}\contentspage}

\usepackage[title]{appendix}

\title{\huge{\textbf{Uncovering Extreme Event Mechanisms for Prediction and Control with Sensitivity-Balanced Projections
}}}

\author[1,2]{Nicholas Zolman\thanks{Corresponding author (nzolman@uw.edu)}}
\author[1,2]{Sajeda Mokbel}
\author[3]{Samuel E. Otto}
\author[1,2]{Steven L. Brunton}
{\small
\affil[1]{\small Department of Mechanical Engineering, University of Washington, Seattle, WA 98195, USA}
\affil[2]{\small AI Institute in Dynamic Systems, University of Washington, Seattle, WA 98195, USA}
\affil[3]{\small Sibley School of Mechanical and Aerospace Engineering, Cornell University, Ithaca, NY, USA}

}
\date{}

\newcommand{\keywords}[1]
{\centerline{
  \small    
  \textbf{\textit{Keywords: }} #1
}
}

\begin{document}
\newcommand{\cobrasPsi}[0]{\mathbf \Psi}
\newcommand{\cobrasPhi}[0]{\mathbf \Phi}
\newcommand{\cobrasOutput}[0]{\mathbf Q}
\newcommand{\controlInput}[0]{\mathbf u}

\newcommand{\energyDissipation}[0]{\varepsilon}
\newcommand{\navierStokesOp}[0]{\mathcal{NS}}
\newcommand{\kFlowVel}[0]{\mathbf v}
\newcommand{\kFlowVort}[0]{\boldsymbol{\omega}}

\newcommand{\mnls}[0]{A}
\newcommand{\mnlsPotentialGradient}[0]{\mathcal G}
\newcommand{\mnlsGauss}[0]{\rho}
\newcommand{\mnlsSpatial}[0]{\eta}


\twocolumn[
    \maketitle
    \begin{@twocolumnfalse}
        
        \vspace{-30pt}
        \begin{abstract}
        Extreme events---such as earthquakes and coronal mass ejections---are common in many chaotic dynamical systems, yet are difficult to characterize and predict due to the subtle instability mechanisms that drive them.  
        In this work, we develop an interpretable technique that reveals the underlying mechanisms behind extreme events and uses them to build data-driven forecasts and intuitive event suppression controllers.
        In particular, we utilize the covariance balancing reduction using adjoint snapshots (CoBRAS) method to identify linear oblique projections that best capture the sensitivity of a quantity of interest and reconstruct the original state. 
        Importantly, we bypass the need for cumbersome adjoint calculations, instead using backpropagation via modern automatically differentiable numerical frameworks.  
        To accommodate spatially localized events, we also introduce a new variant of CoBRAS to obtain local sensitivity-balanced projections. 
        We demonstrate the utility of this approach to characterize extreme events across a diverse set of challenging systems, including turbulent bursts of energy dissipation in the 2D Kolmogorov Flow, spontaneous synchronization in networks of coupled FitzHugh-Nagumo oscillators, and
        the localized formation of ocean rogue waves from a modified nonlinear Schr{\"o}dinger equation.
        For each example, we show that our simple forecast models accurately predict extreme events and that the underlying mechanisms may be used to design control laws to prevent these events.   
        Finally, we demonstrate that by learning a neural network surrogate model of the dynamics directly from data, we may extend this approach to experimental systems and systems that are not natively written in an automatically differentiable programming language.  
        \end{abstract}
        \vspace{20pt}
    \end{@twocolumnfalse}
]{
  \renewcommand{\thefootnote}%
    {\fnsymbol{footnote}}
  \footnotetext[1]{\thanks{Corresponding author (nzolman@uw.edu)}}
}

\noindent Extreme and rare events arise throughout natural and engineered systems with profound consequences. 
Even simple biological phenomena, such as cells spontaneously synchronizing, can result in epileptic seizures~\cite{lehnertz2006epilepsy} and the first heartbeat during embryonic development~\cite{jia2023bioelectrical}. 
Such events can be found across vastly different scales:  from chemical reaction-diffusion systems~\cite{elezgaray1992a} and gene regulatory networks~\cite{vinoth2025extreme}, to extreme aerodynamic gusts interacting with airfoils~\cite{fukami2025extreme}, to geophysical and solar events that disrupt infrastructure for years~\cite{latif2009a, dijkstra2013a, roberts2016a, ummenhofer2017extreme, cliver2022extreme}. 
It is therefore imperative to predict them with as much lead-time as possible and intervene before they cause lasting harm.

Predicting extreme events is particularly challenging because events are typically rare and arise from high-dimensional systems with multiscale physics~\cite{farazmand2019extreme, sapsis2021statistics}. 
Data-driven approaches have sought to address the rarity by leveraging statistical tools from extreme value theory~\cite{coles2001introduction} and large deviation theory~\cite{dembo2009large, ragone2020computation}, which emphasize the tails of a distribution. Recently, machine learning has been used to provide data-driven forecasts for predicting extreme events~\cite{bobra2016predicting, altwegg2017a, guth2019machine, pickering2022discovering, qi2020using, rudy2023output, chang2025extreme}. Yet accurate forecasting alone does not reveal why events occur or how they might be prevented.

\begin{figure*}[th!]
    \centering
    \includegraphics[width=\textwidth]{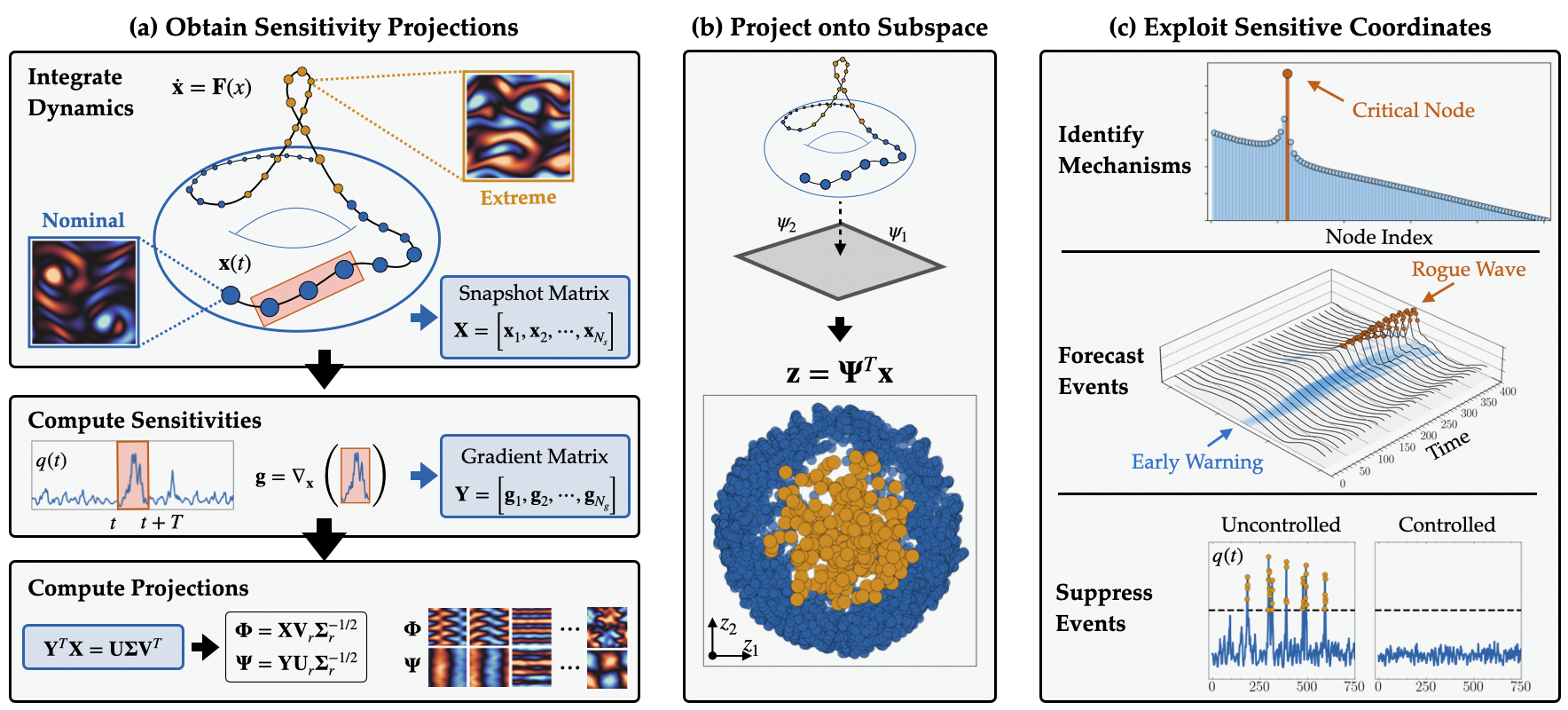}
    \cprotect \caption{\textbf{Schematic of approach.} 
    \textbf{(a)} Obtaining the CoBRAS modes using forward snapshots and gradient samples from the Kolmogorov flow. 
    \textbf{(b)} Projection onto the dominant $\cobrasPsi$ modes of the Kolmogorov flow to reveal sensitive structure. 
    \textbf{(c)} Examples of using the coordinates to identify mechanisms in FitzHugh-Nagumo oscillator networks, predicting event formation of rogue waves in the modified nonlinear Schrödinger equation, and designing controllers to suppress events in the Kolmogorov flow.
    }
     \label{fig:figure1}   
\end{figure*}
 
Extreme events often arise from transient dynamic instabilities~\cite{haller1995a, haller2010a, farazmand2016a, farazmand2016b}, and understanding their mechanisms would not only improve predictive models, but also reveal opportunities for sensing and intervention. Variational approaches have sought states that extremize quantities of interest~\cite{butler1992three, thompson1998initial, farazmand2017a}, but the nonconvex optimization becomes increasingly intractable in high-dimensions, and even a successfully identified optimum is not guaranteed to be a physically realizable trajectory of the system. Perturbation methods can probe sensitivity more directly~\cite{ansmann2013a}, but they become computationally infeasible as the space of possible perturbations grows prohibitively large.

Model reduction methods, such as principal component analysis (PCA) and its equivalent proper orthogonal decomposition (POD), have been instrumental in studying high-dimensional systems~\cite{lorenz1956empirical, berkooz1993proper}; however, these methods struggle to capture extreme events because rare occurrences are underrepresented in the data and their precursors often live at low-energy scales that PCA/POD discards~\cite{farazmand2019extreme}. Nonlinear autoencoder networks have improved the discovery of low-dimensional structure~\cite{fukami2023grasping} but generally sacrifice interpretability and require fine tuning. Sensitivity-aware alternatives, such as balanced POD (BPOD)~\cite{Willcox2002aiaaj, rowley2005model, ilak2008modeling}, resolvent analysis~\cite{mckeon2010critical, luhar2014opposition, herrmann2021data}, and their variants~\cite{benner2018mathcalh_2, benner2024balanced}, are able to capture dynamically relevant structures, but they rely on proximity around invariant sets and struggle when far from them~\cite{otto2022optimizing, otto2023model}. Optimal time-dependent (OTD) modes~\cite{babaee2016minimization, farazmand2016a,
babaee2017reduced, katsidoniotaki2026dynamics} identify the evolution of the most sensitive directions without linearization, but are trajectory-dependent and do not provide a global characterization of the system.

The covariance balancing reduction using adjoint snapshots (CoBRAS)~\cite{otto2023model} addresses many of the aforementioned limitations, generalizing BPOD to nonlinear dynamical systems by optimally balancing information about a system's state and its sensitivity. Taking the singular value decomposition (SVD) of an inner product matrix between state and gradient snapshots, CoBRAS produces a pair of global linear modes that are directly interpretable, capture dynamically relevant structures across all scales, and it has been shown to outperform other balanced truncation methods when the system experiences strong transients~\cite{otto2023model}. Thus, our goal is to develop CoBRAS as a natural framework to characterize and control extreme events.

\smallskip
\noindent\textbf{Contributions.} In this work, we use CoBRAS as a tool for identifying and examining extreme event mechanisms and apply it to three diverse systems: 
(1) turbulent energy bursts in the chaotic 2D Kolmogorov flow, 
(2) spontaneous synchronized firing in large networks of coupled FitzHugh-Nagumo oscillators,
and
(3) localized rogue wave formation in the modified nonlinear Schr{\"o}dinger equation. Specifically, we:

\begin{itemize} 

    \item Reveal the underlying mechanisms of extreme events, identifying interpretable structures such as critical nodes in dynamic networks.
    \item Demonstrate enhanced prediction of extreme events by building simple forecast models directly from CoBRAS coordinates.
    \item Validate the identified mechanisms by developing minimally actuated controllers that completely suppress event formation.
    \item Extend CoBRAS to spatially localized phenomena by applying it to design sliding nonlinear filters.
    \item Develop a fully data-driven pipeline by replacing adjoint calculations with automatically differentiable neural surrogates.
\end{itemize}

Together, these contributions establish CoBRAS as a powerful and interpretable framework for understanding, predicting, and suppressing extreme events in complex dynamical systems.

\section{Background}
\label{sec:background}
Identifying the mechanisms governing extreme events requires projecting high-dimensional dynamics onto a low-dimensional subspace that captures the dynamically relevant phenomena most responsible for rare events, not merely the highest-energy structures. We briefly review the progression of projection methods from proper orthogonal decomposition to CoBRAS, which forms the foundation of our approach.

Throughout this text, we assume that there is some dynamical system of the form:
\begin{equation} \label{eq:dynamical_system}
    \frac{\mathrm{d}{\mathbf x}}{\mathrm{d}t} := \dot{\mathbf{x}} = \mathbf F (\mathbf x, \mathbf u)
\end{equation}
where $\mathbf x \in \mathbb R^m$ is the vector of  state variables, and $\mathbf u$ is a control input that can be used to influence the system. When the system originates from a partial differential equation (PDE), we assume there is an appropriate discretization to reduce to the form in Eq. \ref{eq:dynamical_system}. In each of the methods here, the goal is to obtain a pair of rank-$r$, biorthogonal linear maps $(\cobrasPhi, \cobrasPsi)$  such that
\begin{align} \label{eq:projections}
    \mathbf z = \cobrasPsi^T \mathbf x, \quad
    \mathbf x \approx  \cobrasPhi \mathbf z
\end{align}
where $\mathbf z \in \mathbb R^r$ is a reduced set of coordinates $r \ll m$; i.e. $\cobrasPsi$  \textit{projects into} a low-dimensional space, and $\cobrasPhi$ \textit{reconstructs} the original state. We refer to the bases of columns
$\cobrasPhi = [\phi_1, \phi_2, \dots \phi_r]$ and 
$\cobrasPsi = [\psi_1, \psi_2, \dots \psi_r]$  as ``modes''.

\subsection{POD: Variance-Based Projections}
The proper orthogonal decomposition (POD) \cite{berkooz1993proper}---also known in other communities as principal component analysis (PCA), empirical orthogonal functions (EOF), the Karhunen-Loeve Decomposition, among many others---and its variants \cite{towne2018spectral, schmidt2019conditional, schmidt2026data} have been widely used throughout engineering and science for linear dimensionality reduction. These methods provide a hierarchical basis of modes that capture the most energetic structures in a dynamical system.  For POD, the two orthonormal mode bases are equal, i.e. $\cobrasPhi_{\text{POD}} = \cobrasPsi_{\text{POD}}$, and can be obtained directly from data~\cite{Brunton2022book}. 

Explicitly, given a dataset of $n_x$ samples in $\mathbb R^m$ (such as discretized solutions to a PDE, centered about a mean solution), one builds the data matrix ${\mathbf X = \frac{1}{\sqrt n_x}[\mathbf x_1, \cdots, \mathbf x_{n_x}] \in \mathbb{R}^{m \times n_x}}$. 
When $n_x \ll m$, the ``method of snapshots'' ~\cite{Sirovich:1987} provides an orthonormal projection onto the $r$-dimensional subspace containing the most variance (i.e. energy) by taking the truncated SVD of the $n_x \times n_x$ inner-product matrix $\mathbf X^T \mathbf X = \mathbf U \mathbf \Sigma \mathbf V^T$, where $\mathbf V = \mathbf U$ by symmetry.
The POD modes are given by $\cobrasPsi_{\text{POD}} = \mathbf X \mathbf U_r \mathbf \Sigma_r^{-1/2}$,  ordered hierarchically by their energy content. 
Because extreme events are often driven by transient growth from low-energy states, the dominant POD modes frequently miss the structures most critical to their formation.

\subsection{Balancing Projections}
When studying dynamical systems in an applied control setting, it is critical to understand how one can actually influence the system and what information is possible to sense. 
This information is characterized by the controllability and observability Gramians, $\mathbf W_c$ and $\mathbf W_o$, respectively. For a linear system 
\begin{subequations}
\begin{align}
    \dot{\mathbf x} &= \mathbf{Ax} + \mathbf{Bu}\\ 
    \mathbf y &= \mathbf{Cx}
\end{align} 
\end{subequations}
these are given by
\begin{subequations}
\begin{align}
    \mathbf W_c &= \int_0^\infty e^{\mathbf A t} \mathbf{BB^*}e^{\mathbf A^* t}dt
    \\
    \mathbf W_o &= \int_0^\infty e^{\mathbf A^* t} \mathbf{C^*C}e^{\mathbf A t}dt
\end{align}
\end{subequations}
where $\mathbf{A}^*$ is the adjoint operator given by the conjugate transpose. These are real, positive definite ${m\times m}$ matrices, whose largest eigenvalues correspond to the most controllable and observable directions in the state-space. 

For a system that is controllable and observable, balanced truncation (BT)~\cite{moore1981bt} finds a linear change of coordinates $\mathbf x = \mathbf{Tz}$ such that in the new coordinates $\mathbf z$, the two Gramians are equal and diagonal, \emph{balancing} the importance of controllability and observability. Reducing the system then amounts to truncating modes with low controllability and observability.   

BT quickly becomes computationally intractable for high-dimensional systems, such as fluid flows. Balanced POD (BPOD)~\cite{rowley2005model} addresses this by approximating the most controllable modes from the same snapshot matrix as in POD, $\mathbf X$, and the most observable modes from an \emph{adjoint} snapshot matrix $\mathbf Y = \frac{1}{\sqrt{n_g}}[\mathbf g_1, \cdots, \mathbf g_{n_g}]$ sampled from impulse responses of the linear adjoint system $\dot{\mathbf g} = \mathbf A^* \mathbf g + \mathbf C^*\mathbf v$. 
By taking the $r$-truncated SVD of the balanced inner-product matrix  $\mathbf Y^T\mathbf X = \mathbf U \mathbf \Sigma \mathbf V^T  \in \mathbb{R}^{n_g \times n_x}$, BPOD provides the pair of linear maps:
\begin{subequations}
\begin{align}
    \cobrasPhi_{\text{BPOD}} 
        &= \mathbf{X} \mathbf V_r \mathbf{\Sigma}_r^{-1/2}
    \in \mathbb{R}^{m \times r}
    \\
    \cobrasPsi_{\text{BPOD}} 
        &= \mathbf{Y} \mathbf U_r \mathbf{\Sigma}_r^{-1/2}
    \in \mathbb{R}^{m \times r}
\end{align}
\end{subequations}

Like with POD, the SVD provides a natural hierarchy of BPOD modes; however, instead of identifying the dominant energetic structures, the first $r$ BPOD modes provide a set of coordinates that characterizes the $r$-dimensional subspace that is most controllable and observable. However, a fundamental limitation of BPOD is that it relies on the dynamics being linear and that forward and adjoint data produced through impulse responses. While nonlinear generalizations have been developed~\cite{scherpen1993balancing, fujimoto2008computation, benner2018mathcalh_2, benner2024balanced}, these methods can become infeasible in high-dimensions and their performance suffers from strong nonlinearities when far from the equilibrium where they were constructed~\cite{otto2022optimizing, otto2023model}.

\subsection{Sensitivity Projections} \label{background:CoBRAS}
The covariance balancing reduction using adjoint snapshots (CoBRAS) method~\cite{otto2023model} was developed to generalize BPOD to build reduced-order models of general nonlinear dynamical systems by balancing information about the state and sensitivity of the system. 
CoBRAS also replaces adjoints obtained via linear impulse responses by sampling gradients of an arbitrary function $\cobrasOutput$ defined along trajectories of the system. 
As established by~\cite{constantine2015active, zahm2020gradient}, the covariance matrix of gradient samples, $\mathbf W_g$, can reveal a low-dimensional ``active subspace'' that captures the most sensitive directions that impact $\cobrasOutput$. Importantly, this space is quickly spanned by gradient samples that depend on the effective \textit{rank} of the derivative map, and not the dimensionality of the input.  Thus, in practice very few samples are needed to capture the space. Otto et al.~\cite{otto2023model} observed that the state and gradient covariance matrices transform identically to the controllability and observability Gramians, enabling the BPOD construction to carry over directly. 

Explicitly, let $\cobrasOutput(\mathbf x, {\mathbf u}) \in \mathbb R^n$ map a trajectory, defined by an initial state $\mathbf x$ and a sequence of control inputs ${\mathbf u}$, to a set of $n$ outputs. By sampling a random normal vector  $\boldsymbol{\xi} \in \mathbb{R}^n$, one can obtain gradient samples with respect to the initial state $\mathbf g = \nabla_\mathbf{x} (\boldsymbol{\xi}^T \cobrasOutput)(\mathbf{x}, {\mathbf u})$. Analogous to BPOD, CoBRAS builds a matrix of $n_g$ gradients (instead of adjoint linear responses)  $\mathbf Y = \frac{1}{\sqrt n_g}[\mathbf g_1, \cdots \mathbf g_{n_g}] \in \mathbb{R}^{m \times n_g}$ and takes the $r$-truncated SVD of the balanced inner-product matrix: $\mathbf Y^T\mathbf X = \mathbf U \mathbf \Sigma \mathbf V^T  \in \mathbb{R}^{n_g \times n_x}$, where $\mathbf X$ is the matrix of forward snapshots as before. Likewise, we obtain the pair of linear maps:
\begin{align} \label{eq:cobras_modes}
    \cobrasPhi_{\text{CoBRAS}} 
        &= \mathbf{X} \mathbf V_r \mathbf{\Sigma}_r^{-1/2}
    \in \mathbb{R}^{m \times r}
    \\  
    \cobrasPsi_{\text{CoBRAS}} 
        &= \mathbf{Y} \mathbf U_r \mathbf{\Sigma}_r^{-1/2}
    \in \mathbb{R}^{m \times r}.
\end{align}
For the rest of the text, we will simply write $\cobrasPhi, \cobrasPsi$ to denote these CoBRAS modes. 

As in POD and BPOD, the CoBRAS modes are ordered hierarchically. The first $r$ modes of $\cobrasPsi$ project onto the most sensitive $r$-dimensional subspace of $\cobrasOutput$, and the first $r$ modes of $\cobrasPhi$ optimally reconstruct the states from this subspace. Importantly, these maps are not tied to a particular equilibrium---they are defined \textit{globally}---and capture the most dynamically relevant structures for describing the map $\cobrasOutput$. As shown in~\cite{otto2023model}, these projections provide better reduced-order models than BPOD and its nonlinear variants when transients drive the state far away from the equilibrium, because the CoBRAS projections treat both nominal and transient phenomena on equal footing. This makes CoBRAS a natural tool for studying low-dimensional chaotic attractors where instabilities force trajectories to leave the attractor, such as systems with extreme events.

\section{CoBRAS for Extreme Events}
\label{sec:methods_cobras}
Extreme events arise in high-dimensional, complex dynamical systems where transient growth is amplified along sensitive directions by nonlinear dynamics. CoBRAS reveals a low-dimensional subspace that captures the global sensitivity of measurements evolving along trajectories. In this work, we use CoBRAS as a natural and interpretable framework for identifying the mechanisms behind these extreme events. 

We consider a simplified setting where data is collected along trajectories without control, so Equation \ref{eq:dynamical_system}  becomes $\mathbf{\dot x} = \mathbf F(\mathbf{x})$ and the map $\cobrasOutput$ introduced in Section \ref{background:CoBRAS} is defined by an initial state $\mathbf x$ without the sequence of control inputs. 
In particular, we define $\cobrasOutput(\mathbf x)$ to be the forward evolution of some quantity of interest (QoI) $q(\mathbf x(t))$. 
Letting 
$q_j = q(\mathbf x(t + j\Delta t))$ and $T = n\Delta t$, 
we define the map
$\mathbf{Q}(\mathbf x(t)) = [q_0, \dots  , q_{n-1}] \in \mathbb R ^{n}$, 
which quantifies the extremity of events over a time horizon $[t, t+T)$.

The learned CoBRAS modes $\cobrasPsi$ provide a coordinate system that captures the sensitivity of the future evolution of $q$ with respect to the current state; when large changes in $q$ signal an extreme event, these coordinates characterize the formation mechanism, providing a natural basis for prediction and suppression. A schematic of our method can be found in Figure \ref{fig:figure1}.

For simplicity, we regard $q$ as a scalar, but this can be readily generalized by stacking different quantities. We let $q_*$ denote the threshold for an extreme event.
The matrix of gradients, $\mathbf Y$, is prepared by using the randomized procedure described in Section \ref{background:CoBRAS} and is conveniently obtained via backpropagation using autodifferentiation with differentiable simulators built in JAX~\cite{jax2018github}.

\subsection{Localized Events}
\label{sec:methods_local_cobras}
CoBRAS provides a set of \textit{global} projections of a state variable, characterizing the most sensitive directions. However, many extreme events in nature, such as hurricane formation or coronal mass ejections, are spatially \emph{localized}---the presence of a single event does not impact the entire domain. In these settings, extreme event formation is dominated by the dynamics within a local neighborhood. To accommodate this wider class of phenomena, we adapt CoBRAS for localized QoIs and events by extracting spatially-dependent feature vectors $\mathbf z(\mnlsSpatial)$ and using them to construct sliding nonlinear filters for prediction and control.  

More precisely, suppose our dynamical system has spatial dependence, such as the solution to a PDE, $\mathbf{x}(t, \mnlsSpatial)$ for $\mnlsSpatial \in \Omega$. 
Let $q(\mathbf x; \mnlsSpatial_0)$ be a localized quantity of interest at $\mnlsSpatial_0$, such as the total energy in a neighborhood $D$ of $\mnlsSpatial_0$, $\int_D |\mathbf{x}(t, \mnlsSpatial)|^2 d \mnlsSpatial$, and likewise define $\cobrasOutput(\mathbf{x}; \mnlsSpatial_0)$. 
For translation-equivariant systems, we can translate signals into a common frame of reference, such as the origin, and take gradients in that frame; i.e. $\mathbf g = \nabla_{\mathbf{\tilde x}} \cobrasOutput(\mathbf {\tilde x}; 0)$, where $\mathbf{\tilde x}(t, \mnlsSpatial) = \mathbf {x}(t, \mnlsSpatial - \mnlsSpatial_0)$. 
If $q(\mathbf{x}, \mnlsSpatial_0)$ depends on a local neighborhood
over some time horizon, $T$,
then $\mathbf g$ will be compactly supported on $\Omega$ for each $\mnlsSpatial_0$. Evaluating CoBRAS in this frame of reference provides a set of spatially-dependent modes,
$\cobrasPhi(\mnlsSpatial)$ and $\cobrasPsi(\mnlsSpatial)$, that optimally project and reconstruct signals $\mathbf{x}(\mnlsSpatial)$ for all $\mnlsSpatial$ in a neighborhood of the origin. 

Translation symmetry allows us to move new signals back into this frame of reference.
In particular, $\cobrasPsi$ defines the kernel of a sliding linear filter:
\begin{align}
    \label{eq:local_corr}
    \mathbf z(\mnlsSpatial) = (\cobrasPsi \star \mathbf x)(\mnlsSpatial)
\end{align}
where $\star$ is the cross-correlation operator on real-valued functions. This provides a CoBRAS transformation that characterizes the sensitivity of $\cobrasOutput(\mathbf x; \mnlsSpatial)$ for each point in the domain, and global reconstruction of the signal is
\begin{align}
    \label{eq:global_recon}
    \mathbf{x}(\mnlsSpatial) \approx \cobrasPhi(0) \cdot \mathbf z(\mnlsSpatial).
\end{align}
Note that cross-correlation is extremely efficient in modern signal processing and deep learning frameworks, and the reconstruction only requires $\cobrasPhi$ to be evaluated at a \textit{single} point. We derive these expressions in \ref{si:localized} 
and extend the result to the more general setting of systems with Lie group symmetry.

\subsection{Predicting Events}
\label{sec:methods_preds}
The coordinates $\mathbf z = \cobrasPsi^T \mathbf{x}$  capture the most sensitive directions of $\cobrasOutput$, providing natural features for predicting extreme events.  For simplicity, we train a support vector machine (SVM) with a radial-basis function (RBF) kernel. 

Given a QoI $q(t)$ and coordinate $\mathbf z(t)$, we seek to determine whether an event occurs over some time horizon, $t_{pred}$. The SVM decision function $d(\mathbf z(t))$ predicts:
\begin{equation}
\text{sign} \left(   d(\mathbf z) \right)= \text{sign}\left\{\max_{\tau \in [t, t+t_{pred}]} q(\tau) -q_*   \right\}.
\end{equation}
Note that the QoI $q$ in this definition can be different from the one defining $\cobrasOutput$. For localized events, we simply replace $\mathbf z$ with $\mathbf {z}(\mnlsSpatial)$ from Eq. \ref{eq:local_corr} and provide pointwise predictions across the domain; the predictions then define a \textit{nonlinear filter} for the system,
$d(\mnlsSpatial) = d\left[(\cobrasPsi \star \mathbf x)(\mnlsSpatial)\right]$.

\subsection{Suppressing Events}
\label{sec:methods_suppression}
We verify a causal relationship between the leading CoBRAS modes and extreme events by demonstrating that feedback control laws based on these modes are capable of completely suppressing extreme events.
The CoBRAS $\cobrasPsi$ modes indicate the most sensitive directions for influencing the map $\cobrasOutput$; by designing controllers in this subspace, we have greater influence over how $q$ evolves with the dynamics.
We verify this by designing intuitive, minimally actuated controllers to suppress extreme events.
For the control dynamical system, $\mathbf{\dot x} = \mathbf{F}(\mathbf x) + \mathbf u$, we provide feedback control 
\begin{equation}
    \mathbf u(\mathbf{x}) = \cobrasPhi \cdot  \mathbf{u}_z \left(\cobrasPsi^T \mathbf x \right) = \cobrasPhi \cdot  \mathbf u_z(\mathbf z)
\end{equation}
where $\mathbf u_z(\mathbf z)$ is a control law designed in the projected space, and $\cobrasPhi \mathbf u _z$ lifts back into the original space. Analogously using Eq. \ref{eq:global_recon} for localized events, the control law becomes 
\begin{equation}
    \mathbf u(\mathbf x(\mnlsSpatial)) 
    = \cobrasPhi(0) \cdot \mathbf u_z\left(\mathbf z(\mnlsSpatial)\right).
\end{equation}

As an example, we can compute $\mathbf u_z$ by using the decision of a classifier similar to previous works~\cite{suykens1999chaos}. If there is a classifier for predicting whether an extreme event will occur over some time horizon, such as the kernel SVM previously discussed, then  the negative gradient of the decision function $-\nabla_z d$ points towards the nominal region of space. For simple kernels, like RBFs, the gradient can be computed in closed form. We can, for example, choose
\begin{equation}
    \mathbf u_z = -k_{gain}\nabla_z d /||\nabla_z d||
\end{equation}
to be our controller, where $k_{gain}$ is a small positive constant. To ensure sparse control, we can even multiply $\mathbf 1 \left\{d(\mathbf z) > 0 \right\}$ so that the system is only actuated when the classifier predicts an extreme event is likely. 

Restricting actuation to the
$r$-dimensional projected space provides a direct test of whether CoBRAS has identified the true mechanism; suppressing the formation of an event with minimal control confirms that the essential dynamics are captured.

\begin{figure*}[tp]
    \centering
    \includegraphics[width=0.95\textwidth]{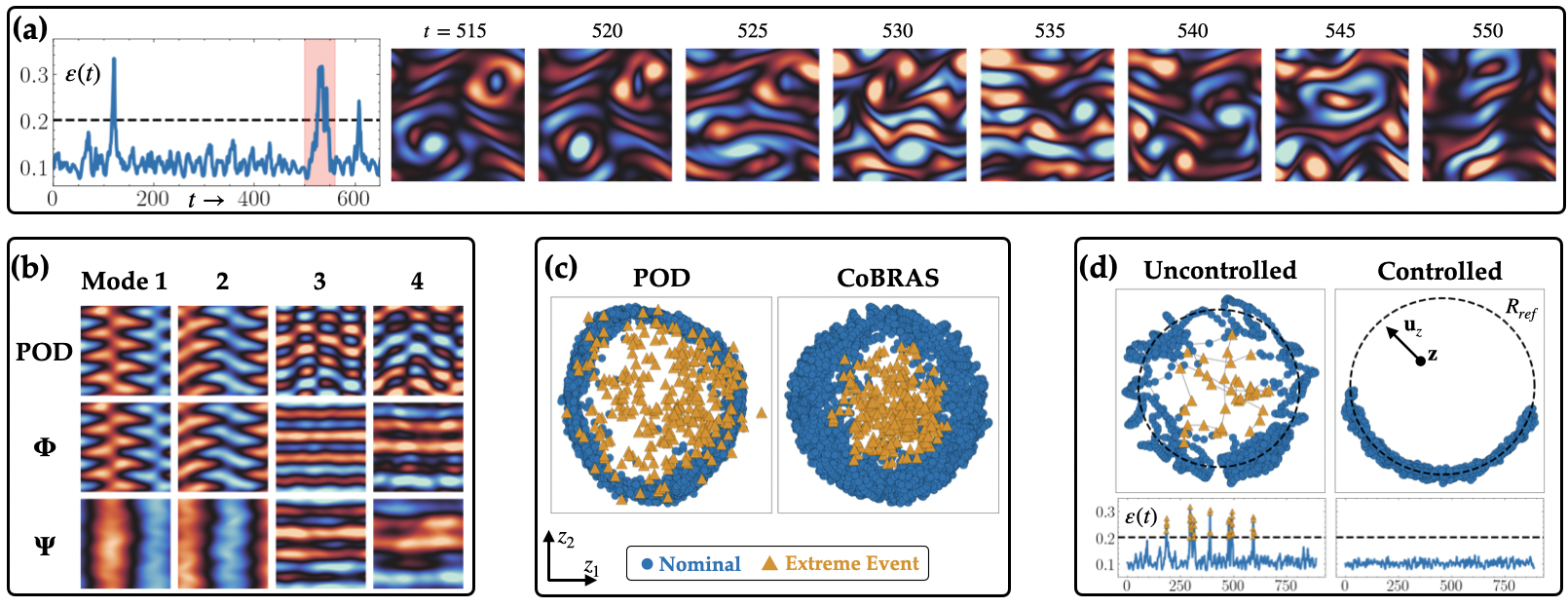}
    \cprotect \caption{\textbf{Kolmogorov Flow} \textbf{(a)} Energy dissipation, $\varepsilon(t)$, of the Kolmogorov flow from $t \in [0, 650]$ and vorticity snapshots taken from the shaded region during an extreme event $t \in [515,550]$. \textbf{(b)} The first four POD and CoBRAS $(\cobrasPhi, \cobrasPsi)$ modes. \textbf{(c)} The projection onto the first two POD and CoBRAS modes; points are colored by whether the energy dissipation is above the energy threshold. \textbf{(d)} Comparison between the uncontrolled flow and the controlled flow where the system is steered to the circle with radius $R_{ref}$ in the CoBRAS projected space, completely suppressing extreme event formation.
    }
     \label{fig:kflow}   
\end{figure*}

\section{Results} 
\label{results}

We now examine the ability to use CoBRAS for predicting, suppressing, and characterizing extreme events on three diverse systems: turbulent energy dissipation bursting in the 2D Kolmogorov Flow, spontaneous synchronization of node activations in a FitzHugh-Nagumo oscillation network, and localized rogue wave formation in the 1D modified nonlinear Schrödinger equation. Finally, we show that gradients can be obtained from a learned surrogate model, enabling CoBRAS in fully data-driven settings where no differentiable simulator is available.

\subsection{Turbulent Energy Bursts}
The Kolmogorov flow is a 2D fluid flow governed by the incompressible Navier-Stokes equations with periodic boundary conditions in both directions driven by a spatially periodic forcing term. 
Its simplicity has made it a widely studied toy model for turbulence in analytic~\cite{platt1991investigation}, simulation~\cite{borue1996numerical,qin2025clean}, and even laboratory settings~\cite{obukhov1983kolmogorov}. In particular, the flow exhibits intermittent extreme events as energy transfers from larger to smaller scales as is typical of turbulent systems.
Explicitly, the Kolmogorov flow takes the form
$$\partial_t \kFlowVel = - \kFlowVel \cdot \nabla \kFlowVel - \nabla p + \nu \Delta \kFlowVel + \mathbf f 
,\qquad \nabla \cdot \kFlowVel = 0$$ 
where $\kFlowVel(t, x,y)$ is the velocity field and $p(t, x, y)$ is the pressure field, $\nu = 1/Re$ is the kinematic viscosity, $Re$ is the Reynolds number, $(x,y) \in [0,2\pi]^2$,  and $\mathbf f(x,y) = (\sin(k_f y),0)$ is the 
characteristic forcing. For simplicity, we analyze the vorticity of the system $\omega = \nabla \times \kFlowVel$ and treat the discretized field as a high-dimensional ODE: 
$d\kFlowVort/dt = \mathbf F(\kFlowVort)$, where $\mathbf F$ is the discretized nonlinear differential operator in vorticity space. 

In this work, we consider the well-studied case $k_f = 4$ and $Re=40$ where the flow is fully chaotic and exhibits intermittent bursts in the energy dissipation $\energyDissipation(\kFlowVort) = \frac{\nu}{4\pi^2} \int |\kFlowVort(t, x,y)|^2 dxdy$. 
When applying CoBRAS to this system, we define the state to be the discretized vorticity field, i.e. $\mathbf x = \kFlowVort$, our QoI to be the energy dissipation  
 $q(t) = \energyDissipation(\kFlowVort(t))$, and consider a time-horizon of $T=4$ to compute gradients. Full details about the simulation and setup can be found in \ref{si:kflow}. We compare against POD to highlight the advantages of sensitivity-balanced projections.

\paragraph{Revealing the Bursting Mechanism.}
The CoBRAS and POD modes capture qualitatively different structures in the Kolmogorov flow, as shown in Figure \ref{fig:kflow}(b).
The first four CoBRAS modes account for over 90\% of the SVD energy (see \ref{si:kflow_40_modes}), capturing the most sensitive $4$-dimensional subspace for describing the energy dissipation. These modes resemble $\sin(x), \cos(x), -\sin(4y)$, and $-\cos(4y)$ respectively.
This is in direct agreement with the finding in~\cite{farazmand2017a} where the authors formulate a constrained variational optimization problem to identify an initial condition that maximizes the change in energy input in the system---a leading indicator of extreme events---and determined that extreme events are correlated with triadic interactions between the first and fourth Fourier modes, $a_{1,0}$ and $a_{0,4}$. 
In fact, the first mode is the solution to the variational problem in~\cite{farazmand2017a} at $Re=40$, up to a scale factor and phase shift. 

The first two POD and CoBRAS $\cobrasPhi$ modes look identical, as they capture \textit{variance} of the snapshot data. 
However, the $\cobrasPsi$ modes define the projection into the embedding space, and these are fundamentally different from the POD modes.
The leading $\cobrasPsi$ modes define the most sensitive directions for $\energyDissipation$, so changes along these modes have the greatest impact on $\energyDissipation$. 

This distinction is most notable when we project the snapshot data into the leading two coordinates of the POD and CoBRAS $\cobrasPsi$ modes in Figure \ref{fig:kflow}(c). Both plots reveal the ring-like structure defining the body of the chaotic attractor where the majority of the nominal, low-energy snapshots live. However, the projection onto the POD modes scatters the high-energy snapshots throughout the domain, causing significant overlap with the nominal states, so there is no clear separation boundary. 
The projection onto $\cobrasPsi$, however, defines a clear and interpretable separation: rare, high-energy events are excursions into the interior of the ring, while nominal states live on the exterior of the ring. 

This radial dependence directly corresponds to the 
magnitude of the first Fourier mode, $|a_{1,0}|$, and has been the defining quantity for predicting and controlling extreme events~\cite{farazmand2017a, farazmand2019closed}. 
CoBRAS recovers the same result without any prior knowledge of Fourier modes being a natural basis due to the periodicity of the domain; instead {CoBRAS automatically identifies the most natural basis and ranks the modes based on their impact on the QoI}.

In \ref{si:kflow}, we provide additional CoBRAS modes and demonstrate that symmetry reduction can clarify the dynamical structures encoded in higher-order modes.

\paragraph{Predicting Bursts}
Following Section \ref{sec:methods_preds}, we train RBF-kernel SVMs to evaluate the predictive capability of CoBRAS against POD over a time horizon $t_{pred}$; the full study is provided in \ref{si:kflow} (including $Re$ dependence, symmetry reduction, and kernel-CoBRAS variants). Using just two CoBRAS modes, the true event rate remains high until $t_{pred} \approx 8$ and plateaus around $t_{pred} \approx 16$; matching this performance requires at least eight POD modes. Despite both being linear projections, the first two {CoBRAS modes are significantly more informative about extreme events than the leading POD modes}.

\paragraph{Suppressing Energy Bursts}
Finally, to verify that the first two CoBRAS coordinates actually describe the \textit{mechanism} governing the extreme event, we design a simple controller in the projected space. 
The clean geometric separation in Figure \ref{fig:kflow}(c) reveals a more direct approach than the SVM-based control design outlined in Section \ref{sec:methods_suppression}. 
Given a snapshot, $\kFlowVort$, we simply project into the two-dimensional space $\mathbf z = \cobrasPsi^T \kFlowVort$ and steer the system towards the closest point on the circle defined by $\mathbf z_{ref} = R_{ref} \mathbf z /||\mathbf z||$, where $ R_{ref}$ is the median radial coordinate of the nominal states. We 
use a proportional controller $\controlInput_z = -(\mathbf z - \mathbf z_{ref})$ in the latent space and
actuate the system in $\kFlowVort$-space by using the $\cobrasPhi$ modes: $\controlInput_\omega = \cobrasPhi \controlInput_z$.

In Figure \ref{fig:kflow}(d), we show the controlled and uncontrolled trajectories for a sample initial condition that leads to multiple extreme events. In the controlled setting, the latent state quickly collapses to the circle defined by $R_{ref}$.  Thus, extreme event formation is completely suppressed, confirming that the dominant mechanism is captured in just two CoBRAS coordinates.

\begin{figure*}[t!]
    \centering
    \includegraphics[width=0.95\textwidth]{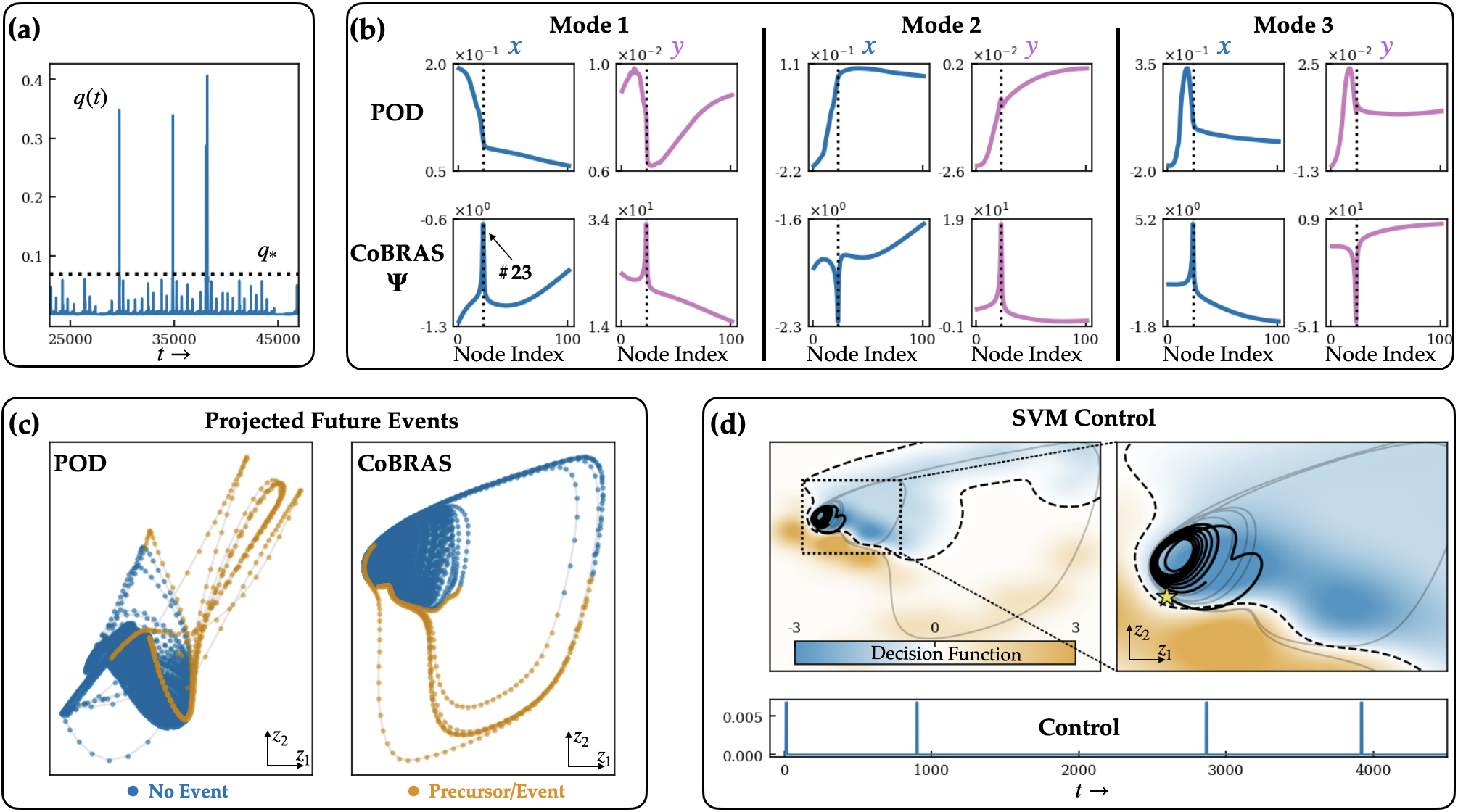}
    \cprotect \caption{\textbf{FitzHugh-Nagumo Oscillators} 
    \textbf{(a)} The QoI for the FHN system over time exhibits small bursting with intermittent large extreme events. Dotted line corresponds to event threshold, $q_*$.
    \textbf{(b)} The first three POD and $\cobrasPsi$ modes separated into the corresponding action on the $x$ and $y$ components. Dashed vertical line indicates node \#23. 
    \textbf{(c)} The projection onto the leading two POD and CoBRAS modes with states in orange indicating that an event will happen within $t_{pred}= 50$ units.  
    \textbf{(d)} \textit{Top}: Heatmap of the SVM decision function classifying whether an event will occur within $t_{pred} = 50$ time units and controlled trajectory preventing the event.  The black dashed line indicates the SVM decision boundary. The light gray line is the reference uncontrolled trajectory, and the solid black line is the controlled trajectory. The star in the zoomed in panel indicates the initial condition for both trajectories. \textit{Bottom}: The magnitude of the sparse control policy $||\mathbf u||$ as a function of time. 
    }
     \label{fig:fhn}   
\end{figure*}

\subsection{Synchronization in Oscillator Networks}
 The FitzHugh-Nagumo (FHN) model has been widely used to model complex networked dynamical systems exhibiting node excitation, such as electro-optical systems~\cite{romeira2016regenerative}, networks of neurons firing~\cite{gerster2020fitzhugh}, cardiac~\cite{nash2004electromechanical} and pancreatic cells~\cite{scialla2021hubs}, and biochemical networks~\cite{lin2004resonance}. A detailed review of the model can be found in~\cite{cebrian2024six}. 

As in~\cite{ansmann2013a, karnatak2014a}, we consider coupled FHN networks of $N$ units $(x_i, y_i)$ governed by
\begingroup
\addtolength{\jot}{-10pt}
\begin{align} \nonumber 
\dot x_i & = x_i (a_i - x_i)(x_i - 1) - y_i + k \sum_j A_{ij}(x_j - x_i)\\  
\dot y_i & = b_i x_i - c y_i
\end{align}
\endgroup
for $i = 1,2, \dots,  N$, where $A_{ij}$ is the adjacency matrix, $a_i, b_i$ and $c_i$ are internal parameters for each unit, and $k = 0.128/(N-1)$ is the coupling strength. For simplicity, we choose $a_i=-0.02651$ and $c_i = 0.02$ to be constants for all nodes and $b_i$ to be distributed in $[0.006, 0.014]$, which breaks the symmetry among all nodes. 
Following~\cite{ansmann2013a, karnatak2014a}, we analyze a fully connected network ($A_{ij} = 1$ for all $i\neq j$) of $N=101$ nodes and linear ramp up in coefficients $b_i = 0.006 + 0.008\frac{i-1}{N-1}$. With these parameters, the system is chaotic and exhibits intermittent synchrony leading to extreme bursts of neuronal activity.  
Additional investigations of the network with $N=1001$ nodes and a small world network consisting of $N=100\times 100$ nodes can be found in \ref{si:fhn}. 

To apply CoBRAS, we take $\mathbf x = (x_1, \dots x_N, y_1, \dots y_N)$ to be our state variable and define our QoI to be the node-average energy of the system $q(t) = \frac{1}{2N} \sum_i \left( x_i(t)^2 + y_i(t)^2 \right)$. The system is chaotic~\cite{ansmann2013a, karnatak2014a}, and this quantity frequently exhibits small bursts, $q \approx 0.06$, when only some nodes synchronize and large bursts, $q \approx 0.3$, when they all synchronize. We take $q_*$ to be four standard deviations beyond the mean, $q_* \approx 0.07$. For computing the gradient samples, we consider a time horizon of $T=100$ time units.
 
\paragraph{Identifying Critical Nodes}
Unlike the Kolmogorov Flow, the fully-connected FHN system does not admit a natural geometric basis, such as Fourier modes, which characterize the system. Instead, the natural modes of the system are defined purely by the dynamics. 
In Figure \ref{fig:fhn}(b), we plot the first three POD modes and CoBRAS $\cobrasPsi$ modes obtained from our method. 
For ease of visualization, we separate the components that act on the $x$- and $y$-variables and plot them against their corresponding node index. 
An extended set of modes, including the corresponding $\cobrasPhi$ modes, can be found in \ref{si:fhn}. 

POD  only captures the variance of the state variables and biases towards low-index nodes since they are active more often. More importantly, however, the $y$-variables are an order of magnitude smaller than the $x$-variables and therefore have less variance; this is reflected in the POD modes which further suppress them. As shown in Figure \ref{fig:fhn}(c), the projection onto the dominant POD modes completely loses the geometric information about the dynamics, whereas CoBRAS retains a smooth geometric structure that clearly indicates the outline of the chaotic attractor, with extreme events following trajectories that escape it.

Across all $\cobrasPsi$ modes (in both $x$- and $y$-variables), an important feature emerges: there is a sharp cusp occurring at around node \#23---indicating that the system is critically sensitive to values of this node. This is consistent with Ansmann et al. who identified a sharp transition in extreme-event likelihood when 23 nodes are active~\cite{ansmann2013a}. Because the heterogeneous coupling excites low-$b_i$ units first, this onset coincides with the activation of node \#23---CoBRAS independently identifies this node as the critical one driving event formation.

Furthermore, CoBRAS identifies that the \textit{relative phase} between the $x$- and $y$-components of node \#23 is important, as shown by the peaks pointing in opposite directions in Figure \ref{fig:fhn}(b). In \ref{si:fhn_101}, we examine this further: when $x_{23}$ and $y_{23}$ are in phase (have the same sign), this facilitates growth in the dynamics, leading to synchrony and extreme events; when the two are out-of-phase (opposite signs), the system de-synchronizes and trajectories return to the body of the attractor. In \ref{si:fhn_1001} and \ref{si:fhn_small_world}, we also verify these results are consistent in larger networks.

\paragraph{Predicting and Preventing Synchrony}
We confirm the mechanism governing the events by actuating the system to completely prevent events from forming. However, the decision boundary for the FHN system has a more complex geometry than the simple circle used for the Kolmogorov flow. Instead, we fit a kernel SVM to find the decision boundary for whether an event will form over the next $t_{pred} = 50$ time units as described in Methods. 

Figure \ref{fig:fhn}(d) shows the SVM decision function when fit using the first two CoBRAS modes. The prediction regions capture the location of the extreme events shown in Figure \ref{fig:fhn}(c) very well---achieving a true nominal accuracy of 96\% and true event accuracy of 87\% on a held-out test set.

\begin{figure*}[tp]
    \centering
    \includegraphics[width=0.95\textwidth]{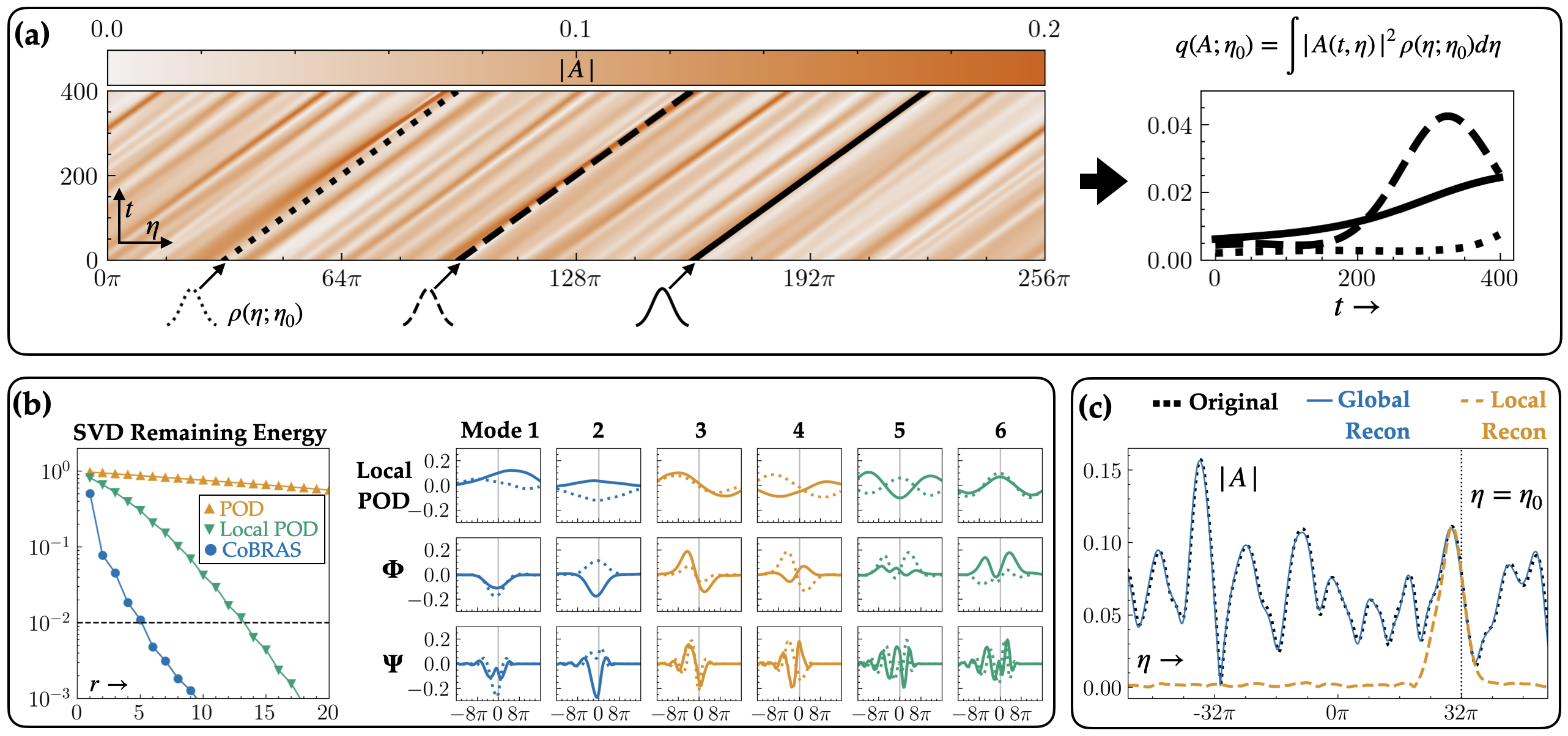}
    \cprotect \caption{\textbf{Modified Nonlinear Schr{\"o}dinger Equation Overview} \textbf{(a)} \textit{Left}: Amplitude, $|\mnls|$, of a solution the MNLS equation from $t \in [0,400]$. Black lines indicate center of Gaussian $\mnlsGauss(\mnlsSpatial; \mnlsSpatial_0)$ 
    when calculating the QoI 
    $q(\mnls; \mnlsSpatial_0)$
    at three different example values of $\mnlsSpatial_0$. \textit{Right}: The corresponding localized QoIs
    $q(\mnls; \mnlsSpatial_0)$.
    \textbf{(b)} \textit{Left}: The remaining energy fraction for the  POD, localized POD, and CoBRAS singular values; dashed line indicates a 1\% threshold. \textit{Middle}: The first six normalized modes obtained by localized POD and CoBRAS. Solid and dotted lines denote the real and imaginary components, respectively. Colors correspond to pairs of modes that differ approximately by a factor of $i$. \textbf{(c)} Example of local and global reconstruction of a snapshot using 8 modes. Vertical dashed line corresponds to the center of the local reconstruction. 
    }
     \label{fig:mnls_1}   
\end{figure*}

Finally, Figure \ref{fig:fhn}(d) shows the effect of the SVM-based control law. The trajectory starts very close to the boundary, and without control will escape the attractor through the ``channel'' structure previously characterized in the literature~\cite{ansmann2013a, karnatak2014a}. 
However, as soon as the trajectory crosses the decision boundary, the controller nudges the state slightly back into the nominal region where it remains until the next visit to the boundary. This demonstrates that {CoBRAS provides a natural framework for identifying critical mechanisms in networked dynamical systems and designing targeted interventions to mitigate extreme events.}

\subsection{CoBRAS for Localized Phenomena}

While extreme events can affect the entire state of a dynamical system, many important extreme phenomena are localized, such as hurricanes in our atmosphere or coronal mass ejections from the sun. 
The globally defined projections from CoBRAS we have discussed so far will struggle to capture localized formation. We address this limitation by developing a local version of CoBRAS as outlined in Section \ref{sec:methods_local_cobras}, and we demonstrate it on the modified nonlinear Schr{\"o}dinger (MNLS) equation.

The MNLS equation is a complex-valued, 1D wave equation and has been used to model the formation of extreme rogue waves, both in the ocean~\cite{dysthe2008a} and optical systems~\cite{solli2007optical}.
Following the treatment by Cousins \& Sapsis in~\cite{cousins2016a}, we examine the form of MNLS for modeling the wave envelope of unidirectional deep ocean rogue waves given by
\begin{align}
    \partial_{t}\mnls
    &+ \frac{1}{2}\partial_\mnlsSpatial \mnls
    + \frac{i}{8}\partial^2_{\mnlsSpatial}\mnls 
    - \frac{1}{16}\partial^3_{\mnlsSpatial}\mnls
    \\ \nonumber
    &+ \frac{i}{2}|\mnls|^2 \mnls 
    + \frac{3}{2} |\mnls|^2 \partial_\mnlsSpatial\mnls
    + \frac{1}{4} \mnls^2 \partial_\mnlsSpatial\mnls^*
    + i \mnls \mnlsPotentialGradient(\mnls) 
    =0
\end{align}
where $\mnls(t,\mnlsSpatial)$ is a complex wave envelope, ${\mnlsSpatial \in [0, 256\pi]}$ is a spatial variable, 
$\mnlsPotentialGradient$ is the derivative of the velocity potential with Fourier transform 
${\mathcal F \left(\mnlsPotentialGradient(\mnls)\right)
    = -|k| \mathcal F \left( |\mnls|^2 \right)}$, and $k$ is the wavenumber in the Fourier domain. 

For the unmodified nonlinear Schr{\"o}dinger (NLS) equation, the amplitude growth or decay is governed by a simple relationship based on a critical scale parameter~\cite{cousins2015a, cousins2016a}. However, the additional cubic nonlinearities in MNLS give rise to a  more complicated growth where wave amplitudes can reach more than eight standard deviations above the mean amplitude~\cite{onorato2005modulational,dysthe2008a}.

We are interested in the localized energy content of a frame moving in the flow with the average group velocity, $c$. Explicitly, given an initial frame centered at $\mnlsSpatial_0$, we define the localized energy at time $t$ to be:
\begin{equation}\label{mnlsQOI}
    q(\mnls;\mnlsSpatial_0) = \int |\mnls(t,\mnlsSpatial)|^2 \mnlsGauss(\mnlsSpatial; \mnlsSpatial_0) d\mnlsSpatial
\end{equation}
where $\mnlsGauss(\mnlsSpatial; \mnlsSpatial_0)$
is a normalized Gaussian centered at $\mnlsSpatial_0 - ct$ with width $L=\pi/2$. 

In localized CoBRAS, the state is $\mathbf x(\mnlsSpatial) = [\text{Re } \mnls(\mnlsSpatial), \text{Im } \mnls(\mnlsSpatial)]$, and the localized quantity of interest is $q(\mnls;\mnlsSpatial_0)$ as in Eq.~\ref{mnlsQOI}. The map $\cobrasOutput(\mnls; \mnlsSpatial_0)$ tracks this energy in the moving frame for $T=200$ time units. 
A full description of our simulation and treatment of MNLS can be found in  \ref{si:mnls}. 

\paragraph{Localized Formation Mechanism}
As shown in Figure \ref{fig:mnls_1}(b), the CoBRAS singular values decay rapidly, with only six modes needed to account for 99\% of the SVD energy, whereas the POD singular values decay very slowly. 
This is because the CoBRAS modes are localized projections and are only approximately supported on the region, whereas the POD modes capture the global variance across the entire domain. However, even localized POD---where snapshots are sub-sampled using 1/8th of the domain---only approximates Fourier modes, and needs more than double the number of modes to account for over 99\% of the SVD energy. Figure \ref{fig:mnls_1}(c) demonstrates the effectiveness of the local CoBRAS projections at reconstructing a neighborhood of a point and the near-perfect global reconstruction via the correlation formula from Equation \ref{eq:global_recon}. 

Because the quantity $q$ only depends on a small localized region near $\mnlsSpatial_0$, the CoBRAS projections are also only approximately supported on a local region. Interestingly, the support of the CoBRAS modes is much larger than the Gaussian, $\mnlsGauss$, used to define $q$, indicating that {our method has learned a natural length scale for MNLS events} and seems to be robust across $\mnlsGauss$ widths (see \ref{si:mnls_gauss} for additional experiments). Fitting a Gaussian to  $|\phi_1(\mnlsSpatial)|$ provides $L_G \approx 10$, which is consistent with the length scale with the highest likelihood of triggering events previously found by Cousins \& Sapsis~\cite{cousins2016a}.

Treating the complex amplitude as a 2D real-valued map, we obtain two pairs of functions per projection; in fact, as indicated by the color in Figure \ref{fig:mnls_1}, pairs of consecutive modes differ by a factor of  $\pm i$, e.g. if $\psi_k = a +ib$, then $\psi_{k+1} = -b  + ia$. 
The CoBRAS modes do not have symmetry about the origin; this effect is pronounced in Figure \ref{fig:mnls_1}(c), where local reconstruction is centered at $\mnlsSpatial = 32\pi$, but clearly supported for a much larger region to the left.
Since the QoI is defined in the frame of reference of the average group velocity, this indicates that $\cobrasOutput$ is more sensitive to incoming information from the left---which can be seen visually in Figure \ref{fig:mnls_1}(a) where a fast-traveling wave collides with the center black dashed line. 
The asymmetry is more evidence that {our method identifies causal information} that other local projections, such as Gabor transforms, would miss.

\begin{figure}[tp]
    \centering
    \includegraphics[width=0.47\textwidth]{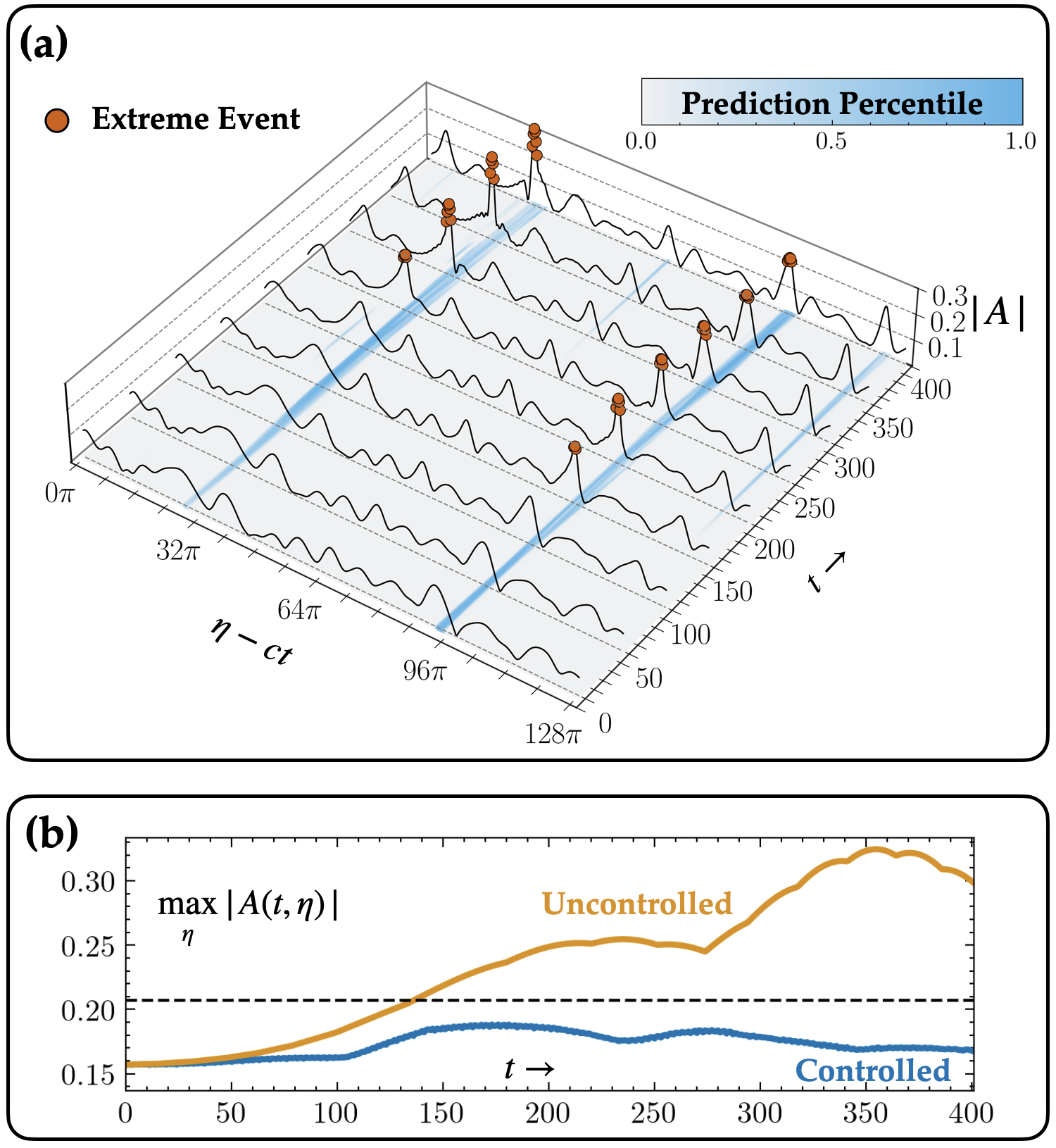}
    \cprotect \caption{\textbf{Modified Nonlinear Schr{\"o}dinger Equation Results} \textbf{(a)} Extreme event predictions obtained from SVM classifier to predict whether an extreme event will occur within 200 time units. The blue streaks indicate the percentile rank of the SVM decision function among positive classifications, estimated from the training set. Red circles indicate the amplitude exceeding the extreme event threshold. 
    \textbf{(b)} The effect of applying feedback control using the first eight CoBRAS modes to suppress extreme wave growth. Dashed line corresponds to the extreme event threshold.
    }
     \label{fig:mnls_2}   
\end{figure}

\paragraph{Predicting Rogue Wave Formation}
We devise a simple strategy to predict rogue wave formation by using the linear CoBRAS projections to train a classifier as described in Section \ref{sec:methods_preds}.
We transform the signal into CoBRAS space $\mathbf{z}(\mnlsSpatial)$ using Equation \ref{eq:local_corr} and predict whether an extreme event will occur within $T=200$ time units in the moving frame $\mnlsSpatial - ct$. 
Training details can be found in \ref{si:mnls}.
Once trained, the SVM predictions act as a \textit{nonlinear filter} on the input signal $\mnls(t, \mnlsSpatial)$, identifying likely regions where an event will form. 
Figure \ref{fig:mnls_2}(a) highlights this predictive capability; the filter clearly identifies that extreme events will form hundreds of time steps in advance. 
The SVM classifier has a true negative rate of $97.7\%$ and a true positive rate of $92.7\%$ on a held-out test set.

\paragraph{Preventing Rogue Waves}
We design a simple controller using the method in Section \ref{sec:methods_suppression} by again transforming the system into CoBRAS space, $\mathbf{z}(\mnlsSpatial)$ using Eq.~\ref{eq:local_corr}, and 
using the SVM as a nonlinear filter to develop a sparse controller. 
To prevent having too much influence on the system, we fix $k_{\text{gain}} = 10^{-2}$. Figure \ref{fig:mnls_2}(b) shows the difference between the maximum amplitude over time for the uncontrolled and controlled system with the same initial condition. 
While the sparse control strategy allows some initial growth, it completely prevents the formation of extreme values. 
Importantly, the actuation is sparse in the spatial domain, using a maximum 3\% of the domain.
This demonstrates that {localized CoBRAS provides an efficient framework for predicting and suppressing spatially localized extreme events.}

\subsection{Data-Driven CoBRAS}
A key drawback of CoBRAS is that it relies on the formation of the gradient sample matrix, $\mathbf{Y} = \frac{1}{\sqrt n_g}[\mathbf g_1 , \mathbf g_2 \dots,  \mathbf{g}_{n_g}]$. While obtaining these gradients is possible with adjoint-based or differentiable simulations, most simulation environments do not provide this capability. Furthermore, it is common to obtain a physical dataset from observations or experiments without a simulation to accompany it. Recent advances in operator learning~\cite{boulle2024mathematical, boulle2024operator} have begun to establish that it is possible to approximate adjoint information using structured models in high- and infinite-dimensional spaces purely from forward solutions. We therefore adopt this strategy and train a surrogate differentiable solver directly from data, using its gradients in place of true adjoint information to obtain the CoBRAS modes.

We focus our attention on the Kolmogorov flow with the same setup as previously presented, but where we only have access to snapshot data. While more sophisticated approaches can be taken, we simply train a Fourier neural operator (FNO)~\cite{li2021fourier} to predict the vorticity snapshot one time unit into the future $\tilde{  \kFlowVort}_{n+1} = \text{FNO}(\kFlowVort_n)$ by minimizing the mean-squared error loss of the prediction over a trajectory,  $\mathcal{L} = 1/n_x\sum_k || \tilde \kFlowVort_k - \kFlowVort_k||^2$. 

Despite only using approximations of the gradient,  we find that the first two CoBRAS modes are nearly identical to when we have true gradient information, as shown in Figure \ref{fig:fno}. 
Just as with true gradient information, the FNO CoBRAS modes also capture the importance of the first and fourth Fourier modes, the dominant energy transfer mechanism for creating extreme events. A more detailed comparison can be found in \ref{si:fno}.
This demonstrates that with an appropriate surrogate model for gradients, {CoBRAS can identify true physical mechanisms directly from data---without any knowledge of the underlying equations, access to a simulator, or explicit adjoint computation}.

\begin{figure}[tp]
    \centering
    \includegraphics[width=0.47\textwidth]{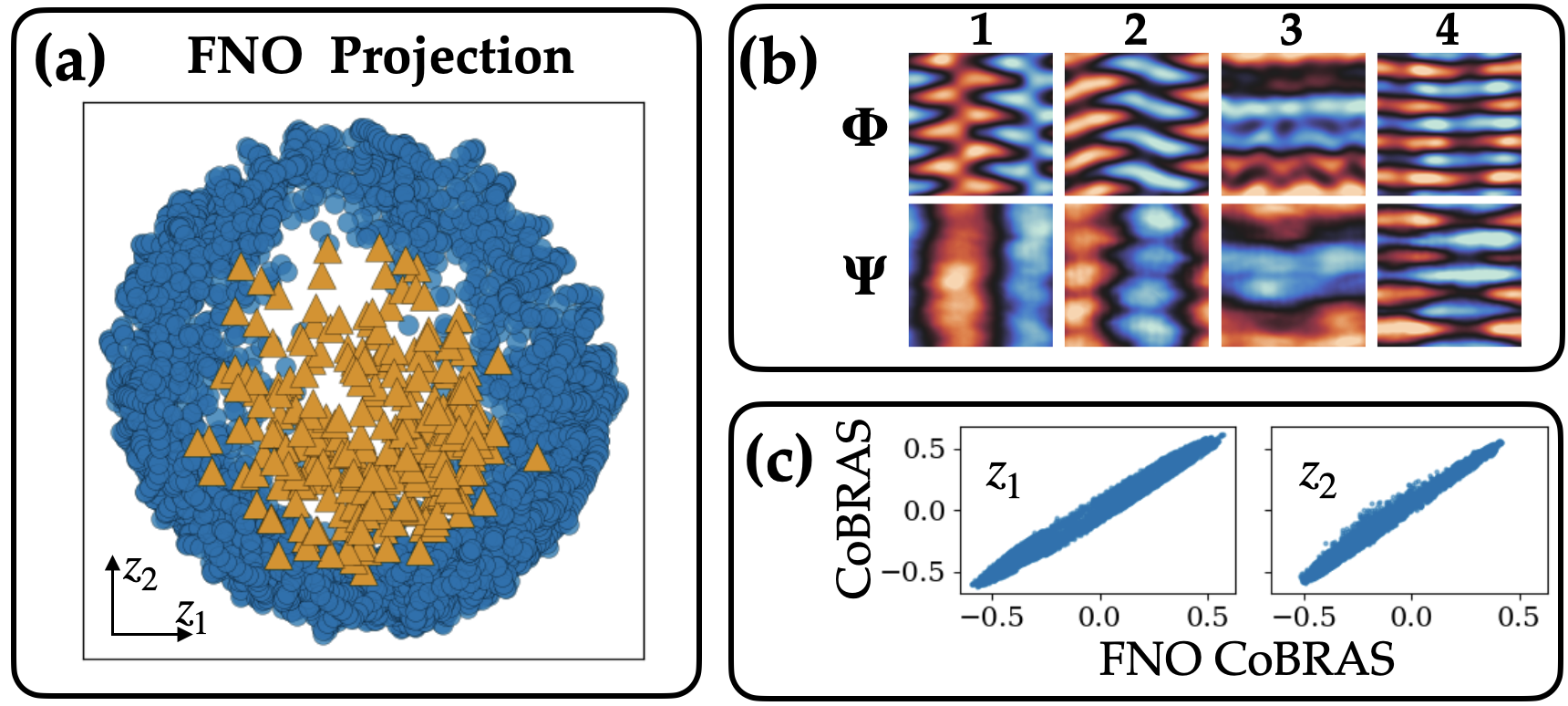}
    \cprotect \caption{\textbf{Non-intrusive CoBRAS} \textbf{(a)} Projection onto the first two dominant CoBRAS $\cobrasPsi$ modes by obtaining gradients using autodiff through a learned FNO surrogate model.
    \textbf{(b)} The first four CoBRAS modes obtained  from the FNO. 
    \textbf{(c)} Comparison of the first two projected coordinates, $\mathbf z = \cobrasPsi^T \kFlowVort$, between CoBRAS modes obtained from a differentiable simulator vs the learned FNO surrogate.
    }
     \label{fig:fno}   
\end{figure}

\section{Discussion}
\label{discussion}
In this work, we introduced a novel method for characterizing and controlling extreme events in dynamical systems by using sensitivity-balanced projections from CoBRAS. Extreme events are sensitive to the geometry of the underlying dynamics; by projecting onto the globally dominant sensitivity spaces, the CoBRAS modes provide a natural, and often interpretable, coordinate system for understanding the geometry and the mechanism governing the extreme event formation. 

We have demonstrated this approach on three very different systems: turbulent bursting in the Kolmogorov flow governed by the Navier-Stokes equations, spontaneous synchronization in FitzHugh-Nagumo oscillator networks, and the localized formation of rogue waves in the ocean modeled by a modified nonlinear Schr{\"o}dinger equation. In each case, the identified modes provided an interpretable decomposition of the extreme event formation, confirming previous findings in the literature, and providing deeper insight into the underlying mechanisms. 
We confirmed these mechanisms by developing simple controllers that completely suppress the formation of an event by minimally actuating the system.

In particular, we demonstrated how applying CoBRAS to dynamics on a network identified the critical role a node plays in event formation. 
This provides evidence that CoBRAS could help identify other critical phenomena in networks, such as identifying stations susceptible to failure in power grids or identifying neuron populations in the brain responsible for particular tasks.  This further provides insight into the appropriate way to influence these systems with control.

By adapting CoBRAS to localized quantities of interest, we developed a mathematical framework for analyzing the sensitivity of localized phenomena and coherent structures. 
The localized modes serve as optimal linear filters, transforming a signal into the most sensitive components by cross-correlating with the CoBRAS modes. 
We have shown that this filter can also be used to predict the formation of localized extreme events by training a simple classifier which acts as a nonlinear filter on the signal and provides advanced warning of event formation. 

There remain rich opportunities for extending the CoBRAS framework developed here. Throughout this work, we have focused on linear representations, as the linear modes directly provide interpretable structures that reveal the underlying mechanisms, such as the energy transfer between modes, critical nodes in networks, and characteristic length scales driving amplitude growth. Kernel CoBRAS~\cite{otto2023model} offers a natural path to better predictive horizons (see \ref{si:background}), though currently at the cost of interpretability. Reconciling predictive power with mechanistic insight remains an open and important direction. The localized framework similarly extends naturally beyond 1D translation-invariant systems (see \ref{si:localized}); since the localized projections take the form of a cross-correlation operator, this provides a natural foundation for extensions to graph convolutions on arbitrary topologies---and with it, the ability to localize extreme event mechanisms in complex networks.

Finally, a practical challenge of CoBRAS is its reliance on a differentiable or adjoint-based simulator, limiting its applicability in settings where only sensor data is available. Motivated by recent advances in operator learning theory~\cite{boulle2024mathematical, boulle2024operator}, we address this by learning a differentiable surrogate and using its gradients to obtain CoBRAS modes, demonstrating that the dominant mechanisms can still be recovered. Improving surrogate gradient quality remains an important open direction. This data-driven pipeline provides a natural path toward extracting physical insight from existing ML surrogates that outperform physical models, such as global weather forecasts, transforming black-box predictions into physical understanding.

\section*{Materials}
All data generation and experiments were performed on a single engineering workstation with twenty-four Intel Core i9-7920X CPUs and a single NVIDIA GeForce RTX 2080 Ti GPU. AI-assisted coding tools were used to help adapt the numerical scheme from \cite{kassam2005fourth} for MNLS and the FNO training to be compatible with JAX. 

\subsection*{Data and Code Availability}
Data and code for all systems are made publicly available and can be found in the open-source repository: \url{https://github.com/nzolman/cobras-extreme}.

\bibliographystyle{unsrt}
 \begin{spacing}{.9}
 \small{
 \setlength{\bibsep}{4.8pt}
 \bibliography{refs}

@article{otto2023model,
  title={Model reduction for nonlinear systems by balanced truncation of state and gradient covariance},
  author={Otto, Samuel E and Padovan, Alberto and Rowley, Clarence W},
  journal={SIAM Journal on Scientific Computing},
  volume={45},
  number={5},
  pages={A2325--A2355},
  year={2023},
  publisher={SIAM}
}

@article{farazmand2019extreme,
  title={Extreme events: Mechanisms and prediction},
  author={Farazmand, Mohammad and Sapsis, Themistoklis P},
  journal={Applied Mechanics Reviews},
  volume={71},
  number={5},
  pages={050801},
  year={2019},
  publisher={American Society of Mechanical Engineers}
}

@article{dysthe2008a,
  author = {Dysthe, K. and Krogstad, H.E. and M{\"u}ller, P.},
  year ={2008},
  title = {Oceanic Rogue Waves},
  volume = {40},
  pages = {287–310},
  language = {lb},
  journal = {Annu. Rev. Fluid Mech},
  number = {1}
}

@article{altwegg2017a,
  author = {Altwegg, R. and Visser, V. and Bailey, L.D. and Erni, B.},
  year ={2017},
  title = {Learning From Single Extreme Events},
  volume = {372},
  pages = {20160141},
  language = {en},
  journal = {Philos. Trans. R. Soc. B},
  number = {1723}
}

@article{ansmann2013a,
  author = {Ansmann, G. and Karnatak, R. and Lehnertz, K. and Feudel, U.},
  year ={2013},
  title = {Extreme Events in Excitable Systems and Mechanisms of Their Generation},
  volume = {88},
  pages = {052911},
  language = {en},
  journal = {Phys. Rev. E},
  number = {5}
}

@article{karnatak2014a,
  author = {Karnatak, R. and Ansmann, G. and Feudel, U. and Lehnertz, K.},
  year ={2014},
  title = {Route to Extreme Events in Excitable Systems},
  volume = {90},
  pages = {022917},
  language = {en},
  journal = {Phys. Rev. E},
  number = {2}
}

@article{latif2009a,
  author = {Latif, M. and Keenlyside, N.S.},
  year = {2009},
  title = {El Ni~no/Southern Oscillation Response to Global Warming},
  volume = {106},
  pages = {20578–20583},
  language = {es},
  journal = {Proc. Natl. Acad. Sci},
  number = {49}
}

@book{dijkstra2013a,
  author = {Dijkstra, H.A.},
  year = {2013},
  title = {Nonlinear Climate Dynamics},
  publisher = {Cambridge University Press},
  language = {hu},
  address = {Cambridge, UK}
}

@article{roberts2016a,
  author = {Roberts, A. and Guckenheimer, J. and Widiasih, E. and Timmermann, A. and Jones, C.K.R.T.},
  year = {2016},
  title = {Mixed-Mode Oscillations of El Ni~no–Southern Oscillation},
  volume = {73},
  pages = {1755–1766},
  language = {fil},
  journal = {J. Atmos. Sci},
  number = {4}
}

@article{haller2010a,
  author = {Haller, G. and Sapsis, T.},
  year ={2010},
  title = {Localized Instability and Attraction Along Invariant Manifolds},
  volume = {9},
  pages = {611–633},
  language = {en},
  journal = {SIAM J. Appl. Dyn. Syst},
  number = {2}
}

@article{elezgaray1992a,
  author = {Elezgaray, J. and Arneodo, A.},
  year = {1992},
  title = {Crisis-Induced Intermittent Bursting in Reaction-Diffusion Chemical Systems},
  volume = {68},
  pages = {714},
  language = {en},
  journal = {Phys. Rev. Lett},
  number = {5}
}

@article{farazmand2016a,
  author = {Farazmand, M. and Sapsis, T.P.},
  year ={2016},
  title = {Dynamical Indicators for the Prediction of Bursting Phenomena in High-Dimensional Systems},
  volume = {94},
  pages = {032212},
  language = {en},
  journal = {Phys. Rev. E},
  number = {3–1}
}

@article{haller1995a,
  author = {Haller, G. and Wiggins, S.},
  year ={1995},
  title = {Multi-Pulse Jumping Orbits and Homoclinic Trees in a Modal Truncation of the Damped-Forced Nonlinear Schr{\"o}dinger Equation},
  volume = {85},
  pages = {311–347},
  language = {en},
  journal = {Phys. D: Nonlinear Phenom},
  number = {3}
}

@article{farazmand2016b,
  author = {Farazmand, M.},
  year ={2016},
  title = {An Adjoint-Based Approach for Finding Invariant Solutions of Navier–Stokes Equations},
  volume = {795},
  pages = {278–312},
  language = {en},
  journal = {J. Fluid Mech}
}

@article{farazmand2017a,
  author = {Farazmand, M. and Sapsis, T.P.},
  year ={2017},
  title = {A Variational Approach to Probing Extreme Events in Turbulent Dynamical Systems},
  volume = {3},
  pages = {1701533},
  language = {en},
  journal = {Sci. Adv},
  number = {9}
}

@article{platt1991a,
  author = {Platt, N. and Sirovich, L. and Fitzmaurice, N.},
  year ={1991},
  title = {An Investigation of Chaotic Kolmogorov Flows},
  volume = {3},
  pages = {681–696},
  language = {en},
  journal = {Phys. Fluids A},
  number = {4}
}

@article{cousins2015a,
  author = {Cousins, W. and Sapsis, T.P.},
  year ={2015},
  title = {Unsteady Evolution of Localized Unidirectional Deep-Water Wave Groups},
  volume = {91},
  pages = {063204},
  language = {en},
  journal = {Phys. Rev. E},
  number = {6}
}

@article{cousins2016a,
  author = {Cousins, W. and Sapsis, T.P.},
  year ={2016},
  title = {Reduced-Order Precursors of Rare Events in Unidirectional NonlinearWaterWaves},
  volume = {790},
  pages = {368–388},
  language = {en},
  journal = {J. Fluid Mech},
  number = {3}
}

@article{qi2020using,
  title={Using machine learning to predict extreme events in complex systems},
  author={Qi, Di and Majda, Andrew J},
  journal={Proceedings of the National Academy of Sciences},
  volume={117},
  number={1},
  pages={52--59},
  year={2020},
  publisher={National Academy of Sciences}
}

@incollection{lehnertz2006epilepsy,
  title={Epilepsy: Extreme events in the human brain},
  author={Lehnertz, Klaus},
  booktitle={Extreme events in nature and society},
  pages={123--143},
  year={2006},
  publisher={Springer}
}

@article{bobra2016predicting,
  title={Predicting coronal mass ejections using machine learning methods},
  author={Bobra, Monica G and Ilonidis, Stathis},
  journal={The Astrophysical Journal},
  volume={821},
  number={2},
  pages={127},
  year={2016},
  publisher={IOP Publishing}
}

@article{ummenhofer2017extreme,
  title={Extreme weather and climate events with ecological relevance: a review},
  author={Ummenhofer, Caroline C and Meehl, Gerald A},
  journal={Philosophical Transactions of the Royal Society B: Biological Sciences},
  volume={372},
  number={1723},
  pages={20160135},
  year={2017},
  publisher={The Royal Society}
}

@article{sapsis2021statistics,
  title={Statistics of extreme events in fluid flows and waves},
  author={Sapsis, Themistoklis P},
  journal={Annual Review of Fluid Mechanics},
  volume={53},
  number={1},
  pages={85--111},
  year={2021},
  publisher={Annual Reviews}
}

@article{chang2025extreme,
  title={Extreme Event Aware ($\eta $-) Learning},
  author={Chang, Kai and Sapsis, Themistoklis P},
  journal={arXiv preprint arXiv:2510.19161},
  year={2025}
}

@article{pickering2022discovering,
  title={Discovering and forecasting extreme events via active learning in neural operators},
  author={Pickering, Ethan and Guth, Stephen and Karniadakis, George Em and Sapsis, Themistoklis P},
  journal={Nature Computational Science},
  volume={2},
  number={12},
  pages={823--833},
  year={2022},
  publisher={Nature Publishing Group US New York}
}

@article{rudy2023output,
  title={Output-weighted and relative entropy loss functions for deep learning precursors of extreme events},
  author={Rudy, Samuel H and Sapsis, Themistoklis P},
  journal={Physica D: Nonlinear Phenomena},
  volume={443},
  pages={133570},
  year={2023},
  publisher={Elsevier}
}

@article{jia2023bioelectrical,
  title={A bioelectrical phase transition patterns the first vertebrate heartbeats},
  author={Jia, Bill Z and Qi, Yitong and Wong-Campos, J David and Megason, Sean G and Cohen, Adam E},
  journal={Nature},
  volume={622},
  number={7981},
  pages={149--155},
  year={2023},
  publisher={Nature Publishing Group UK London}
}

@article{vinoth2025extreme,
  title={Extreme events in gene regulatory networks with time-delays},
  author={Vinoth, S and Kingston, S Leo and Srinivasan, Sabarathinam and Kumarasamy, Suresh and Kapitaniak, Tomasz},
  journal={Scientific Reports},
  volume={15},
  number={1},
  pages={13064},
  year={2025},
  publisher={Nature Publishing Group UK London}
}

@article{fukami2023grasping,
  title={Grasping extreme aerodynamics on a low-dimensional manifold},
  author={Fukami, Kai and Taira, Kunihiko},
  journal={Nature Communications},
  volume={14},
  number={1},
  pages={6480},
  year={2023},
  publisher={Nature Publishing Group UK London}
}

@article{fukami2025extreme,
  title={Extreme vortex-gust airfoil interactions at Reynolds number 5000},
  author={Fukami, Kai and Smith, Luke and Taira, Kunihiko},
  journal={Physical Review Fluids},
  volume={10},
  number={8},
  pages={084703},
  year={2025},
  publisher={APS}
}

@book{coles2001introduction,
  title={An introduction to statistical modeling of extreme values},
  author={Coles, Stuart and Bawa, Joanna and Trenner, Lesley and Dorazio, Pat},
  volume={208},
  year={2001},
  publisher={Springer}
}

@book{dembo2009large,
  title={Large deviations techniques and applications},
  author={Dembo, Amir},
  year={2009},
  publisher={Springer}
}

@article{ragone2020computation,
  title={Computation of extreme values of time averaged observables in climate models with large deviation techniques},
  author={Ragone, Francesco and Bouchet, Freddy},
  journal={Journal of Statistical Physics},
  volume={179},
  number={5},
  pages={1637--1665},
  year={2020},
  publisher={Springer}
}

@article{farazmand2019closed,
  title={Closed-loop adaptive control of extreme events in a turbulent flow},
  author={Farazmand, Mohammad and Sapsis, Themistoklis P},
  journal={Physical Review E},
  volume={100},
  number={3},
  pages={033110},
  year={2019},
  publisher={APS}
}

@article{kochkov2021machine,
  title={Machine learning--accelerated computational fluid dynamics},
  author={Kochkov, Dmitrii and Smith, Jamie A and Alieva, Ayya and Wang, Qing and Brenner, Michael P and Hoyer, Stephan},
  journal={Proceedings of the National Academy of Sciences},
  volume={118},
  number={21},
  pages={e2101784118},
  year={2021},
  publisher={National Academy of Sciences}
}

@article{sweby1984high,
  title={High resolution schemes using flux limiters for hyperbolic conservation laws},
  author={Sweby, Peter K},
  journal={SIAM journal on numerical analysis},
  volume={21},
  number={5},
  pages={995--1011},
  year={1984},
  publisher={SIAM}
}

@article{kassam2005fourth,
  title={Fourth-order time-stepping for stiff PDEs},
  author={Kassam, Aly-Khan and Trefethen, Lloyd N},
  journal={SIAM Journal on Scientific Computing},
  volume={26},
  number={4},
  pages={1214--1233},
  year={2005},
  publisher={SIAM}
}

@article{obukhov1983kolmogorov,
  title={Kolmogorov flow and laboratory simulation of it},
  author={Obukhov, AM},
  journal={Russ. Math. Surv},
  volume={38},
  number={4},
  pages={113--126},
  year={1983}
}

@article{platt1991investigation,
  title={An investigation of chaotic Kolmogorov flows},
  author={Platt, Nathan and Sirovich, Lawrence and Fitzmaurice, Nessan},
  journal={Physics of Fluids A: Fluid Dynamics},
  volume={3},
  number={4},
  pages={681--696},
  year={1991},
  publisher={American Institute of Physics}
}

@article{borue1996numerical,
  title={Numerical study of three-dimensional Kolmogorov flow at high Reynolds numbers},
  author={Borue, Vadim and Orszag, Steven A},
  journal={Journal of Fluid Mechanics},
  volume={306},
  pages={293--323},
  year={1996},
  publisher={Cambridge University Press}
}

@article{qin2025clean,
  title={Clean numerical simulation of three-dimensional turbulent Kolmogorov flow},
  author={Qin, Shijie and Liao, Shijun},
  journal={Physics of Fluids},
  volume={37},
  number={10},
  year={2025},
  publisher={AIP Publishing}
}

@book{constantine2015active,
  title={Active subspaces: Emerging ideas for dimension reduction in parameter studies},
  author={Constantine, Paul G},
  year={2015},
  publisher={SIAM}
}

@article{zahm2020gradient,
  title={Gradient-based dimension reduction of multivariate vector-valued functions},
  author={Zahm, Olivier and Constantine, Paul G and Prieur, Cl{\'e}mentine and Marzouk, Youssef M},
  journal={SIAM Journal on Scientific Computing},
  volume={42},
  number={1},
  pages={A534--A558},
  year={2020},
  publisher={SIAM}
}

@article{moore1981bt,
  author={Moore, B.},
  journal={IEEE Transactions on Automatic Control}, 
  title={Principal component analysis in linear systems: Controllability, observability, and model reduction}, 
  year={1981},
  volume={26},
  number={1},
  pages={17-32},
  doi={10.1109/TAC.1981.1102568}}

@article{rowley2005model,
  title={Model reduction for fluids, using balanced proper orthogonal decomposition},
  author={Rowley, Clarence W},
  journal={International Journal of Bifurcation and Chaos},
  volume={15},
  number={03},
  pages={997--1013},
  year={2005},
  publisher={World Scientific}
}

@article{onorato2005modulational,
  title={Modulational instability and non-Gaussian statistics in experimental random water-wave trains},
  author={Onorato, Miguel and Osborne, Alfred Richard and Serio, M and Cavaleri, L},
  journal={Physics of Fluids},
  volume={17},
  number={7},
  year={2005},
  publisher={AIP Publishing}
}

@article{solli2007optical,
  title={Optical rogue waves},
  author={Solli, Daniel R and Ropers, Claus and Koonath, Prakash and Jalali, Bahram},
  journal={nature},
  volume={450},
  number={7172},
  pages={1054--1057},
  year={2007},
  publisher={Nature Publishing Group UK London}
}

@article{cebrian2024six,
  title={Six decades of the FitzHugh--Nagumo model: A guide through its spatio-temporal dynamics and influence across disciplines},
  author={Cebri{\'a}n-Lacasa, D. and Parra-Rivas, P. and Ruiz-Reyn{\'e}s, D. and Gelens, L.},
  journal={Physics Reports},
  volume={1096},
  pages={1--39},
  year={2024},
  publisher={Elsevier}
}

@article{romeira2016regenerative,
  title={Regenerative memory in time-delayed neuromorphic photonic resonators},
  author={Romeira, Bruno and Av{\'o}, Ricardo and Figueiredo, Jos{\'e} ML and Barland, St{\'e}phane and Javaloyes, Julien},
  journal={Scientific reports},
  volume={6},
  number={1},
  pages={19510},
  year={2016},
  publisher={Nature Publishing Group UK London}
}

@inproceedings{
li2021fourier,
title={Fourier Neural Operator for Parametric Partial Differential Equations},
author={Z. Li and N. Borislavov Kovachki and K. Azizzadenesheli and B. Liu and K. Bhattacharya and A. Stuart and A. Anandkumar},
booktitle={ICLR},
year={2021},
url={https://openreview.net/forum?id=c8P9NQVtmnO}
}

@article{katsidoniotaki2026dynamics,
  title={Dynamics-Informed Deep Learning for Predicting Extreme Events},
  author={Katsidoniotaki, Eirini and Sapsis, Themistoklis P},
  journal={arXiv preprint arXiv:2603.10777},
  year={2026}
}

@article{babaee2016minimization,
  title={A minimization principle for the description of modes associated with finite-time instabilities},
  author={Babaee, H and Sapsis, TP},
  journal={Proceedings of the Royal Society A: Mathematical, Physical and Engineering Sciences},
  volume={472},
  number={2186},
  year={2016},
  publisher={The Royal Society}
}

@article{babaee2017reduced,
  title={Reduced-order description of transient instabilities and computation of finite-time Lyapunov exponents},
  author={Babaee, H. and Farazmand, M. and Haller, G. and Sapsis, T. P.},
  journal={Chaos},
  volume={27},
  number={6},
  year={2017},
  publisher={AIP Publishing}
}

@book{lorenz1956empirical,
  title={Empirical orthogonal functions and statistical weather prediction},
  author={Lorenz, Edward N and others},
  volume={1},
  year={1956},
  publisher={Massachusetts Institute of Technology, Department of Meteorology Cambridge}
}

@article{guth2019machine,
  title={Machine learning predictors of extreme events occurring in complex dynamical systems},
  author={Guth, Stephen and Sapsis, Themistoklis P},
  journal={Entropy},
  volume={21},
  number={10},
  pages={925},
  year={2019},
  publisher={MDPI}
}

@article{schmidt2019conditional,
  title={A conditional space--time POD formalism for intermittent and rare events: example of acoustic bursts in turbulent jets},
  author={Schmidt, Oliver T and Schmid, Peter J},
  journal={Journal of Fluid Mechanics},
  volume={867},
  pages={R2},
  year={2019},
  publisher={Cambridge University Press}
}

@article{schmidt2026data,
  title={Data-driven forecasting of high-dimensional transient and stationary processes via space--time projection},
  author={Schmidt, Oliver T},
  journal={Proceedings of the Royal Society A: Mathematical, Physical and Engineering Sciences},
  volume={482},
  number={2329},
  year={2026},
  publisher={The Royal Society}
}

@article{cliver2022extreme,
  title={Extreme solar events},
  author={Cliver, Edward W and Schrijver, Carolus J and Shibata, Kazunari and Usoskin, Ilya G},
  journal={Living Reviews in Solar Physics},
  volume={19},
  number={1},
  pages={2},
  year={2022},
  publisher={Springer}
}

@article{berkooz1993proper,
  title={The proper orthogonal decomposition in the analysis of turbulent flows},
  author={Berkooz, Gal and Holmes, Philip and Lumley, John L},
  journal={Ann. Rev. Fluid Mech.},
  volume={25},
  number={1},
  pages={539--575},
  year={1993},
  publisher={Annual Reviews 4139 El Camino Way, PO Box 10139, Palo Alto, CA 94303-0139, USA}
}

@article{suykens1999chaos,
  title={Chaos control using least-squares support vector machines},
  author={Suykens, Johan AK and Vandewalle, Joos},
  journal={International journal of circuit theory and applications},
  volume={27},
  number={6},
  pages={605--615},
  year={1999},
  publisher={Wiley Online Library}
}

@article{gerster2020fitzhugh,
  title={FitzHugh--Nagumo oscillators on complex networks mimic epileptic-seizure-related synchronization phenomena},
  author={Gerster, M. and Berner, R. and Sawicki, J. and Zakharova, A. and {\v{S}}koch, A. and Hlinka, J. and Lehnertz, K. and Sch{\"o}ll, E.},
  journal={Chaos},
  volume={30},
  number={12},
  year={2020},
  publisher={AIP Publishing}
}

@article{nash2004electromechanical,
  title={Electromechanical model of excitable tissue to study reentrant cardiac arrhythmias},
  author={Nash, M. P. and Panfilov, A. V.},
  journal={Progress in biophysics and molecular biology},
  volume={85},
  number={2-3},
  pages={501--522},
  year={2004},
  publisher={Elsevier}
}

@article{scialla2021hubs,
  title={Hubs, diversity, and synchronization in FitzHugh-Nagumo oscillator networks: Resonance effects and biophysical implications},
  author={Scialla, Stefano and Loppini, Alessandro and Patriarca, Marco and Heinsalu, Els},
  journal={Physical Review E},
  volume={103},
  number={5},
  pages={052211},
  year={2021},
  publisher={APS}
}

@article{lin2004resonance,
  title={Resonance tongues and patterns in periodically forced reaction-diffusion systems},
  author={Lin, Anna L and Hagberg, Aric and Meron, Ehud and Swinney, Harry L},
  journal={Physical Review E—Statistical, Nonlinear, and Soft Matter Physics},
  volume={69},
  number={6},
  pages={066217},
  year={2004},
  publisher={APS}
}

@article{boulle2024operator,
  title={Operator learning without the adjoint},
  author={Boull{\'e}, Nicolas and Halikias, Diana and Otto, Samuel E and Townsend, Alex},
  journal={Journal of Machine Learning Research},
  volume={25},
  number={364},
  pages={1--54},
  year={2024}
}

@incollection{boulle2024mathematical,
  title={A mathematical guide to operator learning},
  author={Boull{\'e}, Nicolas and Townsend, Alex},
  booktitle={Handbook of Numerical Analysis},
  volume={25},
  pages={83--125},
  year={2024},
  publisher={Elsevier}
}

@article{butler1992three,
  title={Three-dimensional optimal perturbations in viscous shear flow},
  author={Butler, Kathryn M and Farrell, Brian F},
  journal={Physics of Fluids A: Fluid Dynamics},
  volume={4},
  number={8},
  pages={1637--1650},
  year={1992},
  publisher={American Institute of Physics}
}

@article{thompson1998initial,
  title={Initial conditions for optimal growth in a coupled ocean--atmosphere model of ENSO},
  author={Thompson, CJ},
  journal={Journal of the atmospheric sciences},
  volume={55},
  number={4},
  pages={537--557},
  year={1998}
}

@article{herrmann2021data,
  title={Data-driven resolvent analysis},
  author={Herrmann, Benjamin and Baddoo, Peter J and Semaan, Richard and Brunton, Steven L and McKeon, Beverley J},
  journal={Journal of Fluid Mechanics},
  volume={918},
  pages={A10},
  year={2021},
  publisher={Cambridge University Press}
}

@article{ilak2008modeling,
  title={Modeling of transitional channel flow using balanced proper orthogonal decomposition},
  author={Ilak, Milo{\v{s}} and Rowley, Clarence W},
  journal={Physics of Fluids},
  volume={20},
  number={3},
  year={2008},
  publisher={AIP Publishing}
}

@article{otto2022optimizing,
  title={Optimizing oblique projections for nonlinear systems using trajectories},
  author={Otto, Samuel E and Padovan, Alberto and Rowley, Clarence W},
  journal={SIAM Journal on Scientific Computing},
  volume={44},
  number={3},
  pages={A1681--A1702},
  year={2022},
  publisher={SIAM}
}

@article{mckeon2010critical,
  title={A critical-layer framework for turbulent pipe flow},
  author={McKeon, Beverley J and Sharma, Ati S},
  journal={Journal of Fluid Mechanics},
  volume={658},
  pages={336--382},
  year={2010},
  publisher={Cambridge University Press}
}

@article{luhar2014opposition,
  title={Opposition control within the resolvent analysis framework},
  author={Luhar, Mitul and Sharma, Ati S and McKeon, Beverley J},
  journal={Journal of Fluid Mechanics},
  volume={749},
  pages={597--626},
  year={2014},
  publisher={Cambridge University Press}
}

@article{benner2018mathcalh_2,
  title={$\mathcal{H}_2$-quasi-optimal model order reduction for quadratic-bilinear control systems},
  author={Benner, P. and Goyal, P. and Gugercin, S.},
  journal={SIAM J. Matrix Anal. \& App.},
  volume={39},
  number={2},
  pages={983--1032},
  year={2018},
  publisher={SIAM}
}

@article{benner2024balanced,
  title={Balanced truncation for quadratic-bilinear control systems},
  author={Benner, Peter and Goyal, Pawan},
  journal={Advances in Computational Mathematics},
  volume={50},
  number={4},
  pages={88},
  year={2024},
  publisher={Springer}
}

@article{towne2018spectral,
  title={Spectral proper orthogonal decomposition and its relationship to dynamic mode decomposition and resolvent analysis},
  author={Towne, Aaron and Schmidt, Oliver T and Colonius, Tim},
  journal={Journal of Fluid Mechanics},
  volume={847},
  pages={821--867},
  year={2018},
  publisher={Cambridge University Press}
}

@software{jax2018github,
  author = {J. Bradbury and R. Frostig and P. Hawkins and M. J. Johnson and Y. Katariya and C. Leary and D. Maclaurin and G. Necula and A. Paszke and J. Vander{P}las and S. Wanderman-{M}ilne and Q. Zhang},
  title = {{JAX}: composable transformations of {P}ython+{N}um{P}y programs},
  url = {http://github.com/jax-ml/jax},
  version = {0.3.13},
  year = {2018},
}

@article{pedregosa2011scikit,
  title={Scikit-learn: Machine learning in Python},
  author={Pedregosa, Fabian and Varoquaux, Ga{\"e}l and Gramfort, Alexandre and Michel, Vincent and Thirion, Bertrand and Grisel, Olivier and Blondel, Mathieu and Prettenhofer, Peter and Weiss, Ron and Dubourg, Vincent and others},
  journal={the Journal of machine Learning research},
  volume={12},
  pages={2825--2830},
  year={2011},
  publisher={JMLR. org}
}

@article{chandler2013invariant,
  title={Invariant recurrent solutions embedded in a turbulent two-dimensional Kolmogorov flow},
  author={Chandler, Gary J and Kerswell, Rich R},
  journal={Journal of Fluid Mechanics},
  volume={722},
  pages={554--595},
  year={2013},
  publisher={Cambridge University Press}
}

@book{holmes2012turbulence,
  title={Turbulence, coherent structures, dynamical systems and symmetry},
  author={Holmes, Philip},
  year={2012},
  publisher={Cambridge university press}
}

@phdthesis{kidger2021on,
    title={{O}n {N}eural {D}ifferential {E}quations},
    author={Patrick Kidger},
    year={2021},
    school={University of Oxford},
}

@article{scherpen1993balancing,
  title={Balancing for nonlinear systems},
  author={Scherpen, Jacquelien MA},
  journal={Systems \& Control Letters},
  volume={21},
  number={2},
  pages={143--153},
  year={1993},
  publisher={Elsevier}
}

@article{fujimoto2008computation,
  title={Computation of nonlinear balanced realization and model reduction based on Taylor series expansion},
  author={Fujimoto, Kenji and Tsubakino, Daisuke},
  journal={Systems \& Control Letters},
  volume={57},
  number={4},
  pages={283--289},
  year={2008},
  publisher={Elsevier}
}

@article{guibas2021adaptive,
  title={Adaptive fourier neural operators: Efficient token mixers for transformers},
  author={Guibas, John and Mardani, Morteza and Li, Zongyi and Tao, Andrew and Anandkumar, Anima and Catanzaro, Bryan},
  journal={arXiv preprint arXiv:2111.13587},
  year={2021}
}

@article{Willcox2002aiaaj,
	author = {K. Willcox AND J. Peraire},
	journal = {AIAA Journal},
	number = {11},
	pages = {2323--2330},
	title = {Balanced Model Reduction via the Proper Orthogonal Decomposition},
	volume = {40},
	year = {2002}}

@book{Brunton2022book,
	author = {S. L. Brunton and J. N. Kutz},
	edition = {2nd},
	publisher = {Cambridge University Press},
	title = {Data-Driven Science and Engineering: Machine Learning, Dynamical Systems, and Control},
	year = {2022}}

@article{Sirovich:1987,
	author = {L. Sirovich},
	journal = {Q. Appl. Math.},
	number = {3},
	pages = {561--590},
	title = {Turbulence and the dynamics of coherent structures, parts {I-III}},
	volume = {XLV},
	year = {1987}}
 }
 \end{spacing}

\section*{Acknowledgments}
The authors acknowledge support from the National Science Foundation AI Institute in Dynamic Systems
grant number 2112085 and The Boeing Company (NZ, SM SLB).

\section*{Author Contributions}
NZ designed, performed research, and analyzed results;
all authors were involved in discussions to interpret results;
SM wrote the differentiable simulator for the Kolmogorov flow; 
NZ wrote the FHN and MNLS simulators and implemented the methods on all of the different simulation environments;
SEO helped design and formulate the methods;
SLB received funds to support this work;
NZ wrote the paper, and all authors helped to review and edit. 

\section*{Competing Interests}
There are no competing interests.

\clearpage
\onecolumn
\setcounter{section}{0}\setcounter{figure}{0}\setcounter{equation}{0}\setcounter{table}{0}
\renewcommand{\thesection}{S-\arabic{section}}
\renewcommand{\thesubsection}{S-\arabic{section}\Alph{subsection}}
\renewcommand{\theequation}{S-\arabic{equation}}
\renewcommand{\figurename}{Supplementary Figure}
\renewcommand{\tablename}{Supplementary Table}
\section*{Supplementary Information}
\startcontents[supp]
\printcontents[supp]{l}{1}{\setcounter{tocdepth}{2}}
\clearpage

\setcounter{section}{0}\setcounter{figure}{0}\setcounter{equation}{0}\setcounter{table}{0}
\renewcommand{\thesection}{S-\arabic{section}}
\renewcommand{\thesubsection}{S-\arabic{section}\Alph{subsection}}

\renewcommand{\theHsection}{SI.\arabic{section}}
\renewcommand{\theHsubsection}{SI.\arabic{section}.\arabic{subsection}}
\renewcommand{\theHfigure}{SI.\arabic{figure}}
\renewcommand{\theHequation}{SI.\arabic{equation}}
\renewcommand{\theHtable}{SI.\arabic{table}}

\section{Extended Background}
\label{si:background}

\subsection{Kernel CoBRAS}
\label{si:cobras_kernel}
As discussed in the main text, we construct the CoBRAS projections using the SVD of the inner product matrix $\mathbf Y^T \mathbf X$. Because this only relies on inner products, CoBRAS naturally extends to nonlinear projections by defining the action of $\cobrasPsi$ in a reproducing kernel Hilbert space (RKHS) and is discussed rigorously in \cite{otto2023model}. For a differentiable kernel, $K(x,y)$, the computation only relies on evaluating the kernel gradients $\nabla K_y(x) = [\partial^{(e_1,0)}K(x,y), \dots \partial^{(e_n,0)}K(x,y)]$ (i.e. for fixed $y$, the gradient of the function $K(x,y) = K_y(x)$) and the inverse of the derivative Gram matrix $\mathbf G(\mathbf x) = [\partial^{(e_i, e_j)}K(x,x)]_{ij}$ (i.e. the symmetric matrix of partial derivatives evaluated on the diagonal $y=x$). Explicitly, the main object of interest is the function: 

\begin{equation}
    C(\mathbf {\tilde {x}},\mathbf g, \mathbf x) 
    = 
    \mathbf{g}^T \mathbf{G}(\mathbf x)^{-1}\nabla K_\mathbf x(\mathbf{\tilde{x}})
\end{equation}

\noindent where $\mathbf g$ is a sampled gradient of the quantity of interest as with linear CoBRAS. For $n_g$ sampled points $\mathbf {\tilde x}_i$ and gradients $\mathbf g_i$ we can construct the vector-valued function by evaluating $C$ at the data:

\begin{align}
        \tilde C(\mathbf x) 
    &= \frac{1}{\sqrt{n_g}} C(\mathbf{\tilde X}, \mathbf Y, \mathbf{x}) 
        = \left[\frac{1}{\sqrt{n_g}}C(\mathbf{\tilde x}_i, \mathbf{g}_i, \mathbf{x})
        \right]_{i=1}^{n_g}
\end{align}

\noindent Evaluating this at (possible different) samples $\mathbf x_i$, we build the balanced inner product matrix $\mathbf{\tilde C}$:

\begin{align}
    \mathbf {\tilde C} 
        &= \frac{1}{\sqrt{n_x n_g}} C(\mathbf{\tilde x}_i, \mathbf{g}_i, \mathbf{x}_j)
        = \mathbf U \mathbf \Sigma \mathbf V^T
\end{align}

\noindent and take its truncated SVD to obtain a final nonlinear embedding formula: 

\begin{equation}
    \mathbf z = h(\mathbf x) = \mathbf{\Sigma}_r^{-1/2} \mathbf{U}_r^T 
    \left(
        \tilde C(\mathbf{x}) - \tilde C(\mathbf{0}) 
    \right)
\end{equation}

For generic kernel, $K$, computing $\mathbf G(\mathbf x)^{-1}$ requires inverting the matrix at each point $\mathbf x$ which can be very expensive. However, for many kernels---such as linear, polynomial, and radial basis functions (RBFs)---both $\nabla K_y(x)$ and $\mathbf G(\mathbf x)^{-1}$ can be constructed analytically, bypassing the need to evaluate the derivative. In particular, for an RBF with length scale $\sigma$, $\mathbf G(\mathbf x)^{-1} = \sigma^2 \mathbf{I}$.

\section{Localized CoBRAS}
\label{si:localized}
We extend CoBRAS to localized phenomena by defining a reference neighborhood (such as the origin of a translationally invariant system), mapping snapshots and gradient samples into this neighborhood, and then building valid CoBRAS projections in this space. CoBRAS can then be applied to any reference point simply by mapping a new region of interest into this space. However, as we show below, we find that the CoBRAS $\cobrasPsi$ modes define a set of linear filters transforming the signal into a new one characterized by the most sensitive local modes, which we call the CoBRAS space. Surprisingly, the reconstruction of the entire signal is obtained just by multiplying the projection by the $\cobrasPhi$ modes at the reference point, $\eta = \eta_{ref}$, such as the origin ($\eta_{ref}=0$).

This approach should also extend to more general systems as long as there is a principled mapping between neighborhoods. This is natural with transitive group actions where every point can be mapped to some reference by an element of the group. For example, the $n$-torus and $\mathbb R^n$ under translations, $T(n)$,  or the $n$-sphere embedded in $\mathbb R^{n+1}$ under rotations, $SO(n+1)$. In these cases, the correlation operator depends on the group.

\subsection{A Global Signal Reconstruction Formula for Symmetric Spaces}
Let $\mathcal H = L^2(\Omega, \mu)$ be the Hilbert space of real-valued square integrable functions defined on the domain $\Omega$ with measure $\mu$. Let $\mathcal X \subseteq \mathcal{C}(\Omega) \cap \mathcal H$ be the set of continuous solutions to the dynamical system  
$$
\partial_t  \mathbf x 
    =\mathbf F(\mathbf x(t,\mnlsSpatial))
    ,\quad 
    \forall \mnlsSpatial \in \Omega
$$
\noindent and suppose the system is invariant with respect to some group of symmetries, $G$. That is, for any $\mathbf x \in \mathcal X, g\in G$, $g \cdot \mathbf x \in \mathcal X$. Most commonly, the group action is inherited by the action on the domain: $g \cdot \mathbf x(\eta) := \mathbf x(g\cdot \eta)$. 
For example, in the case of translations acting on $\Omega \subseteq \mathbb R^m$, translation by $\eta_0$ is given by $g_{\eta_0} \cdot \mathbf x(\eta) = \mathbf x(\eta - \eta_0)$.  A group is called \textit{transitive} if for every $\eta_1, \eta_2 \in \Omega$, there exists some $g \in G$ such that $g\cdot \eta_1 = \eta_2$. 
For $G$ acting transitively on $\Omega$ and some reference point $\eta_{\textit{ref}} \in \Omega$ (such as the origin if $\Omega \subseteq \mathbb R^m$, or the identity if $\Omega$ is a Lie group), we denote $g_\eta$ as a (possibly non-unique) element such that $g_\eta \cdot \eta_{\textit{ref}} = \eta$.

\begin{definition} \label{si_def:eps_recon}
    For $\epsilon > 0$, an operator $P$ \textbf{ $\epsilon$-reconstructs} the signal $\mathbf x$ on a set $U \subseteq \Omega$  if $|P(\mathbf x)(\eta) - \mathbf x(\eta)| < \epsilon$ for all $\eta \in U$. If $U$ is a neighborhood of some point $\eta_{\textit{ref}}$, we say $P$ \textbf{locally $\epsilon$-reconstructs} the signal near $\eta_{\textit{ref}}$. 
    If $U = \Omega$, we say that $P$  \textbf{globally $\epsilon$-reconstructs} the signal.
\end{definition}

Denote  $\cobrasPhi = [\phi_1(\eta), \dots, \phi_r(\eta)] \subset \mathcal H$ and $\cobrasPsi = [\psi_1(\eta), \dots, \psi_r(\eta)] \subset \mathcal H$, and define the projection operator  

\begin{align}
\label{eq:si_local_recon}
    P(\mathbf x) = \mathbf{\tilde x}(\eta) = \sum_k \phi_k(\eta) \langle \psi_k, \mathbf x\rangle    
\end{align}

\begin{proposition} \label{si_prop:lie_conv}
Let a group $G$ act transitively on $\Omega$ and $\mathcal X \subseteq \mathcal{C}(\Omega) \cap \mathcal H$ be invariant to $G$ as above. Suppose that  $\cobrasPhi$ and $\cobrasPsi$ define an operator $P$ as in Eq. ~\ref{eq:si_local_recon} that locally $\epsilon$-reconstructs a neighborhood $U$ of $\eta_{\textit{ref}}$ for every $\mathbf x \in \mathcal X$. If $\mu$ is $G$-invariant, then $\mathbf{x}(\eta)$ is globally $\epsilon$-reconstructed by

\begin{align}
    P_{\textit{global}}(\mathbf{x})(\eta) 
        &= \sum_k \phi_k(\eta_{\textit{ref}})\left(\psi_k \star_G \mathbf x\right)(g_\eta) \\
        &= 
        \cobrasPhi(\eta_{\textit{ref}}) 
        \cdot 
        \left(\cobrasPsi \star_G \mathbf{x} \right)(g_\eta)
\end{align}

\noindent where $\star_G$ is group transformed correlation operator:

$$\left(\psi \star_G \mathbf x\right)(g)
= \int_\Omega \psi(g^{-1} \cdot \eta) \mathbf{x}(\eta)d\mu(\eta)
$$

\end{proposition}

\begin{proof}
Define the function $\mathbf y(\xi) = \mathbf x(g_\eta\cdot  \xi)$ for arbitrary $\xi, \eta \in \Omega$. Since $\mathcal{X}$ is $G$-invariant, $\mathbf y \in \mathcal X$. Choose $\xi = \eta_{\textit{ref}}$, then  $\mathbf y(\eta_{\textit{ref}}) = \mathbf x(g_\eta\cdot  \eta_{\textit{ref}}) = \mathbf x(\eta)$. 
Because every signal is locally $\epsilon$-reconstructed in a neighborhood of $\eta_{\textit{ref}}$ and $\eta_{\textit{ref}}$ trivially lives in this neighborhood, we have $|\mathbf y(\eta_{\textit{ref}})  -  P (\mathbf{y})(\eta_{\textit{ref}})| < \epsilon$ where 

$$
 P (\mathbf{y})(\eta_{\textit{ref}})   
    = \sum_k \phi_k(\eta_{\textit{ref}}) \langle \psi_k, \mathbf y\rangle
$$

\noindent Substituting in the definition for the inner product: 
\begin{align*}
    \langle \psi_k, \mathbf y \rangle 
    &= \int_\Omega \psi_k(\xi) \mathbf y(\xi) d\mu(\xi)
    \\
    &= \int_\Omega \psi_k(\xi) \mathbf x(g_\eta \cdot \xi) d\mu(\xi) \\
    &= \int_\Omega \psi_k(g_\eta^{-1} \cdot \zeta) \mathbf x(\zeta) d\mu(\zeta) \\
    &= (\psi_k \star_G \mathbf x)(g_\eta)
\end{align*}

\noindent where we made the substitution $\xi = g_\eta^{-1} \cdot \zeta$ and used the fact that $\mu$ is $G$-invariant: $d\mu(g\cdot \zeta) = d\mu(\zeta)$ for all $g\in G$. Therefore:

$$
P(\mathbf y)(\eta_{\textit{ref}})= \sum_k \phi_k(\eta_{\textit{ref}}) (\psi_k \star_G \mathbf x)(g_\eta) = P_{\textit{global}}(\mathbf{x})(\eta)
$$

Since $\mathbf y(\eta_{\textit{ref}}) = \mathbf{x}(\eta)$, we have 

\begin{align*}
    &|P(\mathbf{y})(\eta_{\textit{ref}}) - \mathbf{y}(\eta_{\textit{ref}})| < \epsilon \\
    \Rightarrow 
    &|P(\mathbf{y})(\eta_{\textit{ref}}) - \mathbf{x}(\eta)| < \epsilon \\
    \Rightarrow
    &|P_{\textit{global}}(\mathbf{x})(\eta) - \mathbf{x}(\eta)| < \epsilon 
\end{align*}

\end{proof}

When $G$ is the space of translations acting on a periodic or infinite domain, the correlation operator is simply the normal choice and the origin acts as a standard reference. Note that we specifically used a $G$-invariant measure. When $\mu$ is not $G$-invariant, a Jacobian term enters the integral and requires extra care.

\textbf{Note}: While we restricted $\mathcal X$ to be in the space of continuous signals to simplify the presentation, it is possible to generalize Proposition ~\ref{si_prop:lie_conv} to functions that are discontinuous on a set of measure zero, i.e. PDE solutions with shock formation. This is achieved by relaxing the pointwise $\epsilon$-reconstruction definition in Definition ~\ref{si_def:eps_recon} with $\mu$-a.e. bounds and additionally requiring that $\eta_{ref}$ is a Lebesgue point for all signals $\mathbf x \in \mathcal X$ and functions $\phi_k$. This additional requirement is trivially satisfied in practice when signals are presented as discretized data. 

\clearpage

\clearpage
\section{Kolmogorov Flow}
\label{si:kflow}

\begin{figure}[t]
    \centering
    \includegraphics[width=0.95\textwidth]{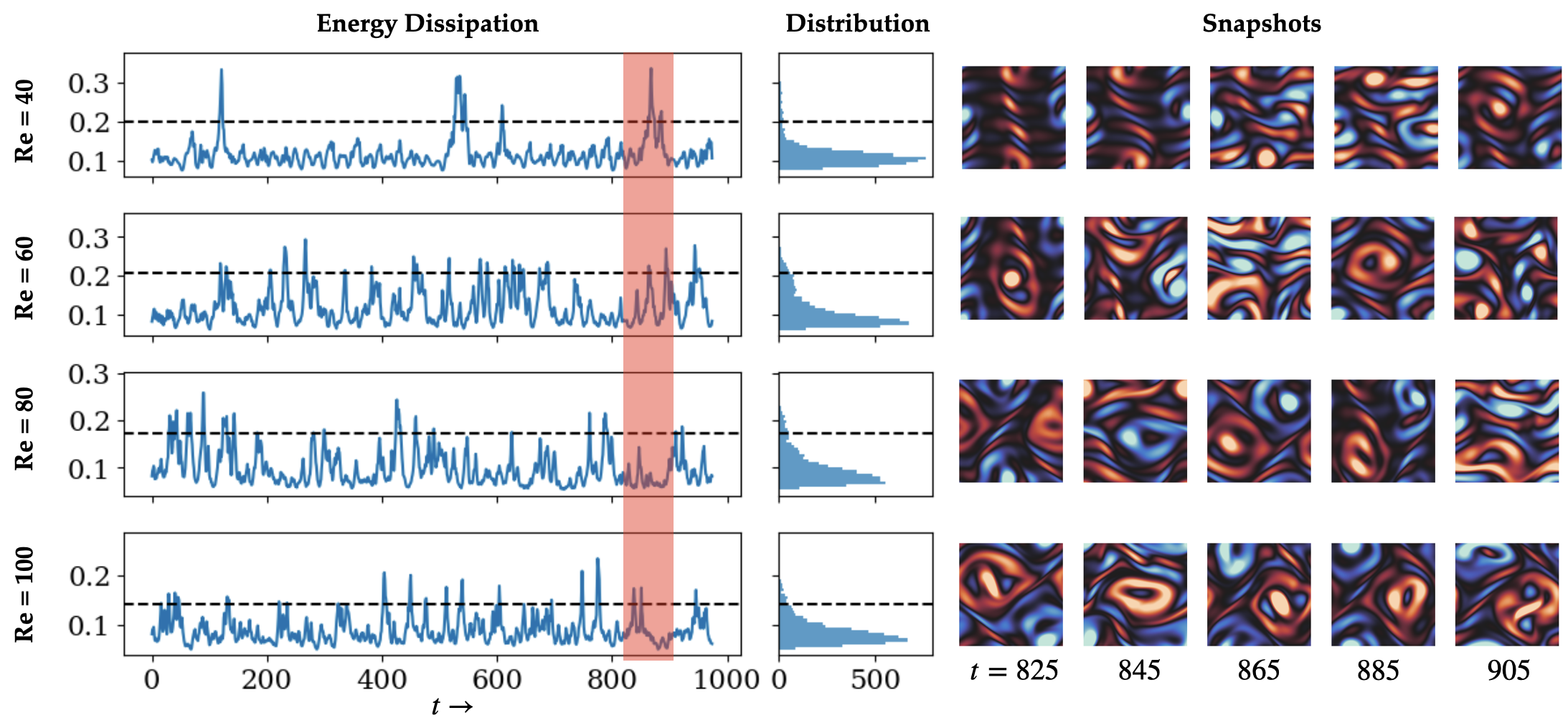}
    \cprotect \caption{\textbf{Kolmogorov flow at different $Re$.} \textit{Left}: Time series and histogram of energy dissipation values at various Re. Dashed horizontal line corresponds to two standard deviations above the mean for the respective flow's energy dissipation.  \textit{Right}: snapshots for each $Re$  between $t=825-905$ corresponding to the red band on the left.
    }
     \label{fig:si_kflow_snaps}   
\end{figure}

The Kolmogorov flow is a 2D fluid flow governed by the incompressible Navier-Stokes equations with periodic boundary conditions in both directions driven by a spatially periodic forcing term. 
Its simplicity has made it a widely studied toy model for turbulence in analytic studies \cite{platt1991investigation}, simulation \cite{borue1996numerical,qin2025clean}, and even laboratory settings \cite{obukhov1983kolmogorov}. In particular, for a certain range of Reynolds number, the flow exhibits intermittent extreme events as energy transfers from larger to smaller scales as is typical of turbulent systems.
Explicitly, the Kolmogorov flow takes the form

\begin{equation} \label{eq:kflow}
    \partial_t \kFlowVel = - \kFlowVel \cdot \nabla \kFlowVel - \nabla p + \nu \Delta \kFlowVel + \mathbf f,  \qquad \nabla \cdot \kFlowVel = 0
\end{equation}

\noindent 
where $\kFlowVel(t, \mathbf x)$ is the velocity field and $p(t, \mathbf x)$ is the pressure field, $\nu = 1/Re$ is the kinematic viscosity and $\mathbf f(\mathbf x ) = (\sin(k_f y),0)$ is the time-independent periodic forcing. We consider spatial variables defined on the torus, $\mathbf x = (x,y) \in [0, 2\pi] \times [0, 2\pi]$. For all of our results, we consider $k_f = 4$. 

For this 2D system, we can write $\mathbf{v}= (\nabla \psi) \times \hat {\mathbf{e}}_3$, where $\psi$ is the stream function\footnote{
This is a slight abuse of notation---this has no relationship to the CoBRAS modes $\cobrasPsi = [\psi_1, \dots \psi_r]$.
}.
Taking the curl of this expression yields the vorticity $\omega(t, \mathbf x) = \nabla \times \kFlowVel = -\nabla^2 \psi$. 
Thus Equation ~\ref{eq:kflow} can be represented entirely in terms of the vorticity;
for brevity, we will rewrite the discretized Equation ~\ref{eq:kflow} as $\partial_t \kFlowVort = \mathbf F(\kFlowVort, \mathbf{x})$. 
The energy dissipation can be written in terms of an energy functional of the vorticity 

\begin{equation}
    \energyDissipation(\kFlowVort) = \frac{\nu}{4\pi^2} \int |\kFlowVort(t, \mathbf x)|^2 d \mathbf x
\end{equation}

\noindent where $4\pi^2$ is the area of the spatial domain. 
As the Reynolds number, $Re$, increases, the system begins to exhibit intermittent (and chaotic) bursting at $Re=35$ \cite{farazmand2017a}. We focus our attention primarily on $Re=40$, when the bursts are rare, but we also analyze the behavior at $Re \in \{60,80, 100\}$ (see \figurename ~\ref{fig:si_kflow_snaps}).

To use our method, we simulate the Kolmogorov flow with a $256 \times 256$ 
spatial resolution ($\kFlowVort \in \mathbb{R}^{65,536}$) and integrate a single trajectory from 
$t \in [0, 5000]$ with $dt_{sim} = 10^{-3}$, 
saving every $1$ time unit for a total of 5000 forward snapshots. Following \cite{kochkov2021machine}, we integrate the system using a pseudospectral solver with a Crank-Nicolson RK4 implicit-explicit time stepping scheme  \cite{sweby1984high}. We use the initial condition 
$\kFlowVort(x,y) = \cos(x) - \sin(y)$
to satisfy the zero-divergence property. We find that after just a few time units, the solution collapses to the attractor---we remove the first $10$ time units to avoid the transient when constructing our projections and performing analysis. 

For CoBRAS, we define our quantity of interest $q$ to be the energy dissipation and examine the map over the time horizon: 

$$
\mathbf \cobrasOutput(\kFlowVort(t)) = [\energyDissipation_1, \energyDissipation_2, \dots, \energyDissipation_{n}]
$$

\noindent where $\energyDissipation_k = \energyDissipation(t + k \Delta t)$, and we choose $\Delta t = 0.5$. In Section ~\ref{si:kflow_T}, we examine the effect of the adjoint horizon $T = n/2$ on the CoBRAS modes for $T \in \{1,2,4,8\}$. Finally, we denote $\mathbf k =(k_x,k_y)$ as the wave number vector and $a_{ij}$ as the projection onto the $(i,j)$-th Fourier mode (corresponding to the product of the $i$-th Fourier mode in the $x$-direction and the $j$-th Fourier mode in the $y$-direction).

\subsection{CoBRAS Modes}
\label{si:kflow_40_modes}

\begin{figure*}[t!]
    \centering
    \includegraphics[width=0.95\textwidth]{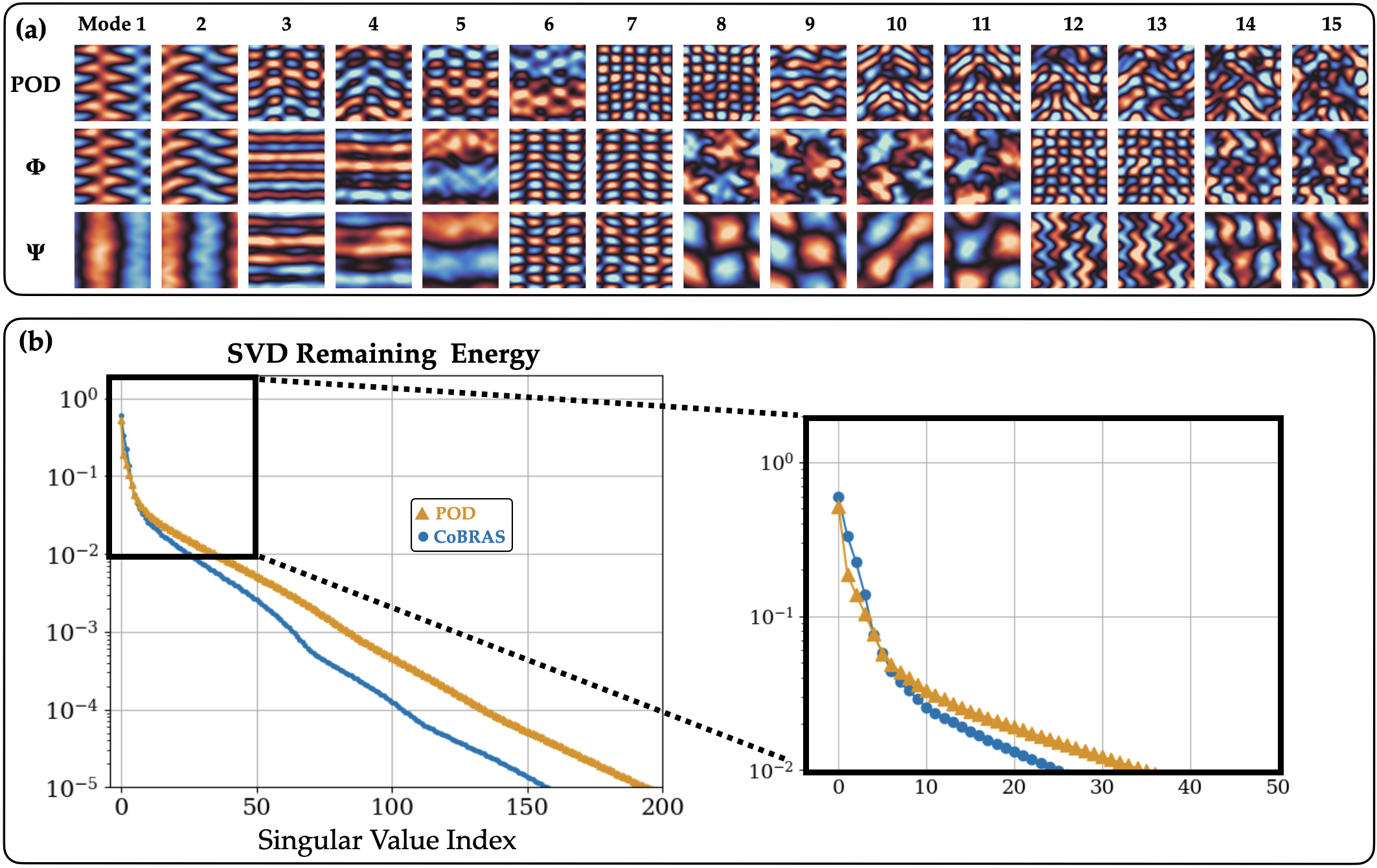}
    \cprotect \caption{\textbf{Kolmogorov Flow Modes.}
    \textit{Top:} The first 15 POD and $(\cobrasPhi, \cobrasPsi)$ modes. 
    \textit{Bottom:} The remaining energy for the POD and CoBRAS SVDs. 
    }
     \label{si_fig:kflow_modes}   
\end{figure*}

For all of our studies with the Kolmogorov flow, we use the first $n_x = n_g = 4000$ snapshots to compute the POD and CoBRAS modes. In this section, we specifically discuss the modes at $Re=40$. In Sections ~\ref{si:kflow_Re} and ~\ref{si:kflow_T}
 we examine how the modes change across $Re$ and $T$ respectively.

In \figurename ~\ref{si_fig:kflow_modes}(a), we provide the first several CoBRAS and POD modes beyond what was shown in the main text. As discussed in the main text, the first four CoBRAS $\cobrasPsi$ modes resemble components of  $(1,0)$ and $(0,4)$ Fourier modes. The 5th $\cobrasPsi$ mode resembles components of $(0,1)$, the 6th and 7th modes resemble $(2,4)$, and the 8th-11th modes resemble $(1,1)$.
Notably, these are all substantially different from the 6th and 7th POD modes, which resemble $(3,4)$. This is the range of singular values where the spectra begin to diverge in \figurename ~\ref{si_fig:kflow_modes}(b), indicating that the $(2,4)$ and $(1,1)$ modes are more relevant for characterizing the energy dissipation than the $(3,4)$ modes.

\subsection{Controlling the Kolmogorov Flow}
\label{si:kflow_ctrl}

At $Re=40$, the projection onto the first two linear CoBRAS modes admits a clear radial structure, where low-energy events live on a annulus-shaped body of the chaotic attractor and the high-energy events live off the main body on the interior of the ring in Figure \ref{fig:kflow}. This provides a natural control strategy to mitigate extreme events: by staying on the annulus, we should never produce an extreme event. To achieve this, one can simply define a reference radius, $R_{ref}$, within the annulus and apply proportional control to it. Explicitly, we project the state $\mathbf \kFlowVort$ onto the first two modes: $\mathbf z = \cobrasPsi^T \mathbf\kFlowVort$. Let $R_z = ||\mathbf z||$ and  $\mathbf{\hat z} = \mathbf z/R_z$. 
Then the closest point on the reference circle is $\mathbf z_{ref} = R_{ref}\mathbf{\hat{z}}$. 
The proportional control law in the projected space is then $\controlInput_z = -k_{gain}(\mathbf z - \mathbf{z}_{ref})\mathbf{\hat{z}}$; i.e. we simply push towards the boundary of the circle. 
In this work, we set the gain to be $k_{gain} = 1$ for simplicity.
We actuate the system by lifting back into vorticity coordinates, 
$\controlInput_{\omega} = \mathbf{\cobrasPhi} \controlInput_z$, 
and obtain the dynamics: 
$\partial_t\omega = \mathbf F(\omega, \mathbf{x}) + \controlInput_{\omega}$
\footnote{Note that 
    $\controlInput_z \in \mathbb R^2$, resulting in low-dimensional actuation, but $\cobrasPhi \controlInput_z$ affects the entire domain.
    }.

To investigate this, we select a withheld snapshot from our test set (after $t=4000$) and apply zero-hold feedback control for $1$ time unit (i.e. $dt_{ctrl}=1$ with actuation every $1/dt_{sim}$ simulation time steps) before recalculating the control input. We find that the control is effective at keeping the projected state on the circle and---while it does move substantially along the circle--it never produces an extreme event. Moreover, while we allow for unbounded control, there is very little energy added or subtracted into the system; where $\frac{\max_t ||\mathbf u(t)||^2}{\min_t||\mathbf x(t)||^2} \approx 5 \times 10^{-3}$ for the controlled trajectory. 

\subsection{Prediction of Events}
\label{si:kflow_prediction}
\begin{figure*}[t!]
    \centering
    \includegraphics[width=0.95\textwidth]{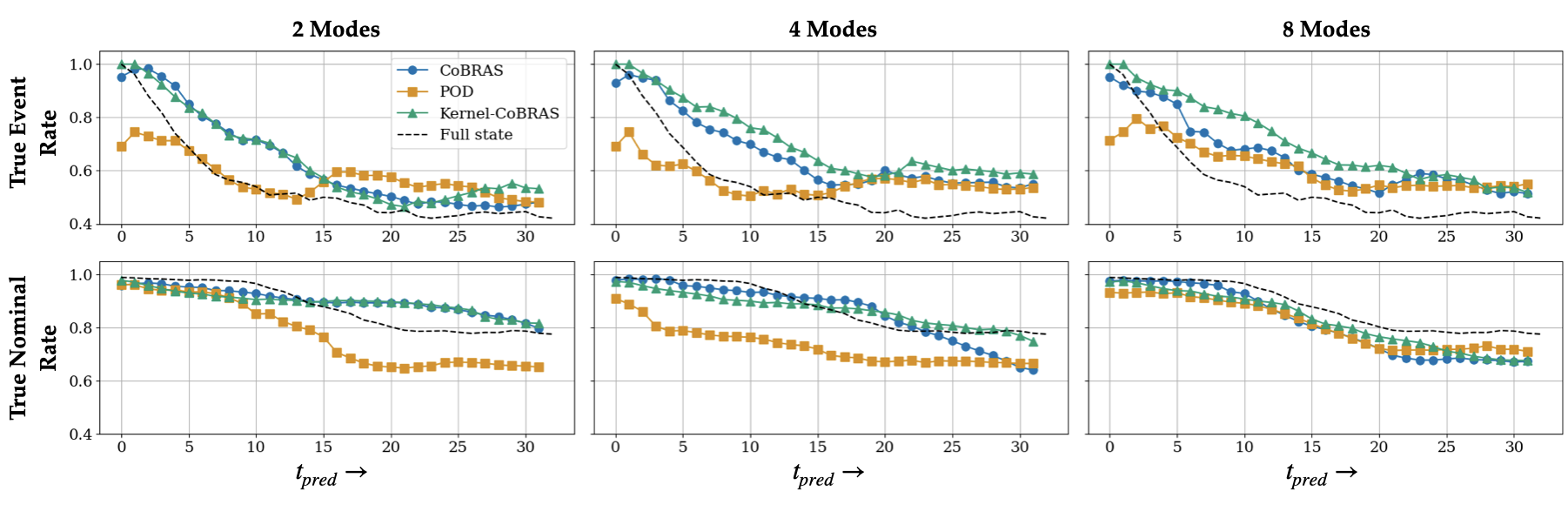}
    \cprotect \caption{\textbf{Kolmogorov Flow Prediction at $Re=40$}. 
    The predictive capability of RBF-SVMs using  for POD, CoBRAS, and Kernel-CoBRAS with different numbers of features. Each plot includes a full-state SVM baseline (dashed line). 
    }
     \label{si_fig:kflow_svm-40}   
\end{figure*}

We consider the predictive capability of these modes by training simple radial basis function (RBF) support vector machine (SVM) classifiers to determine whether an event will happen within a given time horizon. For each time $t_k$, we ask whether 

\begin{equation}
    \max_{t \in [t_k,t_k + t_{pred}]} \energyDissipation(t) > \energyDissipation^*
\end{equation}

\noindent where $\energyDissipation^*$ is two standard deviations above the mean. For all experiments, we use Scikit-Learn \cite{pedregosa2011scikit} to train the RBF classifier. The RBF Kernel is defined as 

$$
    k(x,x') = \exp(-\gamma ||x - x'||^2)
$$

To enable fair comparison across input scales, we first normalize the data before training to have zero mean and unit variance. To compare across the number of input feature dimensions $r$, we choose $\gamma = 1/r$, as is common practice. 
To accommodate the extremely imbalanced dataset, we use Scikit-Learn's ``balanced`` loss, which provides class-dependent weights to reduce bias towards the nominal class.
In this setting, we fix the $L^2$ regularization hyperparameter $C=1$ for all experiments.

In \figurename ~\ref{si_fig:kflow_svm-40}, we present the true positive and true negative rates for POD and CoBRAS at $Re=40$ for various prediction horizons $t_{pred} \in [0, 32]$. While not a focus of the main text due to the lack of interpretable modes, we also include a comparison with kernel-CoBRAS (see ~\ref{si:cobras_kernel}) to examine the effect of nonlinear sensitivity-balanced projections\footnote{
We also use an RBF kernel for defining the RKHS for kernel-CoBRAS. The scale factor $\gamma$ is chosen based on the variance of the input data. 
}. For each method, we project onto the first  $r$-modes (for $r=2,4$, and $8$) and compare to the performance of just using the full-state information.

As $t_{pred}$ increases, the class imbalance becomes less pronounced (since more points will \textit{eventually} become extreme events), but the decision boundary becomes less separable and each classifier's performance deteriorates.

Most strikingly, the CoBRAS projections significantly outperform the POD projections across time horizons---especially for  $t_{pred} < 10$---even with just two modes. For each case, CoBRAS true event rate plateaus to $50\%$ after $t_{pred}=16$ time units. The CoBRAS modes also provide more predictive power for accurately classifying the true events compared to the full-state SVM. This provides further evidence that by capturing the \textit{sensitivity} of the energy dissipation, the CoBRAS projections provide a meaningful predictor of events. 
It is worth noting that the CoBRAS performance remains relatively invariant to the number of modes---indicating that two modes truly capture the most relevant characteristics about the event formation.

We observe that the kernel-CoBRAS provides modest improvement predictions compared to linear CoBRAS. While further improvements could be observed with a different kernel defining the CoBRAS RKHS, we did not explore that in this work. Finally, while we observe improved predictions with CoBRAS, we note that we lose interpretability of the identified mechanism when developing the nonlinear projections. 

\subsection{Reynolds Number Dependence}
\label{si:kflow_Re}
\begin{figure*}[t!]
    \centering
    \includegraphics[width=0.90\textwidth]{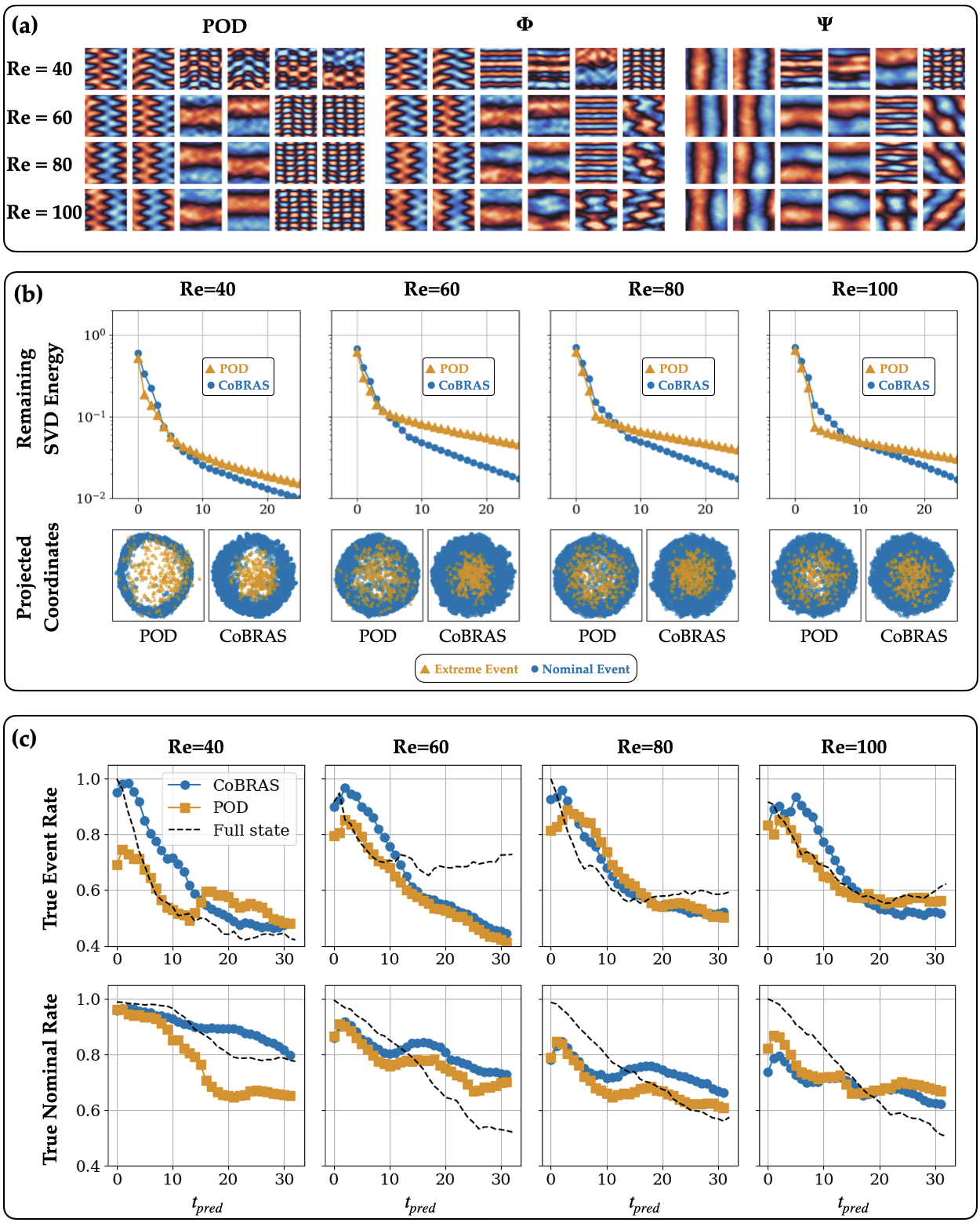}
    \cprotect \caption{\textbf{Effect of $Re$ $(T=4)$}. 
    \textbf{(a)} The POD and CoBRAS $(\cobrasPhi, \cobrasPsi)$ modes for $Re \in \{40,60,80, 100\}$. 
    \textbf{(b)} The remaining energy from the POD and CoBRAS SVDs and the projection onto the first two leading modes across $Re$. 
    \textbf{(c)} The effect of $Re$ on the SVM predictions for CoBRAS, POD, and the full-state across $t_{pred}$. CoBRAS and POD both use $r=2$ modes. 
    }
     \label{si_fig:kflow_re-modes}  
\end{figure*}

In this section, we examine the effect the Reynolds number has on the obtained modes and event predictability. Simulations are run for $Re\in \{40,60,80,100\}$ using the identical setup described in ~\ref{si:kflow}; all CoBRAS results were obtained with adjoint horizon $T=4$. 

As $Re$ increases, the number of events increases (as shown in \figurename ~\ref{fig:si_kflow_snaps}). In \figurename ~\ref{si_fig:kflow_re-modes}(a), we provide the first 6 POD and CoBRAS $(\cobrasPhi, \cobrasPsi)$ modes. After $Re=40$, the POD modes remain relatively stable. Notably, the CoBRAS modes that resemble the $(0,4)$-th Fourier mode shift importance with the ones that resemble $(0,1)$-th Fourier modes (from the 3rd $\cobrasPsi$ mode at $Re=40$ to the 5th mode at $Re=60-80$).

However, it's clear that $(1,0)$-Fourier mode dependence remains constant across all $Re$ that were examined. Indeed, the projection onto the first two modes is shown in \figurename ~\ref{si_fig:kflow_re-modes}(b). For both POD and CoBRAS, the extreme events remain concentrated in the interior of the disk, but the nominal events spread throughout the disk. In \figurename ~\ref{si_fig:kflow_re-modes}(c), we provide analogous SVM results detailed in Section ~\ref{si:kflow_prediction} using just $r=2$ modes. All models continue to degrade as $t_{pred}$ increases across all $Re$. With $r=2$ modes, the CoBRAS models remain better predictors of extreme events across a range of $Re$, with performance gap narrowing after $Re=60$.

\clearpage
\subsection{Adjoint Horizon Dependence}
\label{si:kflow_T}

\begin{figure*}[t!]
    \centering
    \includegraphics[width=0.95\textwidth]{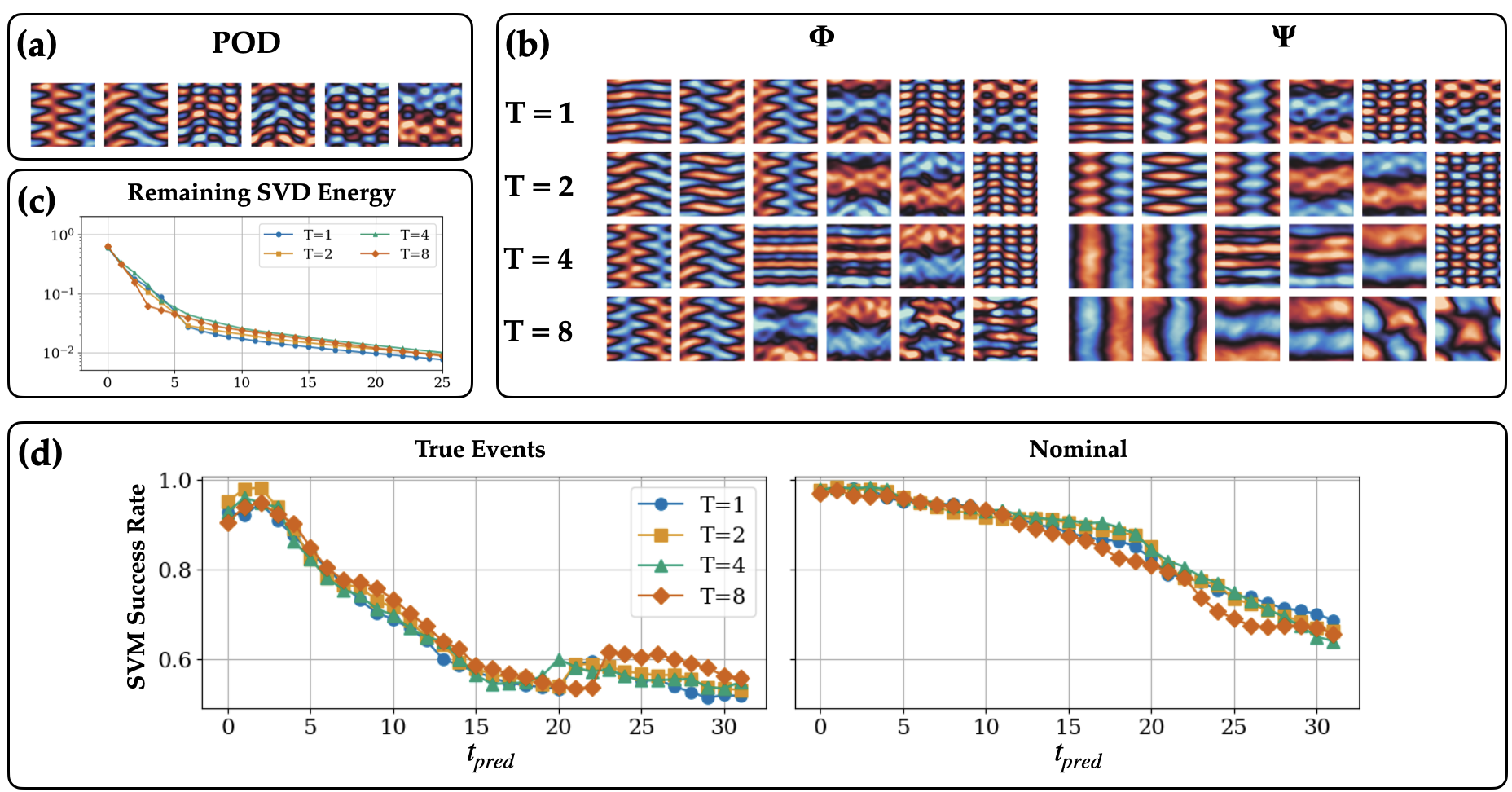}
    \cprotect \caption{\textbf{Effect of Adjoint Horizon, $T$ $(Re=40)$}. 
    \textbf{(a)} POD modes for $Re=40$. 
    \textbf{(b)} CoBRAS $(\cobrasPhi, \cobrasPsi)$ modes for $T \in \{1,2,4,8\}$.  
    \textbf{(c)} Remaining energy for the truncated CoBRAS SVD for across T. 
    \textbf{(c)} Effect of $T$ on SVM predictions for CoBRAS projections using $r=4$ modes across values $t_{pred}$. 
    }
     \label{si_fig:kflow_T-modes}   
\end{figure*}

In this section, we examine the effect of the adjoint horizon, $T$, when applying CoBRAS to the Kolmogorov flow. We use an identical setup to that described in ~\ref{si:kflow} at $Re=40$. In \figurename ~\ref{si_fig:kflow_T-modes}(a-b), we provide the first $6$ POD and CoBRAS $(\cobrasPhi, \cobrasPsi)$ modes across $T \in \{1,2,4,8\}$. 
As $T$ increases, the dependence on the $(0,4)$-th Fourier mode shifts. 
At $T=1$, $(0,4)$ appears as the first CoBRAS mode; at $T=2$, it 
mixes between the first three; at $T=4$ it mixes between the third and fourth; and at $T=8$ does not clearly appear in the first six modes. 
This indicates that the energy dissipation becomes less sensitive to the $(0,4)$-th direction over long periods of time. Because the system is chaotic, the adjoint information saturates over long time horizons and is dominated by the information contained in the body of the attractor. In \figurename ~\ref{si_fig:kflow_T-modes}(d), we provide the SVM predictions across $t_{pred}$ for $r=4$. The results nearly overlap across all values of $t_{pred}$, indicating that the 4-dimensional subspace contains nearly identical information for predicting extreme events regardless of the adjoint horizon used to obtain the CoBRAS modes. This indicates that the choice of $T$ is relatively insensitive for obtaining the dominant mechanisms.

\subsection{Continuous Symmetry Reduction}
\label{si:kflow_symm}
\begin{figure*}[t!]
    \centering
    \includegraphics[width=0.95\textwidth]{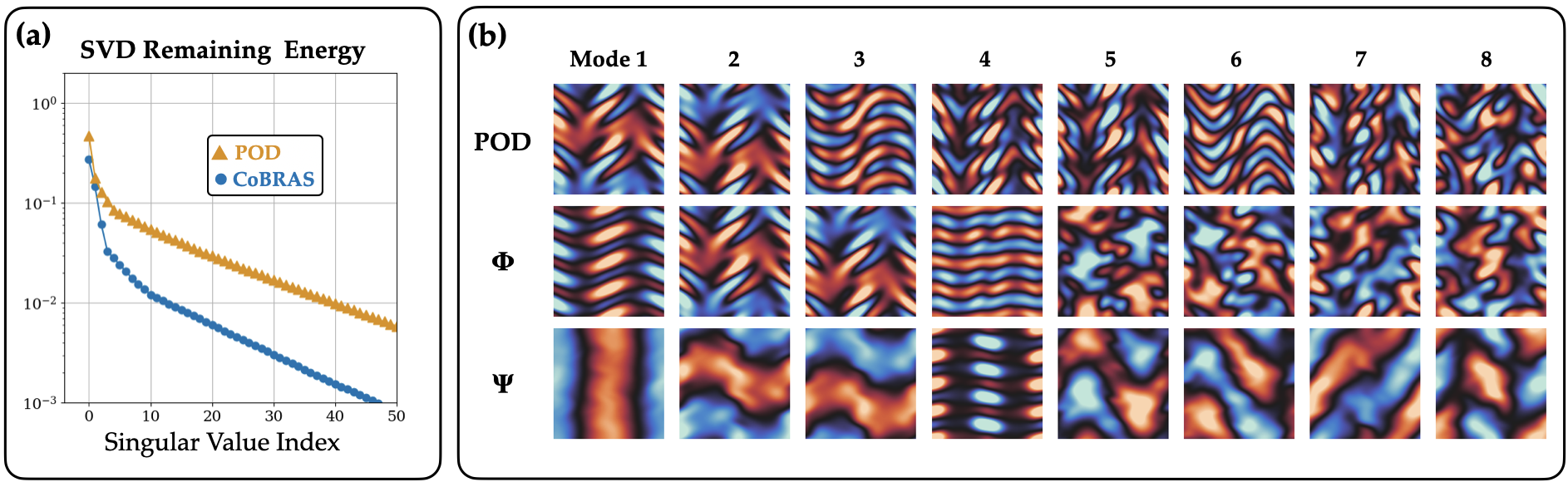}
    \cprotect \caption{\textbf{Effect of Symmetry Reduction $(T=4, Re=40)$}.
    \textbf{(a)} The remaining energy for the symmetry-reduced POD and CoBRAS SVDs. 
    \textbf{(b)} The first eight symmetry-reduced POD and CoBRAS $(\cobrasPhi, \cobrasPsi)$ modes.
    }
     \label{si_fig:kflow_symm}  
\end{figure*}

The Kolmogorov flow exhibits a significant amount of symmetry, including a continuous symmetry along the $x$-direction, discrete translations in the $y$-direction, and a reflection symmetry \cite{platt1991a, chandler2013invariant}. These symmetries describe redundancy in dynamics; by working in a particular frame of reference, one can ``quotient'' out the symmetry---effectively isolating the important dynamics of interest. While not essential to our method, we provide a brief exploration into how performing a symmetry reduction before applying CoBRAS might provide more insight into the underlying mechanisms at $Re=40$ and $T=4$. We only examine the effect of removing the continuous symmetry in the $x$-direction by translating all snapshots to a reference as in \cite{holmes2012turbulence}. In particular, we align snapshots into a frame of reference where the imaginary component of the first Fourier mode is zero for each snapshot. Explicitly, denote $S_x$ to be the shift operator that translates a snapshot by $x$ and denote $\tilde \kFlowVort = S_x \kFlowVort$ as the shifted snapshot. We define $\tilde{\mathbf X}$ as the set of shifted snapshots. 

While this can be done \textit{before} computing the gradients of the system, in practice, we collected the gradients and performed the symmetry reduction as a post-processing step. In particular, we compute the shift operator, $S_x$ for each initial snapshot $\kFlowVort_0$ used to compute the gradient $\mathbf g = \nabla_{\kFlowVort} (\boldsymbol{\xi}\cobrasOutput) (\kFlowVort_0)$. We then apply it to the gradient, $\tilde{\mathbf{g}} = S_x \mathbf g$, and build the matrix of gradients in the space $\tilde{\mathbf Y}$. We then proceed with CoBRAS as usual using the matrices $\tilde{\mathbf X}$ and $\tilde{\mathbf Y}$.

In \figurename ~\ref{si_fig:kflow_symm}(a), we present the first 8 POD and CoBRAS $(\cobrasPsi, \cobrasPhi)$ modes in the symmetry-reduced space. As one would expect, the first two CoBRAS modes (which differ only by a translation in the $x$-direction) from the previous study reduce to a single mode under the symmetry reduction. This indicates that there is only a single direction in the symmetry-reduced space; in fact, in the reduced space, this is identical to the radius of the disk shown in Figure \ref{fig:kflow}. Interestingly, whereas in the full space, we found a dependence on the $(0,1)$ and $(0,4)$ Fourier modes in $\cobrasPsi$, we see a different clean set of coherent structures appear in the symmetry-reduced $\cobrasPsi$. The 2nd and 3rd reduced modes resemble $\sin(y+\sin(x))$ up to a scale factor and phase in the $y$-direction---i.e., a modulation of the $(0,1)$-mode. The fourth $\cobrasPsi$ mode seems to have the periodicity of the $(0,4)$ Fourier mode in the $y$-direction and the $(1,0)$ Fourier mode in the $x$-direction, but there is a non-trivial modulation. Remarkably, this closely resembles  the traveling wave solution $T_8$ found in Figure 8c of \cite{farazmand2016b}, indicating that pushing solutions in the direction of the traveling wave may be responsible for higher-order sensitivities of the energy dissipation. While we do not examine these modes in detail, we believe that the symmetry-reduced modes may provide a promising direction for understanding the Kolmogorov flow in more detail.

\clearpage
\section{FitzHugh-Nagumo Networks}
\label{si:fhn}
The FitzHugh-Nagumo (FHN) model has been widely used to model complex networked dynamical systems exhibiting node excitation, such as electro-optical systems \cite{romeira2016regenerative}, networks of neurons firing \cite{gerster2020fitzhugh}, cardiac \cite{nash2004electromechanical} and pancreatic cells \cite{scialla2021hubs}, and biochemical networks \cite{lin2004resonance}. A detailed review of the model can be found in \cite{cebrian2024six}. 
As in \cite{ansmann2013a, karnatak2014a}, we consider coupled FHN networks of $N$ units $(x_i, y_i)$ governed by
\begingroup
\addtolength{\jot}{-10pt}
\begin{align} \nonumber 
\dot x_i & = x_i (a_i - x_i)(x_i - 1) - y_i + k \sum_j A_{ij}(x_j - x_i)\\  
\dot y_i & = b_i x_i - c y_i
\end{align}
\endgroup
\noindent for $i = 1,2, \dots,  N$, where $A_{ij}$ is the adjacency matrix, $a_i, b_i$ and $c_i$ are internal parameters for each unit, and $k = 0.128/(N-1)$ is the coupling strength. For simplicity, we choose $a_i=-0.02651$ and $c_i = 0.02$ to be constants for all nodes and $b_i$ to be distributed in $[0.006, 0.014]$, which breaks the symmetry among all nodes. 
Following \cite{ansmann2013a, karnatak2014a}, we consider three cases: 
\begin{enumerate}
    \item ~\ref{si:fhn_101}: a fully connected network of $N=101$ nodes
    \item ~\ref{si:fhn_1001}: a fully connected network of $N=1001$ nodes
    \item ~\ref{si:fhn_small_world}: a small-world network of $N=100 \times 100 = 10^4$ nodes
\end{enumerate}

\noindent The fully-connected networks have $A_{ij} = 1$ for all $i \neq j$ and we choose $b_i$ to be linearly increasing $b_i = 0.006 + 0.008\frac{i-1}{N-1}$. For the small-world case, we use a periodic lattice of $100\times100$ nodes with $60$-nearest neighbor connectivity and randomly generated long-range connections. In each of these three cases, the dynamics are chaotic and will exhibit spontaneous synchronization. 

To efficiently simulate this stiff system, we use a Dormand-Prince solver in the Diffrax package \cite{kidger2021on} (\verb|Dopri5| for the fully connected network, \verb|Dopri8| for the small-world), which is built on JAX and supports autodifferentiation. In particular, we use a PID stepsize-controller. For the fully-connected network, we use $10^{-9}$ absolute and relative tolerances, an initial $dt=10^{-2}$, minimal $dt=10^{-5}$ and a maximal $dt=0.5$. For the small-world network, we use tighter parameters: $10^{-12}$ absolute and relative tolerance, an initial $dt=10^{-3}$ and maximum $dt=5\times 10^{-3}$. Initial conditions are uniformly sampled within a $[-0.5,0.5]^{2N}$ hypercube. The first several thousand time units are removed to ensure the system was on the attractor. 

To apply CoBRAS in each case, we take $\mathbf x = (x_1, \dots, x_N, y_1, \dots y_N)$ to be our state variable and define our QoI to be the node-average energy of the system $q(t) = \frac{1}{2N}\sum_i\left( x_i(t)^2 + y_i(t)^2\right)$. We consider events to be extreme if this quantity exceeds 4 standard deviations above the mean.

\paragraph{Support Vector Machine (SVM).} Just as in ~\ref{si:kflow_prediction}, we use Scikit-Learn to train our classification models with an RBF kernel using features $\mathbf z = \cobrasPsi^T \mathbf x$ with balanced loss, scale factor $\gamma = 1/r$, and $C=1$. SVM input features are centered and scaled to have unit norm before training. To train the SVM we use $t_{pred} = 50$ time units to define our labels.

\paragraph{Control.} We design a controller based on the SVM decision boundary. The decision function, $d(\mathbf z)$, for an SVM with kernel $k$ can be computed: 
$$
    d(\mathbf z) = \sum_i \alpha_i \ell_i k(\mathbf z, \mathbf z_i) + \beta
$$
\noindent where $(\mathbf z_i, \ell_i)$ are the support vectors and labels, and $\alpha_i \ell_i$ are the dual coefficients. Likewise, the gradient of the RBF kernel can be analytically computed: 
$$
\nabla_ z d(\mathbf z) = -2 \gamma \sum_i \alpha_i \ell_i k(\mathbf z, \mathbf z_i) (\mathbf z - \mathbf z_i)
$$
The negative gradient of the decision function $-\nabla_z d$ points toward the nominal region of space: by actuating along this direction, we can nudge the system back into the nominal region. We define $\mathbf u_{z} = \nabla_z d(\mathbf z)$ and carry the sign into the final control law. We lift back to the original space: $\mathbf u_x = \cobrasPhi \mathbf u_z$.

To ensure that we barely affect the system, we take two additional steps to the control design.
First, to standardize and constrain the amount of influence on the system at each step, we force the controller to have a fixed norm:

$$
\tilde{\mathbf u}(\mathbf z) = \frac{\mathbf u_x (\mathbf{z})}{ ||\mathbf u_x(\mathbf{z})||}
$$

\noindent Second, we only apply control when we predict there will be an extreme event in the next $t_{pred}$ units, i.e. $d_{+}(\mathbf z) =\mathbf 1\{{d(\mathbf z) > 0}\}$. The final control law is then 
$$
\mathbf u(\mathbf z) = -k_{gain}d_{+}(\mathbf z) \tilde{\mathbf u}(\mathbf{z}) 
$$
\noindent where $k_{gain}$ is chosen to be 1\% of the mean norm of $\mathbf x$. This effectively creates a bang-bang controller with small constant energy input only when we predict the system could become an extreme event. 

Control is simulated using $\dot{\mathbf x} = \mathbf F(\mathbf x) + \mathbf u$ with zero-hold control and actuation frequency equal to $1$ time step.

\subsection{$N=101$ Fully Connected Network}
\label{si:fhn_101}
\begin{figure*}[t!]
    \centering
    \includegraphics[width=0.95\textwidth]{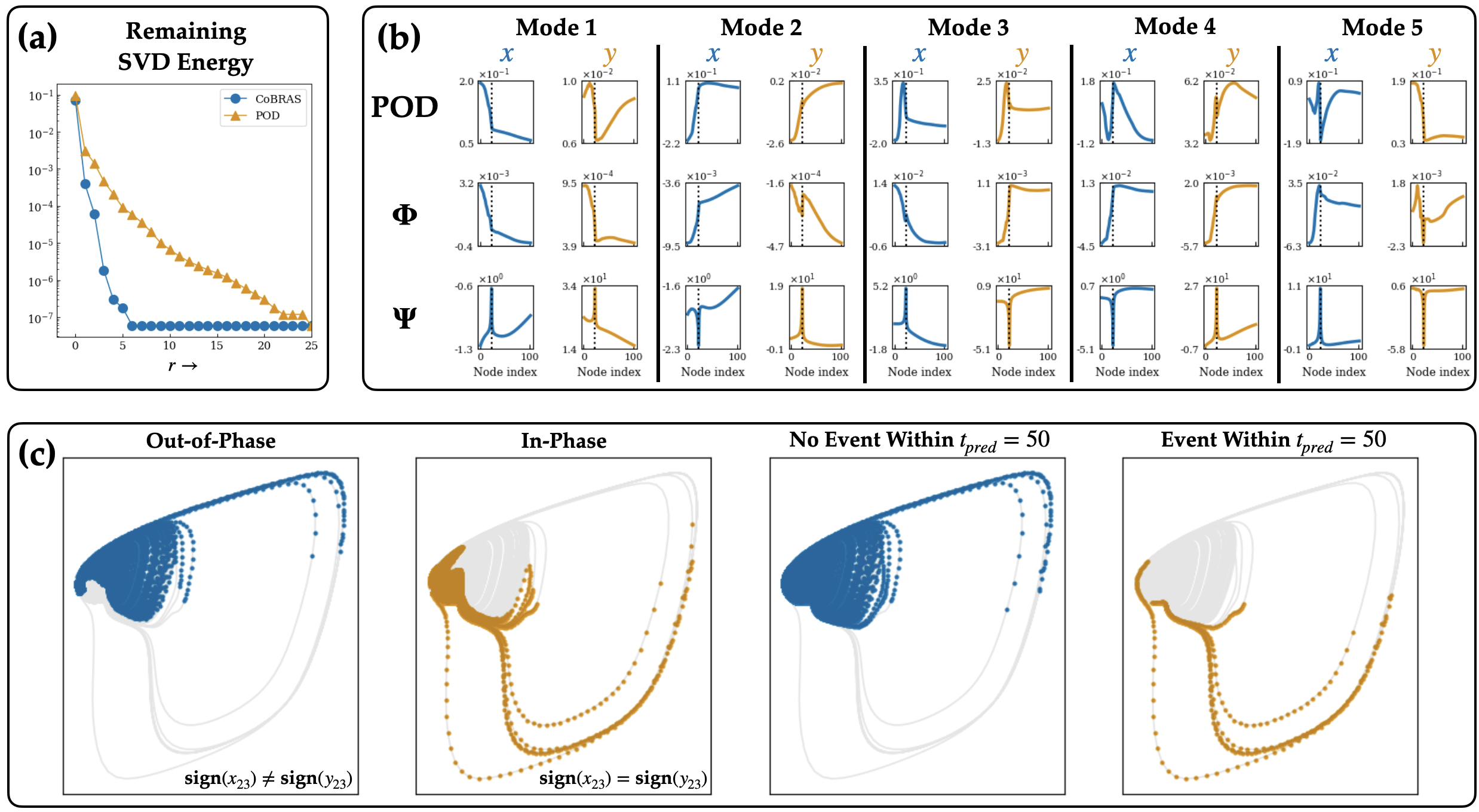}
    \cprotect \caption{\textbf{FHN POD and CoBRAS Modes $(N=101)$}.
    \textbf{(a)} The remaining energy for POD and CoBRAS SVDs. 
    \textbf{(b)} The first five  POD and CoBRAS $(\cobrasPhi, \cobrasPsi)$ modes with $x$- and $y$-components separated. Dashed line indicates the critical node \#23. 
    \textbf{(c)} Projection onto the leading two CoBRAS $\cobrasPsi$ modes comparing: the relative phase of $x_{23}$ and $y_{23}$ (\textit{left}) and whether an event occurs with $t_{pred}$ units (\textit{right}).
    }
     \label{si_fig:fhn_101-modes}  
\end{figure*}

For the $N=101$ case, we simulate the system for a total of $T=10^5$ time units, discarding the first $5000$ to ensure the state is on the attractor. To build the snapshot matrix $\mathbf X$ and train the SVM classifier, we use the first half of this trajectory and downsample by $10\times$. To calculate gradients,  we use an adjoint horizon of $T=100$ time units and form the matrix of gradients $\mathbf Y$ by sampling these snapshots used in $\mathbf X$ as initial conditions. To control the system, $k_{gain}$ is chosen as described in ~\ref{si:fhn} and  $k_{gain} \approx 6.76 \times 10^{-3}$.

We report the truncated energy for the POD and CoBRAS in \figurename ~\ref{si_fig:fhn_101-modes}(a). While both spectra decay rapidly, CoBRAS achieves machine precision after just 6 modes, while POD needs approximately 25. This indicates that the true rank of gradient-state inner product matrix $\mathbf{Y}^T \mathbf X$ is less than $6$; in fact, since the first two modes capture $99.9\%$ of the SVD energy, the effective rank is just $2$. 
In \figurename ~\ref{si_fig:fhn_101-modes}(b), we report the first five POD and CoBRAS $(\cobrasPhi, \cobrasPsi)$ modes. As observed in the main text, the criticality of node \#23 is evident across each of the $\cobrasPsi$ modes. 

Finally, in \figurename ~\ref{si_fig:fhn_101-modes}(c), we illustrate the importance of the relative phase of node \#23: when the state is in phase ($\text{sign}(x_{23}) = \text{sign}(y_{23})$, the system leads to excursions away from the body of the attractor; however, when the state is out-of-phase ($\text{sign}(x_{23}) \neq \text{sign}(y_{23})$), the system remains---or returns---to the body of the attractor.

\subsection{$N=1001$ Fully Connected Network}
\label{si:fhn_1001}

\begin{figure*}[t!]
    \centering
    \includegraphics[width=0.95\textwidth]{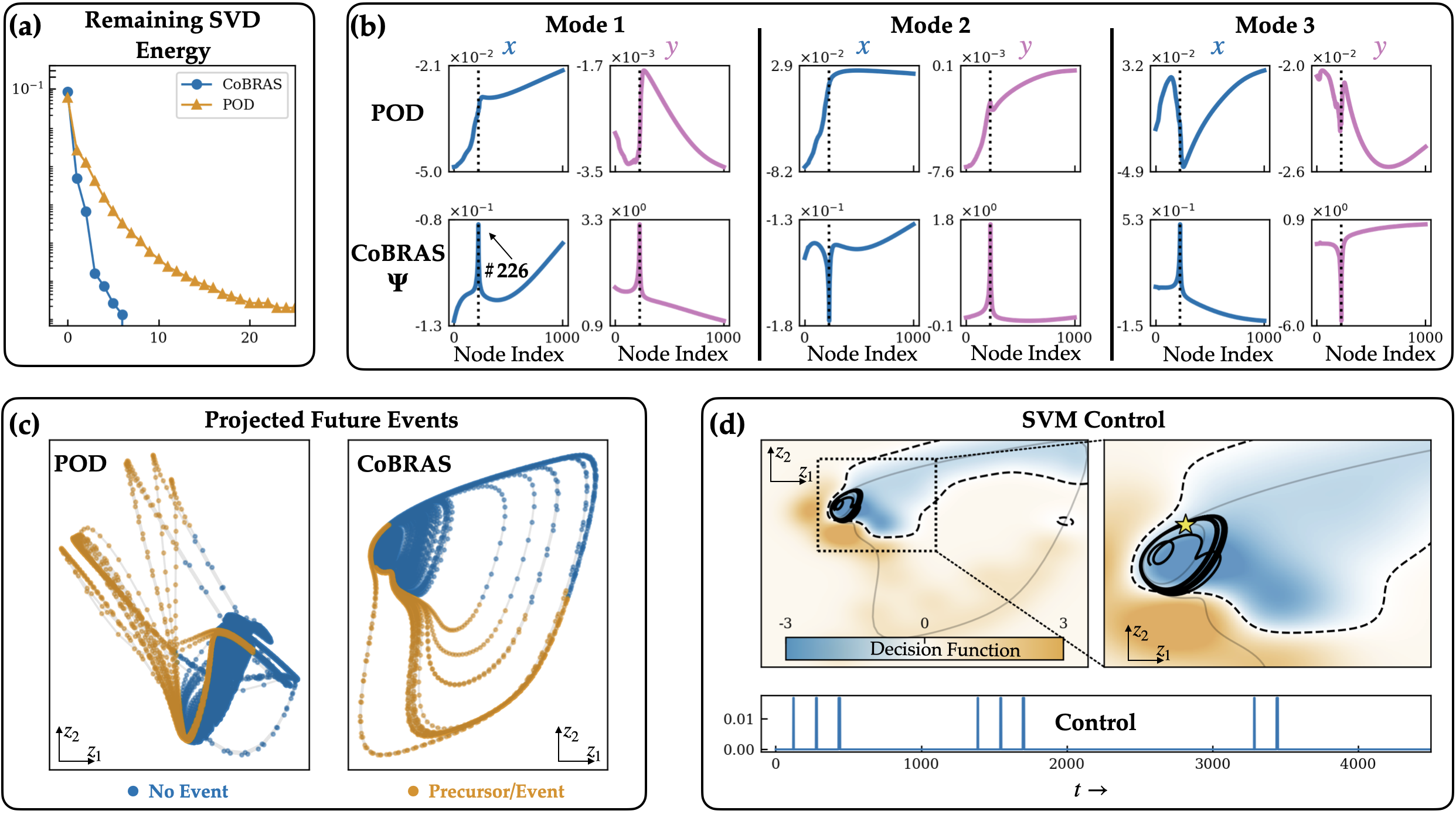}
    \cprotect \caption{\textbf{FitzHugh-Nagumo Oscillators $(N=1001)$} 
    \textbf{(a)} The remaining energy for the POD and CoBRAS SVDs. 
    \textbf{(b)} The first three POD and $\cobrasPsi$ modes separated into the corresponding action on the $x$ and $y$ components. Dashed vertical line indicates node \#226. 
    \textbf{(c)} The projection onto the leading two POD and CoBRAS modes with extreme events highlighted. 
    \textbf{(d)} The projection onto the leading two POD and CoBRAS modes with labels for classifying whether an extreme event will occur within the next $t_{pred}=50$ time units. Orange corresponds to a positive label, blue to a nominal label.
    \textbf{(e)} \textit{Top}: Heatmap of the SVM decision function classifying whether an event will occur within 50 time units and controlled trajectory preventing the event.  The black dashed line indicates the SVM decision boundary. The gray line is the reference uncontrolled trajectory, and the solid black line is the controlled trajectory. The yellow star in the zoomed in panel indicates the initial condition for both trajectories. \textit{Bottom}: The magnitude of the sparse control policy $||\mathbf u||$ as a function of time. 
    }
     \label{si_fig:fhn_1001}   
\end{figure*}

Just as in ~\ref{si:fhn_101}, we simulate the system for a total of $T=10^5$ time units. However, we discard the first $1.5 \times 10^4$ time units to ensure the state is on the attractor. To build the snapshot matrix and train the SVM classifier, $\mathbf X$, we use the first half of this trajectory and downsample by $10\times$. To calculate gradients,  we use an adjoint horizon of $T=100$ time units and form the matrix of gradients $\mathbf Y$ by sampling these snapshots used in $\mathbf X$ as initial conditions. To control the system, $k_{gain}$ is chosen as described in ~\ref{si:fhn} and  $k_{gain} \approx 1.68 \times 10^{-2}$.

In \figurename ~\ref{si_fig:fhn_1001}(a), we once again see that the first two CoBRAS modes account for $99.9\%$ of the SVD energy and the rest of the spectrum decays much faster than the POD spectrum. Just as in the main text, we plot the first three POD and $\cobrasPsi$ modes in \figurename ~\ref{si_fig:fhn_1001}(b). Analogous to the $N=101$ case, there is a critical node (or small, concentrated region of critical nodes) at \#226. This is consistent with \cite{ansmann2013a}, where no extreme events were found when there were fewer than $223$ excited nodes and almost always an extreme event when more than $224$  nodes were excited. 

\figurename ~\ref{si_fig:fhn_1001}(c-d) show the projection onto the first two POD and CoBRAS. Again, CoBRAS retains smoother geometry in the projection. While the regions between extreme events appear reasonably separable for POD in panel (c), it's clear that the precursors are not. However, just like in the $N=101$ case, the CoBRAS coordinates cleanly separate the precursors (with $t_{pred} = 50$) as living on the boundary of the body of the attractor. Using an unseen test set, the RBF kernel SVM achieved $96.4\%$ and $92.6\%$ success rate on the true nominal and true positive classes. In \figurename ~\ref{si_fig:fhn_1001}(e), we plot the SVM decision function which cleanly separates the nominal state space and demonstrate the effect of the controller gently nudging a trajectory back into the nominal class.

\clearpage
\subsection{$N=10000$ Small World Network}
\label{si:fhn_small_world}
To generate the small world network, we first generate the $60$-nearest neighbor graph. Then from every edge in this graph, there is a 20\% chance that it is replaced with a random edge by sampling two random nodes (resampling if the new edge conflicts with an existing edge). A schematic of a node's connectivity can be found in \figurename ~\ref{si_fig:fhn_small_world}   (a). This long range connectivity can facilitate rich dynamics, including chaotic bursting and spiral wave formation, as indicated in \figurename ~\ref{si_fig:fhn_small_world}(b). 

Due to memory constraints, we create a training set of two independently sampled trajectories, each consisting of $3\times 10^4$ time units. We remove the first $5000$ time units of each trajectory to ensure the data lives on the attractor and save every $10$-th time unit. These points are used to build the $\mathbf X$ matrix and serve as initial conditions for building the $\mathbf Y$ matrix. 

In \figurename ~\ref{si_fig:fhn_small_world}(c), we once again provide the remaining energy for the POD and CoBRAS SVDs. Unlike the previous fully-connected systems in ~\ref{si:fhn_101} and ~\ref{si:fhn_1001}, both CoBRAS and POD take many more modes in order to reach machine precision. However, the CoBRAS spectrum decays significantly faster; just two modes capture $99\%$ of the SVD energy. In \figurename ~\ref{si_fig:fhn_small_world}(d), we provide the three most dominant POD and $\cobrasPsi$ modes. While there is not a particular node that appears critical as in the previous examples, there is a consistent presence of voids in both the $x$- and $y$- components of the first two modes; indicating that there is a relevant spatial set of nodes. Interestingly, the third mode seems to capture the importance of phase information about a large, circular structure of nodes on the left. Notably, this structure appears in both the small and large bursting events shown in trajectories (i) and (iv) in panel (b). Because the network topology is randomly generated and results may vary across instantiations, we do not investigate the generality of these observations.

In \figurename ~\ref{si_fig:fhn_small_world}(e)-(f), we show the projection onto the the dominant POD and CoBRAS modes. The structures once again strongly resemble the attractors we saw with the fully-connected networks, CoBRAS retains smoother geometric structure. In particular, CoBRAS separates regions based on the quantity $q$, as seen in the trajectories in panel (e); trajectories with larger energy get further from the body of the attractor. In contrast, POD twists these trajectories together, and makes it difficult to separate extreme events, as shown in panel (f).

Unlike the fully-connected network, the small-world network exhibits spiral wave formation, where the $q$ values live on the border of the threshold; this makes the classification much more difficult, with the SVM classifier only achieving $84.3\%$ true nominal and $95.7\%$ true extreme events. The decision function and boundary are shown in \figurename ~\ref{si_fig:fhn_small_world}(g). However, the control law using the SVM remains sufficient to prevent the onset of extreme events.

\begin{figure*}[t!]
    \centering
    \includegraphics[width=0.95\textwidth]{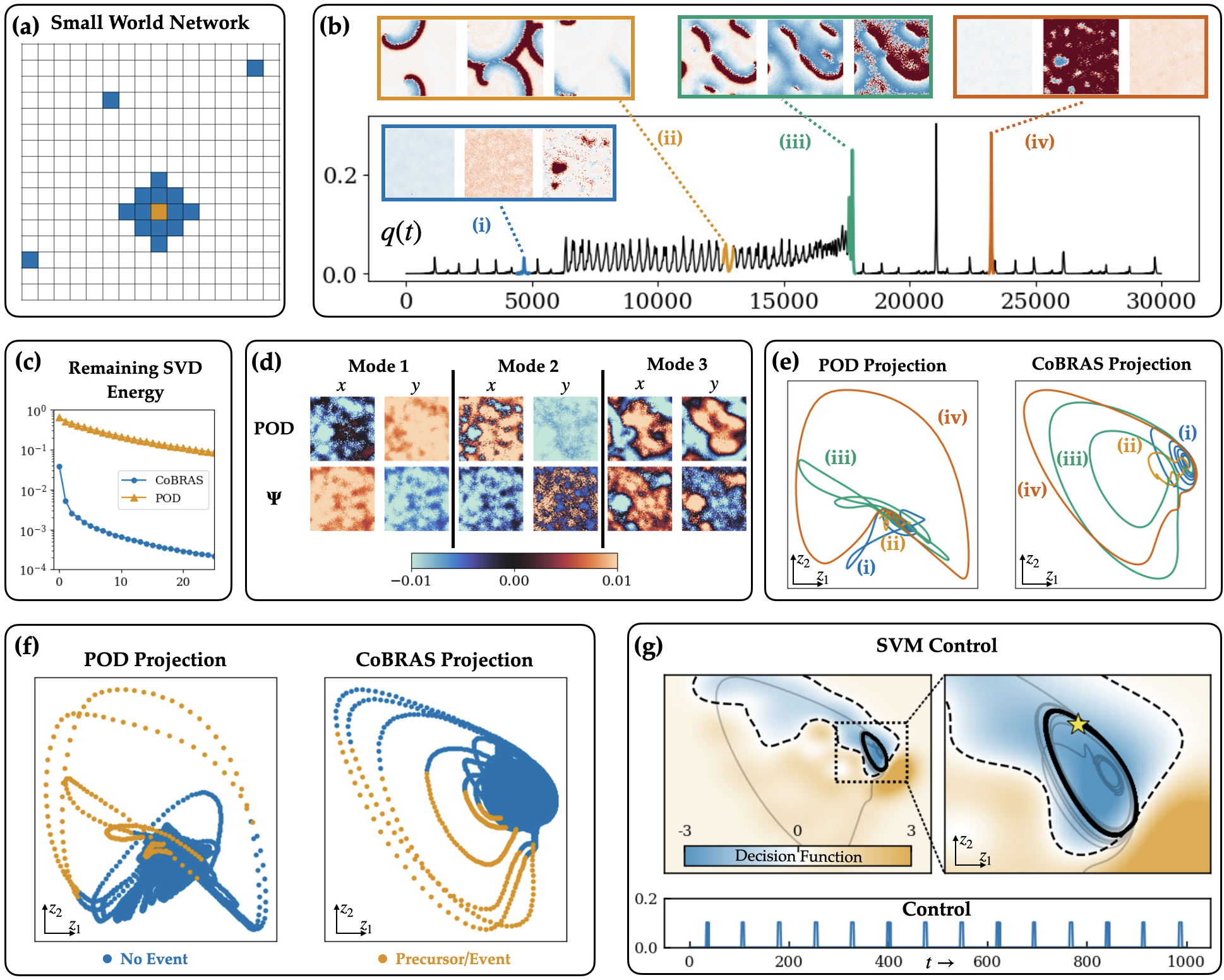}
    \cprotect \caption{\textbf{FitzHugh-Nagumo Small-World Network $(N=10^4)$} 
    \textbf{(a)} Schematic of small-world connectivity. The orange node has blue neighboring nodes on the lattice, with nearest neighbor and random long-range connections.  
    \textbf{(b)} Example of $q(t)$ and different types of behaviors with accompanying $x(t)$ snapshots: (i) small bursts with localized synchrony, (ii) spiral wave formation, (iii) spontaneous synchronization transition back to small bursts, (iv) large spontaneous synchronization.
    \textbf{(c)} The remaining energy for the CoBRAS and POD SVDs. 
    \textbf{(d)} The POD and CoBRAS $\cobrasPsi$ modes broken into $x$- and $y$-components. 
    \textbf{(e)} The same four trajectories corresponding to (a) projected onto the dominant POD and CoBRAS spaces.
    \textbf{(f)} The projection onto the dominant POD and CoBRAS spaces with points highlighted in orange if an event will happen within $t_{pred}$ time units.
    \textbf{(g)} \textit{Top}: Heatmap of the SVM decision function classifying whether an event will occur within 50 time units and controlled trajectory preventing the event.  The black dashed line indicates the SVM decision boundary. The gray line is the reference uncontrolled trajectory, and the solid black line is the controlled trajectory. The yellow star in the zoomed in panel indicates the initial condition for both trajectories. \textit{Bottom}: The magnitude of the sparse control policy $||\mathbf u||$ as a function of time. 
    }
     \label{si_fig:fhn_small_world}   
\end{figure*}

\clearpage
\section{Modified Nonlinear Schr{\"o}dinger (MNLS)}
\label{si:mnls}

Following the treatment by Cousins and Sapsis in \cite{cousins2016a}, we examine the form of the 1D Modified Nonlinear Schr{\"o}dinger (MNLS) Equation given by:

\begin{equation}
    \partial_{t}\mnls
    + \frac{1}{2}\partial_\mnlsSpatial \mnls
    + \frac{i}{8}\partial^2_{\mnlsSpatial}\mnls 
    - \frac{1}{16}\partial^3_{\mnlsSpatial}\mnls
    + \frac{i}{2}|\mnls|^2 \mnls 
    + \frac{3}{2} |\mnls|^2 \partial_\mnlsSpatial\mnls
    + \frac{1}{4} \mnls^2 \partial_\mnlsSpatial\mnls^*
    + i \mnls \mnlsPotentialGradient(\mnls) 
    =0
\end{equation}

\noindent where $\mnls(t,\mnlsSpatial)$ is a complex wave envelope, $\mnlsSpatial \in [0, 256\pi]$ is the spatial parameter, 
and
$\mnlsPotentialGradient$  corresponds to the derivative of the velocity potential and whose Fourier transform satisfies: 
$\mathcal F \left(\mnlsPotentialGradient(\mnls)\right)
    = -|k| \mathcal F \left( |\mnls|^2 \right)$, where $k$ is the wavenumber in the Fourier domain. 
    
Just as in \cite{cousins2016a}, we implement a pseudospectral solver with de-aliasing\footnote{While the cubic nonlinearities in MNLS only require retaining 1/2 of the wavenumbers, the dispersive terms cause energy growth in high-frequency terms that lead to rapid oscillations throughout the entire domain. To maintain numerical stability over longer periods of time, we retain only 1/3 of the wavenumbers.} and integrate the system using an explicit 4th-order exponential integrator with the method developed by Kassam \& Trefethen in \cite{kassam2005fourth} to avoid the numerical cancellations in the exponential terms.
We use a $dt_{sim}= 0.025$ and discretize the domain using $2^{10}= 1024$ grid points. While integration is performed using complex variables with double precision, we treat the solution variable $\mnls(t, \mnlsSpatial)$ as a two-dimensional real variable 
$\mnls = [\text{Re}(\mnls), \text{Im}(\mnls)]$; when it is discretized, we denote the variable as
$\mathbf{\mnls} \in \mathbb{R}^{2048}.$\footnote{In particular, we do not compute complex derivatives, as the dynamics are not holomorphic due to the presence of the conjugate $\mnls^*$} An event at $(t, \mnlsSpatial)$ is considered extreme if $|\mnls(t, \mnlsSpatial)|$ exceeds $0.2$; i.e. approximately four standard deviations above the mean.

Given a point in the spatial domain, $\mnlsSpatial_0$, we define the \textit{localized} quantity of interest at time $t$, to be the localized energy functional from Eq. \ref{mnlsQOI}, \noindent where $\mnlsGauss(\mnlsSpatial;\mnlsSpatial_0) = \frac{1}{L\sqrt{2\pi}}\exp\left(-\frac{(\mnlsSpatial - \mnlsSpatial_0 - ct)^2}{2L^2}\right)$ is a normalized Gaussian bump centered $\mnlsSpatial_0 - ct$. Here, we choose $c=0.5$ to be the average linearized group velocity. Essentially, $q$ measures the localized energy of the solution in a naturally moving reference frame.
We define the relevant forward map $\cobrasOutput$ as tracking this energy over time in the given frame:
\begin{equation}
    \cobrasOutput(\mnls; \mnlsSpatial_0) = [q_0, q_1, \dots q_{n-1}]  
\end{equation}

\noindent  where 
$q_k = q(\mnls; t+k\Delta t, \mnlsSpatial_0)$.
For simplicity, we choose $\Delta t = 1$, and we consider $n = 200$. Following \cite{cousins2016a}, we sample initial conditions for the flow according to a Gaussian spectrum of random phases:

$$
\mnls(0,\mnlsSpatial) 
= \sum_{k = -N/2 + 1}^{N/2} \sqrt{2 \Delta_k \rho(k\Delta_k)} e^{i( \omega_k \mnlsSpatial + \xi_k)}
, \qquad \rho(k) = \frac{\epsilon^2}{\sigma \sqrt{2\pi}}\exp(-k^2/2\sigma^2)
$$

\noindent where $\epsilon = 0.05$, $\sigma = 0.1$, and $\xi_k \sim  \text{Uniform}\left([0, 2\pi)\right)$.

To form a forward dataset, we sample $n_x = 5000$ initial conditions and  integrate them for $t\in [0,400]$. To assemble our snapshot and gradient matrices, we uniformly sample snapshots among our trajectories $\mathbf \mnls_i = \mnls(t_i,\mathbf \mnlsSpatial)$ (having restricted $t_{i} \in [0,200]$) to form our snapshot matrix $\mathbf X = [\mathbf\mnls_1, \dots, \mathbf \mnls_{n_x} ]$. We randomly sample $\mnlsSpatial_{i}$ uniformly in $[0, 256\pi)$ for computing our localized forward map $\mathbf \cobrasOutput(\mnls_i; \mnlsSpatial_{i})$ and taking gradients with respect to the initial snapshot $\mathbf \mnls_i$.

\subsection{CoBRAS Modes}
\label{si:mnls_modes}

\begin{figure*}[t!]
    \centering
    \includegraphics[width=0.95\textwidth]{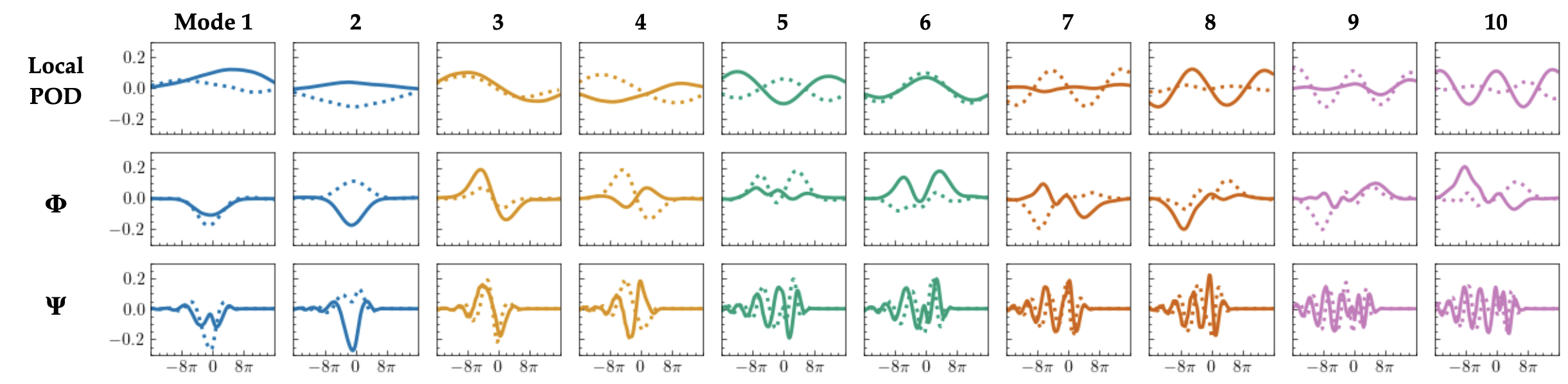}
    \cprotect \caption{\textbf{MNLS Modes} The first 10 local POD and CoBRAS $(\cobrasPhi, \cobrasPsi)$ modes. Solid and dotted lines correspond to the corresponding real and imaginary parts of the modes. Colors indicate consecutive pairs of similar CoBRAS $(\phi_i, \psi_i)$ modes that appear to differ only by a factor of $\pm i$. 
    }
     \label{si_fig:mnls_modes}   
\end{figure*}

In \figurename ~\ref{si_fig:mnls_modes}, we provide the first 10 local POD and CoBRAS $(\cobrasPhi, \cobrasPsi)$ modes on the range $[-16 \pi, 16\pi]$. Local POD modes were computed using the same snapshots to form the CoBRAS $\mathbf X$ matrix, but only using the middle 1/8 of the domain (of length $32\pi$). Because the statistics of the flow are translationally invariant, we observe that the local POD modes resemble Fourier modes confined to the interval $[-16\pi, 16\pi]$. However, the CoBRAS modes lack symmetry and are approximately compactly supported---capturing the most sensitive information of the system.  

\subsubsection{Sensitivity to Gaussian Length}
\label{si:mnls_gauss}

\begin{figure*}[t!]
    \centering
    \includegraphics[width=0.95\textwidth]{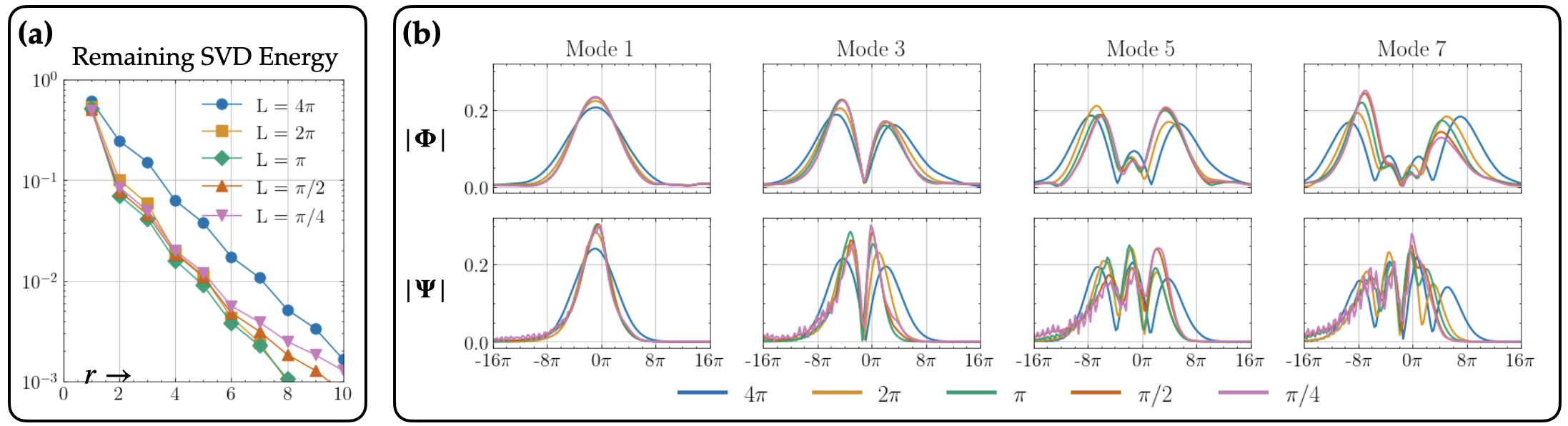}
    \cprotect \caption{\textbf{Effect of Gaussian Width}
    \textbf{(a)} The remaining energy for the CoBRAS SVDs across Gaussian width parameter, $L$.
    \textbf{(b)} The magnitudes $|\cobrasPhi|$ and $|\cobrasPsi|$ for modes 1,3,5,and 7 across Gaussian widths $L$. Modes 2,4,6 are excluded since they have nearly identical magnitudes due to the $\pm i$ phase offset.
    }
     \label{si_fig:mnls_gauss}   
\end{figure*}

In this work, we presented a method for using CoBRAS with localized phenomena by defining localized quantities of interest $q(\mnls; \eta_0)$ at the point $\eta_0$. For the MNLS, we choose to define this quantity as the energy over a localized measure
$$
q(\mnls; \eta_0) = \int |\mnls(t, \mnlsSpatial)|^2 \mnlsGauss(\mnlsSpatial; \mnlsSpatial_0)d\mnlsSpatial
$$
\noindent When restricting to the Dirac delta $\mnlsGauss(\mnlsSpatial, \mnlsSpatial_0) = \delta(\mnlsSpatial - \mnlsSpatial_0)$, we simply get the pointwise energy: $q(\mnls; \mnlsSpatial_0) = |\mnls(t, \mnlsSpatial_0)|^2$. However, because MNLS has strong dispersive dynamics, obtaining gradients amounts to integrating through the dispersion, leading to rapid spatial oscillations. Therefore for this example we defined $\mnlsGauss$ to be a small Gaussian: $\mnlsGauss(\mnlsSpatial;\mnlsSpatial_0) = \frac{1}{L\sqrt{2\pi}}\exp\left(-\frac{(\mnlsSpatial - \mnlsSpatial_0 - ct)^2}{2L^2}\right)$. In this section, we investigate the dependence of the Gaussian width $L$ on the modes.  In particular, we compare widths $L \in \{\pi/4, \pi/2, \pi, 2\pi, 4\pi\}$. Note that for a discretization of $N=1024$ on a domain of size $256\pi$, we have $dx= \pi/4$; hence this is the smallest natural scale we consider. 

As shown in \figurename ~\ref{si_fig:mnls_gauss}(a), the spectra and modes are in very close agreement for $L \leq 2\pi$. In \figurename ~\ref{si_fig:mnls_gauss}(b), we compare the modes across values of $L$. To avoid phase differences between the real and imaginary components, we plot the magnitude of the modes, i.e. $|\phi_k| = \sqrt{\text{Re}(\phi_k)^2 + \text{Im}(\phi_k)^2}$. We see there is much overlap in the modes across $L$ and they all qualitatively appear similar.  
However, $L=4\pi$ once again serves as the outlier; as $L$ increases, the QoI is now influenced by a larger domain---in particular, there is more sensitivity on the right hand side. We also note that the $\cobrasPsi$ modes have dispersive effects for $L = \pi/4 = dx$, which can be seen particularly in the higher-order modes. We therefore choose $L=\pi/2$, which avoids the dispersive effects and remains the most localized in space.  

\subsection{Forecasting Extreme Events}
\label{si:mnls_predictions}

\begin{figure*}[t!]
    \centering
    \includegraphics[width=0.95\textwidth]{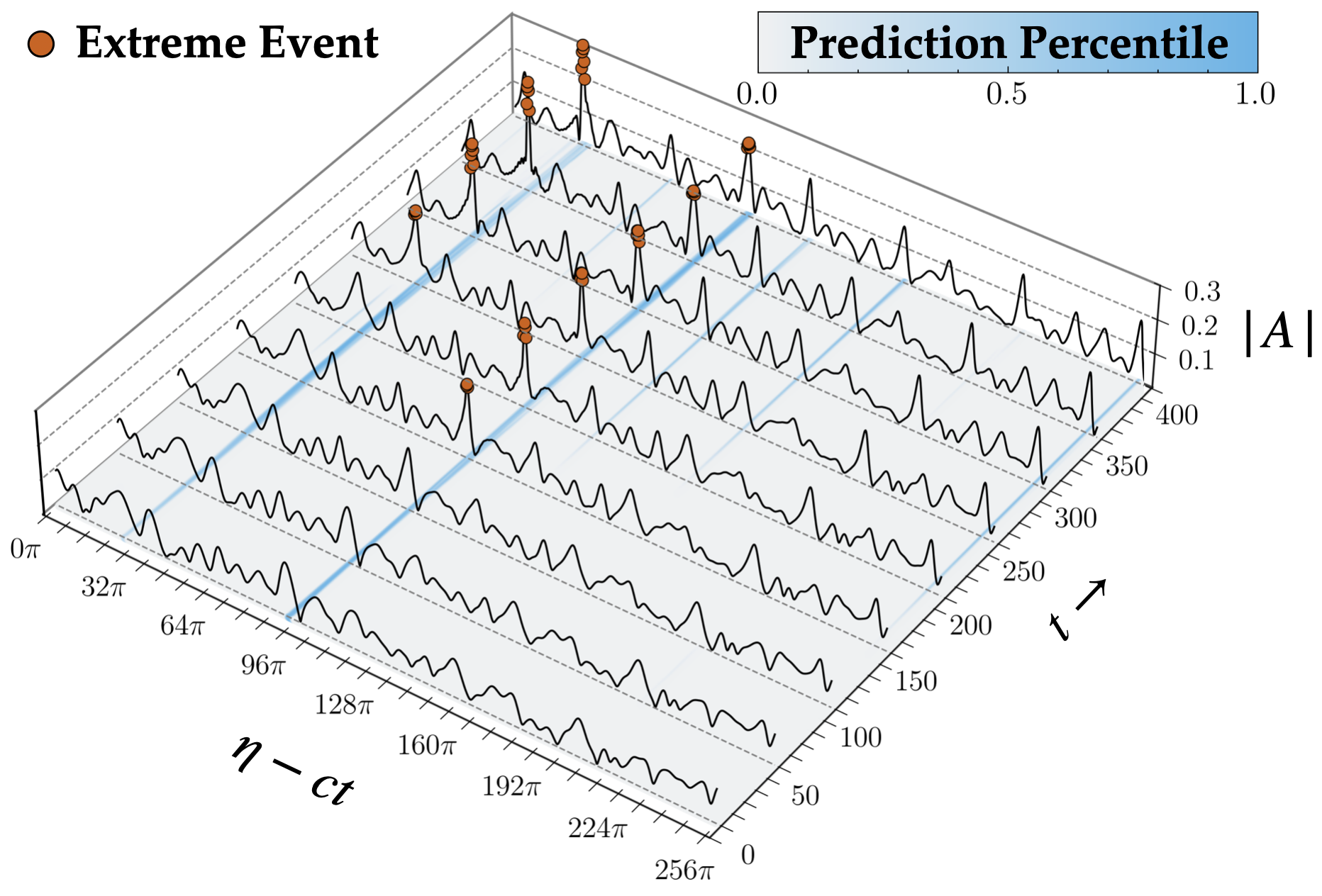}
    \cprotect \caption{\textbf{Localized Prediction of Rogue Waves.} Extreme event predictions obtained from SVM classifier to predict whether an extreme event will occur within 200 time units for the full $[0, 256\pi)$ domain. The blue streaks indicate the percentile rank of the SVM decision function among positive classifications, estimated from the training set. Red circles indicate the amplitude exceeding the extreme event threshold. 
    }
     \label{si_fig:mnls_preds} 
\end{figure*}

Just as in ~\ref{si:kflow_prediction} and ~\ref{si:fhn}, we use Scikit-Learn to train our classification models with an RBF kernel using features $\mathbf z = \cobrasPsi \mathbf x$ with balanced loss, scale factor $\gamma = 1/r$, and $C=1$. The CoBRAS signals $\mathbf z(\mnlsSpatial)$ are obtained by cross-correlating the $\cobrasPsi$ modes with the solutions $\mnls(t,\mnlsSpatial)$ using the cross-correlation formula $\mathbf z(\mnlsSpatial) = (\cobrasPsi \star \mnls)(\eta)$. SVM inputs are obtained by sampling these signals at different spatiotemporal points $(t_j, \mnlsSpatial_j)$ from these signals across five hundred independently generated solutions to MNLS (80\% are used for training, 20\% reserved for testing). Before training, inputs are centered and scaled to have unit variance for use with the classifier. For each point, $\mnlsSpatial$, we attempt to predict whether an extreme event will occur over $t_{pred}=200$ time units in the frame of reference moving through the domain at speed $c=0.5$ (the average group velocity). That is, we train the classifier to satisfy: 
$$
\text{sign} \left(   d(\mathbf z(\mnlsSpatial) \right)= \text{sign}\left\{\max_{\tau \in [t, t+t_{pred}]} |\mnls(\tau, \mnlsSpatial - c \tau)| -q_*   \right\}
$$
The decision function $d$ defines a \textit{nonlinear filter} for the system,
$d(\mnlsSpatial) = d\left[(\cobrasPsi \star \mnls)(\mnlsSpatial)\right]$. For this demonstration, we use $r=8$ modes to train the classifier. In \figurename ~\ref{si_fig:mnls_preds}, we plot the predictions for the full spatiotemporal domain, $\eta-ct \in [0,256\pi]$,  $t \in [0,400]$. Just as in the main text, blue streaks indicate positive predictions from the SVM. The magnitude is the calculated percentile from positive classes (i.e. a value of $0.0$ is close to the decision boundary, $0.5$ corresponds to the median positive value from the training data, and a value of $1.0$ is the \textit{maximal} decision output from the training data). For the two extreme events that form in the domain, the prediction scheme indicates with high confidence that they will form hundreds of time units before their true formation. While there are false positives throughout the domain, they all capture times where there is wave growth that do not result in a matured rogue wave---indicating that the CoBRAS modes are still finding spatiotemporal regions that are \textit{sensitive} to the QoI. 

\subsection{Suppressing Extreme Events}
\label{si:mnls_control}
As in the previous systems, we design a controller to suppress extreme events by using features generated by the CoBRAS modes. In particular, we use the SVM decision function in the latent CoBRAS space to determine when and how  we should actuate the domain. 

As before, let $\mathbf z(\mnlsSpatial) = (\cobrasPsi \star \mnls)(\mnlsSpatial)$ and define $d(\mnlsSpatial) = d[\mathbf z(\mnlsSpatial)]$ to be the decision function evaluated at the point $\mnlsSpatial$, and denote $d_{+}(\mnlsSpatial) =\mathbf 1\{{d(\mathbf z(\mnlsSpatial))} > 0\}$ as the function that is $+1$ when predicting an extreme event and $0$ otherwise. As described previously, the gradient of an RBF SVM decision function can be computed analytically at each point $\nabla_z d(\mnlsSpatial)$. We compute the term 
$$
\mathbf u_z(\mnlsSpatial) = 
    -k_{gain} d_{+}(\mnlsSpatial)
    \frac{(\nabla_z d)(\mnlsSpatial)}{||\nabla_z d(\mnlsSpatial)||}
$$
\noindent When $k_{gain} >0$, this points towards the nominal class, only when the decision function is predicting the formation of an extreme event. In this work, we choose $k_{gain} = 0.01$  to have minimal influence on the system. As shown in the main text, we lift back to the original state using $\mathbf u = \cobrasPhi(0)\cdot \mathbf u_z(\mnlsSpatial)$. 
To actuate the system, we use zero-hold control for $T=1$ units\footnote{Note that while both the CoBRAS \textit{modes} and prediction task are designed in the $c$-moving frame of reference, the actual prediction and control are performed pointwise and agnostic of the frame. In this work, we apply zero-hold control over a small time horizon, but one could in principle design the control to shift smoothly with the group velocity as well.}.

In \figurename ~\ref{si_fig:mnls_ctrl}, we provide the magnitude of the controller shown in the main text. As designed, the control is sparse throughout the spatiotemporal domain with a main emphasis on the locations that would have developed into rogue waves, as in \figurename ~\ref{si_fig:mnls_preds}. A small amount of control is constantly used to prevent the event from occurring at $\mnlsSpatial - ct = 32\pi$, while control is only applied at $96\pi$ and $128\pi$ until the SVM predicts there will no longer be an extreme event.

\begin{figure*}[t!]
    \centering
    \includegraphics[width=0.95\textwidth]{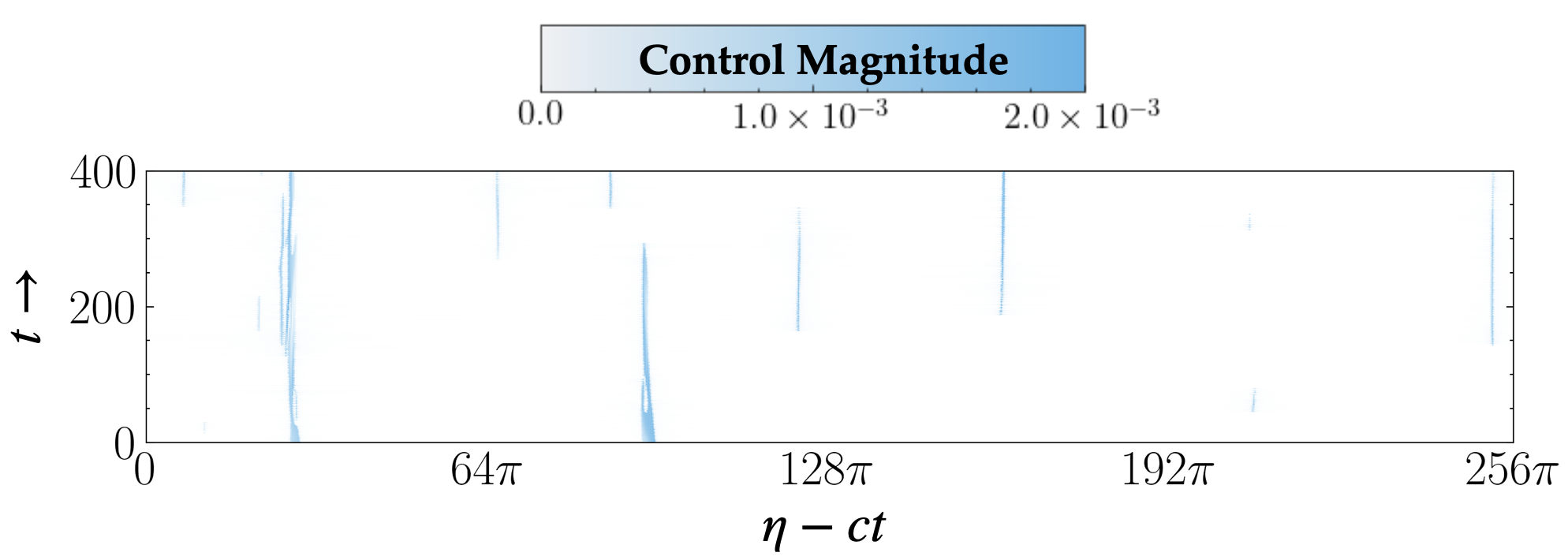}
    \cprotect \caption{\textbf{Sparse Rogue Wave Suppression} Heatmap of the magnitude of the control inputs across the entire spatiotemporal domain. 
    }
     \label{si_fig:mnls_ctrl} 
\end{figure*}

\clearpage
\section{FNO Surrogate Model}
\label{si:fno}
A central challenge to using CoBRAS is obtaining the sensitivity information. In the original development of CoBRAS \cite{otto2023model}, gradients were obtained by sampling the adjoint equations from a known forward set of equation. In this work, we used autodifferentiable solvers built using JAX \cite{jax2018github}. However, we often do not have the privilege of accessing these gradients. It is much more common to only have access to forward solutions in the form of a dataset---either collected from physical sensors or produced via numerical simulations. To address this challenge, we demonstrate the ability to approximate the gradient information by using autodifferentiation through a learned \textit{surrogate solver}. In particular, we choose to leverage a Fourier Neural Operator (FNO) \cite{li2021fourier}, whose structure naturally provides low-rank derivatives through the filtering layers. 

We use the operator to fit the forward solution $\mathbf x_{k + 1} = \mathbf F(\mathbf x_k)$, where $\mathbf{x}_{k} = \mathbf{x}(k\Delta t)$ is the solution at the $k$-th time step. While longer rollouts could be considered, we simply choose a single-step $L^2$ loss function $\mathcal L(\mathbf x_k) =||\mathbf{x}_{k+1} - \mathbf F(\mathbf{x}_k)||^2$. Following \cite{li2021fourier}, we use four FNO layers with $16$ truncated modes per layer and a width of $32$. Following \cite{guibas2021adaptive}, we use the Adam optimizer with a warmup and cosine-decay learning rate schedule. The Adam optimizer uses a $10^{-4}$ weight decay,  the scheduler has a peak learning rate of $10^{-3}$ and minimum of $10^{-5}$, and we clip gradient norms at $1.0$. Finally, we fix the batch-size equal to $32$.

We demonstrate the method on the Kolmogorov flow at $Re=40$ and define the vorticity to be the state variable $\kFlowVort = \mathbf{x}$. We define $\tilde \cobrasOutput$ to be the surrogate evolution of the energy dissipation:
$$
\tilde \cobrasOutput(\mathbf x_0) = \left[\energyDissipation[
    \mathbf F (\mathbf{x}_0)], 
    \energyDissipation[\mathbf F^{(2)} (\mathbf{x}_0)], 
    \dots,
    \energyDissipation[\mathbf F^{(n)} (\mathbf{x}_0)] 
    \right]
$$
\noindent where $\mathbf F^{(k)}(\mathbf x)$ denotes the $k$-fold composition of $\mathbf F$ acting on the initial condition $\mathbf x$. Analogous to ~\ref{si:kflow}, we set $\Delta t = 0.5$ time units and the adjoint horizon $T=4$ time units, for a total output dimension of $n=8$. To train the network, we use $64\times 64$ resolution images downsampled from the original $256\times 256$ dataset. However, for our demonstration, we still use the $256\times 256$ resolution dataset to build the snapshot matrix $\mathbf X$ and compute the surrogate gradients $\tilde{\mathbf g} = \nabla_\mathbf{x} \tilde (\boldsymbol{\xi^T}\cobrasOutput)(\mathbf x_0)$ to build the surrogate gradient matrix, $\tilde{\mathbf{Y}} = [\tilde{\mathbf g}_1, \dots, \tilde{\mathbf{g}}_{n_g}]$.  CoBRAS modes are obtained by analogously using the surrogate gradients, i.e. taking the SVD of $\tilde{\mathbf Y}^T \mathbf X$.

In \figurename ~\ref{si_fig:fno_modes}(a), we provide the remaining SVD energy from CoBRAS with the differentiable simulator (see ~\ref{si:kflow}) and a learned FNO simulator. The spectra are well-aligned for the first four modes before deviating slightly---however the tail ends also decay at similar rates. This is in agreement with the modes presented in \figurename ~\ref{si_fig:fno_modes}(b). The relevant first and fourth Fourier modes are captured within the dominant four-dimensional space, but the after the fifth mode, the $(\cobrasPhi, \cobrasPsi)$ FNO pairs begin to deviate. Notably, the FNO $\cobrasPsi$ modes seem to capture more low-energy structure---this is a well-known phenomenon in operator learning called ``spectral bias'', where models fail to capture high-frequency information. While unexplored in this work, there may be further opportunities to improve how well the dominant subspaces are captured by training on longer rollouts or incorporating additional regularization.

\begin{figure*}[t!]
    \centering
    \includegraphics[width=0.95\textwidth]{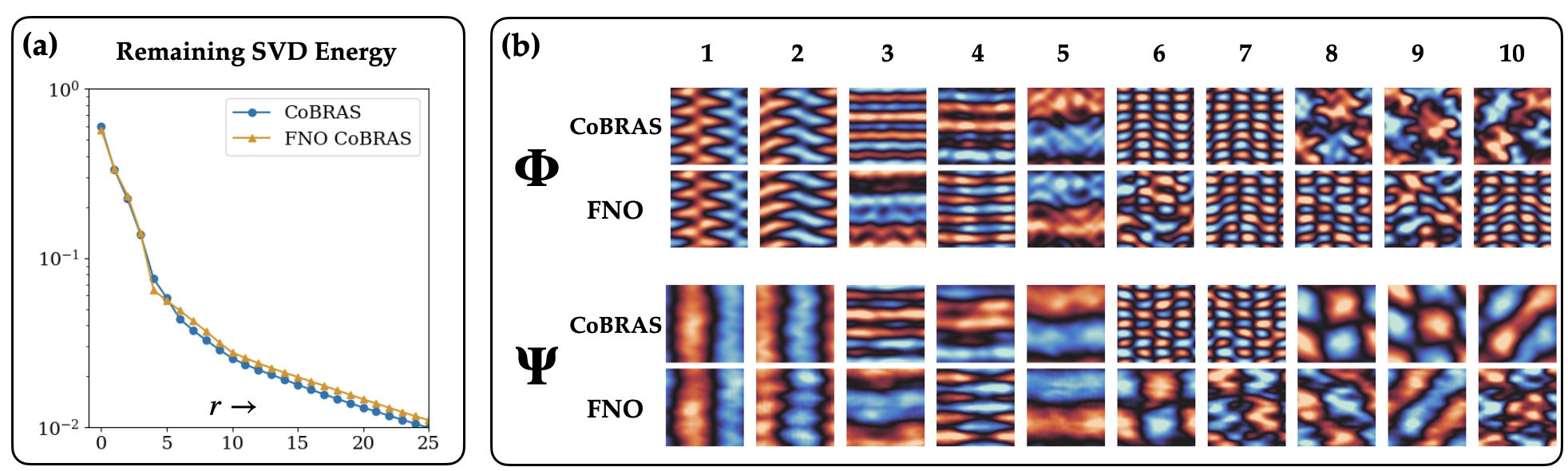}
    \cprotect \caption{\textbf{Surrogate CoBRAS Modes} 
    \textbf{(a)} The remaining energy for the  CoBRAS SVDs obtained from the CoBRAS differentiable simulator, ``CoBRAS'' (\ref{si:kflow}) or the learned surrogate model ``FNO CoBRAS''.
    \textbf{(b)} The first 10 $\cobrasPsi$ modes obtained through CoBRAS using the differentiable simulator vs the surrogate FNO.
    }
     \label{si_fig:fno_modes}   
\end{figure*}

\stopcontents[supp]

\end{document}